\documentclass[11pt]{article}

\usepackage[a4paper,margin=1in]{geometry}
\usepackage[T1]{fontenc}
\usepackage{lmodern}
\usepackage{microtype}
\usepackage{amsmath,amssymb,amsthm}
\usepackage{mathtools}
\usepackage{graphicx}
\usepackage{xcolor}
\usepackage[round,authoryear]{natbib}
\usepackage{url}
\usepackage[hidelinks]{hyperref}
\usepackage{authblk}

\theoremstyle{plain}
\newtheorem{theorem}{Theorem}
\newtheorem{lemma}[theorem]{Lemma}
\newtheorem{proposition}[theorem]{Proposition}
\newtheorem{corollary}[theorem]{Corollary}

\theoremstyle{definition}
\newtheorem{definition}[theorem]{Definition}

\theoremstyle{remark}

\newenvironment{proofof}[1]{\begin{proof}[Proof of #1]}{\end{proof}}

\newenvironment{keywords}{\medskip\noindent\textbf{Keywords:} }{\par\medskip}

\newcommand{\acks}[1]{\section*{Acknowledgments}#1}

\makeatletter
\newcommand*{\@jmlr@reflistsep}{, }
\newcommand*{\@jmlr@reflistlastsep}{ and }
\DeclareRobustCommand*{\objectref}[5]{%
  \let\@objectname\@empty
  \def\@objectref{}%
  \let\@prevsep\@empty
  \@for\@thislabel:=#1\do{%
    \toks@{\@prevsep}%
    \protected@edef\@objectref{\@objectref\the\toks@
      #4\ref{\@thislabel}#5}%
    \ifx\@objectname\@empty
      \let\@objectname#2%
    \else
      \let\@objectname#3%
      \let\@prevsep\@jmlr@reflistsep
    \fi
  }%
  \ifx\@objectname#3%
    \let\@prevsep\@jmlr@reflistlastsep
  \fi
  \@objectname~\@objectref
}
\makeatother

\newcommand*{\sectionrefname}{Section}
\newcommand*{\sectionsrefname}{Sections}
\newcommand*{\appendixrefname}{Appendix}
\newcommand*{\appendixsrefname}{Appendices}
\newcommand*{\equationrefname}{Equation}
\newcommand*{\equationsrefname}{Equations}
\newcommand*{\algorithmrefname}{Algorithm}
\newcommand*{\algorithmsrefname}{Algorithms}
\newcommand*{\theoremrefname}{Theorem}
\newcommand*{\theoremsrefname}{Theorems}
\newcommand*{\lemmarefname}{Lemma}
\newcommand*{\lemmasrefname}{Lemmas}
\newcommand*{\corollaryrefname}{Corollary}
\newcommand*{\corollarysrefname}{Corollaries}
\newcommand*{\definitionrefname}{Definition}
\newcommand*{\definitionsrefname}{Definitions}

\newcommand*{\sectionref}[1]{%
  \objectref{#1}{\sectionrefname}{\sectionsrefname}{}{}}
\newcommand*{\appendixref}[1]{%
  \objectref{#1}{\appendixrefname}{\appendixsrefname}{}{}}
\newcommand*{\equationref}[1]{%
  \objectref{#1}{\equationrefname}{\equationsrefname}()}
\newcommand*{\algorithmref}[1]{%
  \objectref{#1}{\algorithmrefname}{\algorithmsrefname}{}{}}
\newcommand*{\theoremref}[1]{%
  \objectref{#1}{\theoremrefname}{\theoremsrefname}{}{}}
\newcommand*{\lemmaref}[1]{%
  \objectref{#1}{\lemmarefname}{\lemmasrefname}{}{}}
\newcommand*{\corollaryref}[1]{%
  \objectref{#1}{\corollaryrefname}{\corollarysrefname}{}{}}
\newcommand*{\definitionref}[1]{%
  \objectref{#1}{\definitionrefname}{\definitionsrefname}{}{}}

\makeatletter
\newcounter{algorithm}
\renewcommand{\thealgorithm}{\@arabic\c@algorithm}
\newcommand{\algorithmname}{Algorithm}
\newcommand{\ext@algorithm}{loa}
\newcommand{\fps@algorithm}{htbp}
\newcommand{\ftype@algorithm}{8}
\newcommand{\fnum@algorithm}{\algorithmname~\thealgorithm}
\newenvironment{algorithm}[1][\fps@algorithm]%
{%
  \@float{algorithm}[#1]%
  \renewcommand{\@makecaption}[2]{%
    \parbox[t]{\linewidth}{{\bfseries ##1:} ##2}\par}%
}%
{\end@float}
\makeatother

\newcommand{\jmlralgorule}{\kern2pt\hrule height.8pt depth0pt\kern2pt}

\newcommand{\tableconts}[3]{%
  #2\label{#1}\vskip\baselineskip
  {\centering #3\par}%
}

\makeatletter
\newcommand{\floatconts}[3]{%
  \@ifundefined{\@captype conts}{\tableconts{#1}{#2}{#3}}%
  {\csname\@captype conts\endcsname{#1}{#2}{#3}}%
}
\makeatother

\newcommand{\abs}[1]{\left\lvert #1 \right\rvert}
\newcommand{\floor}[1]{\left\lfloor #1 \right\rfloor}
\newcommand{\ceil}[1]{\left\lceil #1 \right\rceil}

\newcommand{\val}{\mathrm{val}}
\newcommand{\ar}{\mathrm{ar}}
\newcommand{\zip}{\mathrm{zip}}
\newcommand{\ARC}{\mathrm{ARC}}
\newcommand{\sZLC}{\mathrm{sZLC}}
\newcommand{\LZLC}{\mathrm{LZLC}}
\newcommand{\AC}{\mathrm{AC}}
\newcommand{\poly}{\mathrm{poly}}
\newcommand{\Eval}{\mathrm{Eval}}
\newcommand{\PEval}{\mathrm{PEval}}
\newcommand{\vote}{\mathrm{vote}}
\newcommand{\Rt}{\mathrm{R}}
\newcommand{\Lt}{\mathrm{L}}

\DeclareMathOperator*{\argmax}{arg\,max}

\newcounter{alglineno}
\newlength{\algindentunit}
\makeatletter
\newenvironment{alglines}%
{%
  \setcounter{alglineno}{0}%
  \list{\stepcounter{alglineno}\thealglineno.}%
    {\setlength{\topsep}{0.3ex}%
     \setlength{\partopsep}{0pt}%
     \setlength{\itemsep}{0.2ex}%
     \setlength{\parsep}{0pt}%
     \setlength{\leftmargin}{2.2em}%
     \setlength{\labelwidth}{1.6em}%
     \setlength{\labelsep}{0.4em}%
     \setlength{\itemindent}{0pt}%
     \setlength{\listparindent}{0pt}%
     \renewcommand{\makelabel}[1]{\hbox to\labelwidth{\hfil##1}}}%
}%
{\endlist}
\makeatother
\newcommand{\algind}[1]{\hspace*{#1\algindentunit}\ignorespaces}

\newcommand*{\propositionrefname}{Proposition}
\newcommand*{\propositionsrefname}{Propositions}
\newcommand*{\propositionref}[1]{%
  \objectref{#1}{\propositionrefname}{\propositionsrefname}{}{}}

\newcommand{\eqv}{\mathrel{\approx}}
\newcommand{\neqv}{\mathrel{\not\approx}}

\title{Stringological sequence prediction III:\\
layered ziplines and a tradeoff between efficiency and expressivity}

\author[1,2]{Vanessa Kosoy\thanks{%
  \href{mailto:vanessa@alter.org.il}{\texttt{vanessa@alter.org.il}}}}
\affil[1]{Faculty of Mathematics, Technion, Haifa, Israel}
\affil[2]{Computational Rational Agents Laboratory, Delaware, USA}

\date{}

\begin{document}

\maketitle

\begin{abstract}%
  In previous papers \citep{Kosoy2026a,Kosoy2026b}, we began the study of
  sequence prediction algorithms adapted to stringological word complexity
  measures. In particular, we defined a complexity measure called Arithmetic
  Repetition Complexity (ARC) which admits a polynomial-time prediction
  algorithm with a mistake bound quasilinear in the complexity. Here, we show a
  weaker complexity measure related to ARC that admits an especially efficient
  prediction algorithm: an algorithm that runs in quasilinear time and polylog
  space for appropriate highly-structured sequences. The complexity measure is
  defined via a restricted class of ``zipline programs'' (a variant of
  straight-line programs we defined in \citep{Kosoy2026b}), which we call
  \emph{layered}. We thus get a less expressive measure with a more efficient
  algorithm (compared to our results for ARC), demonstrating a possible
  tradeoff.%
\end{abstract}


\section{Introduction}
\label{sec:intro}

Since this paper is part of a series, we split the introduction in two. In \sectionref{sec:intro-series} we introduce the series as a whole, and in \sectionref{sec:intro-installment} we introduce the results of this specific installment.

\subsection{Introduction to Entire Series}
\label{sec:intro-series}

Sequence (or ``time series'') prediction is a classical field of study in
artificial intelligence, machine learning and statistics
\citep{Hutter2024,Cesa-Bianchi2006,Brockwell2002}. Despite the long history of
the field, there are few examples of prediction algorithms which are
\emph{simultaneously} (i) computationally efficient, (ii) satisfy strong
provable mistake bounds, and (iii) guarantee asymptotically near-perfect
prediction for natural rich classes of deterministic sequences.

For any \emph{two} of these 3 natural requirements, there is ample prior work
that satisfies them.

Solomonoff induction \citep{Hutter2024} is the gold standard for sequence
prediction (assuming no domain-specific knowledge). On any sequence, it
asymptotically performs as well as any computable predictor, answering
requirements (ii) and (iii). It has strong conceptual justification as a
formalization of Occam's razor. However, Solomonoff induction is uncomputable:
it fails requirement (i) catastrophically, making it completely unsuitable for
any practical implementation.

Much work on statistics has focused on series of continuous variables (see
e.g.\ \citep{Brockwell2002}). This led to algorithms requiring assumptions
such as particular form of probability distributions, e.g.\ normal. While some
such algorithms answer requirements (i) and (ii), most of them are inapplicable
in the discrete setting, and they don't yield interesting classes of
deterministic sequences, failing requirement (iii).

In practice, methods based on deep learning are incredibly successful in
next-token prediction \citep{Brown2020}, answering requirement (i) and
suggesting that requirement (iii) is at least empirically satisfied. However,
the theoretical understanding of generalization bounds for deep learning is in
its infancy \citep{Fan2020}. In particular, examples where strong rigorous
mistake bounds can be proved are lacking, and requirement (ii) fails.

There are computationally efficient algorithms with strong provable
generalization bounds for some interesting classes of \emph{stochastic}
processes, for example context-tree weighting methods
\citep{Kontoyiannis2020}. However, such classes often degenerate in the
deterministic case, e.g.\ admitting only periodic sequences. Hence,
requirements (i) and (ii) are satisfied, but requirement (iii) fails.

Notably, the perceptron algorithm for online linear classification
\citep{Cesa-Bianchi2006} applied to features computed from the sequence does
come close to meeting our desiderata. However, the corresponding class of
predictable deterministic sequences is still quite limited.

In this series of papers, we present a novel family of sequence prediction
algorithms that are computationally efficient and satisfy provable mistake
bounds. The mistake bounds imply an asymptotically vanishing frequency of
mistakes for many interesting classes of sequences from combinatorics on words.
Thus, our algorithms satisfy (a form of) all 3 of the requirements above.

Formally, we work in the following setting. There is a fixed finite alphabet
$\Sigma$ and we are interested in predictors of the form $P : \Sigma^* \to
\Sigma$ (here, $P(u)$ is the symbol predicted to follow the prefix $u$). The
assumptions about the data are expressed as a \emph{complexity measure} $C :
\Sigma^* \to \mathbb{N}$. We are then interested in bounding the number of mistakes
$P$ makes on a sequence $u \in \Sigma^*$ in terms of $m \coloneqq C(u)$ and $n
\coloneqq \abs{u}$. At the same time, we require that $P$ is computable in
polynomial time. Moreover, we wish to simultaneously bound the size of the
\emph{internal state} of the predictor in terms of $m$ and $n$ (this can be
interpreted as the predictor \emph{compressing} the sequence). The latter is
interesting because it leads to predictors that are space efficient and in some
cases run in quasilinear\footnote{I.e.\ linear up to logarithmic factors.}
time.

We study multiple natural complexity measures $C$ with strong mistake and
compression bounds. (Several such measures were addressed in
\citep{Kosoy2026a,Kosoy2026b}; more are introduced here.) Conceptually, such
complexity measures can be viewed as candidate tractable analogues of
Kolmogorov complexity. We strive to simultaneously get a mistake bound of the
form $O(m \poly(\log n))$ and a compression bound of the form $\poly(m, \log
n)$. In particular, when operating on an infinite sequence for which $m$ grows
polylogarithmically, such a predictor runs in quasilinear time and polylog
space while making only polylog mistakes (see \definitionref{def_comp_eff}).

\subsection{Introduction to Current Installment}
\label{sec:intro-installment}

In \citep{Kosoy2026b} we defined the Arithmetic Repetition Complexity (ARC) of
a word $u \in \Sigma^*$ as the minimal size of an ``arithmetic repetition
system'': a system of equations that uniquely determines $u$, where each
equation either specifies the symbol $u[i] \in \Sigma$ at some position $i <
\abs{u}$ or mandates the equality of two subwords $u[\alpha]$ and $u[\beta]$
for some integer arithmetic sequences $\alpha : [k] \to \mathbb{N}$ and $\beta
: [k] \to \mathbb{N}$ such that $\alpha(i) < \beta(i)$ for all $i < k$. We
showed that there is a polynomial-time predictor that makes $O(\ARC(u) \cdot
\log \abs{u})$ mistakes on all words, but we had no meaningful bound on the
size of its state.

We also introduced ``zipline programs''. Analogously to how a straight-line
program is an acyclic context-free grammar whose productions rules are
deterministic, and which therefore defines a unique word
\citep{Burgisser1997}, a zipline program is the same combinatorial object where
the production rules are interpreted as ``zipping'' (i.e.\ round-robin
interleaving) words rather than concatenating them. As opposed to straight-line
programs, there is no ``Chomsky normal form'' for zipline programs: we can only
assume that all production rules are of \emph{prime} arity, not just binary. We
then showed how zipline programs give rise to arithmetic repetition systems of
comparable size. In particular, for what we call ``synchronous'' zipline
programs, the arithmetic repetition system is no larger than the program.

Here we consider a restricted class of zipline programs that we call ``layered
zipline programs''. Given a set $R$ of possible arities, we define the
$R$-adic layered zipline complexity $\LZLC_R(u)$ of a word $u \in \Sigma^*$ to
be the minimal ``width'' of a layered zipline program that produces an extension of $u$ s.t.\
all of its arities are in $R$. Crucially, we show that for any fixed $R$, there
exists a predictor (we call it $P^3$) for which both the mistakes and the state
size have bounds that are quasilinear in $\LZLC_R$ (although the state size
depends exponentially on $\abs{R}$). Thus, we can improve the efficiency of
prediction at the cost of reducing expressivity, by restricting from
synchronous zipline programs to layered zipline programs with fixed allowed
arities. Whether this trade-off is unavoidable or merely a limitation of our
techniques is left as an open problem.

\section{Setting}

In this section, we recall the framework of stringological sequence prediction
from \citep{Kosoy2026a}: the online prediction protocol, formulated for
state-based algorithms, and the criteria of statistical and computational
efficiency relative to a word complexity measure.

\subsection{Preliminaries and Notation}

Let $\Sigma$ be a fixed finite\footnote{We assume $\Sigma$ is finite purely for
ease of presentation. We could instead assume that $\Sigma = \mathbb{N}$, in
which case the factors of $\log \abs{\Sigma}$ in the bounds would be replaced
by $\log M$, where $M$ is the highest number that actually appeared in the
sequence so far.} alphabet. We denote by $\Sigma^*$ the set of finite words
over $\Sigma$ and by $\Sigma^\omega$ the set of right-infinite sequences. For a
word $u \in \Sigma^*$, $\abs{u}$ is its length and $u[i]$ is its $i$-th symbol
(indexing is 0-based); $u[i:j]$ denotes the factor (subword) $u[i] u[i+1] \dots
u[j-1]$, and $u[:j] \coloneqq u[0:j]$. A word $v$ is a prefix of $u$, denoted
$v \sqsubseteq u$, if $u = v w$ for some $w \in \Sigma^*$. The set of natural 
numbers $\mathbb{N}$ includes $0$. Logarithms are taken to base 2 unless 
otherwise specified, and over the natural numbers we use the convention 
$\log 0 := 0$. For $m \in \mathbb{N}$, $[m] \coloneqq 
\{0, 1, \dots, m - 1\}$.

\subsection{The Prediction Protocol}
\label{sec:protocol}

We operate in the standard deterministic online prediction setting. To discuss
memory and time constraints rigorously, we model the predictor as a state-based
machine: a \emph{predictor} is a tuple $\Pi = (\mathcal{S}, s_{\mathrm{init}},
\mathcal{U}, \mathcal{P})$, where $\mathcal{S} \subseteq \{0, 1\}^*$ is the set
of possible internal states (represented as binary strings),
$s_{\mathrm{init}} \in \mathcal{S}$ is the initial state, $\mathcal{U} :
\mathcal{S} \times \Sigma \to \mathcal{S}$ is the \textbf{state-update
function} and $\mathcal{P} : \mathcal{S} \to \Sigma$ is the
\textbf{state-prediction function}. On a target sequence $x \in \Sigma^T$,
the predictor starts at $s_0 \coloneqq s_{\mathrm{init}}$ and, at each round $t
= 0, 1, \dots, T-1$, outputs the hypothesis $\hat{x}_t \coloneqq \mathcal{P}(s_t)$,
observes the true symbol $x[t]$ and updates its state to $s_{t+1} \coloneqq
\mathcal{U}(s_t, x[t])$. We denote the
total number of mistakes and the maximal state size during the processing of
$x$ by
\[
  M_\Pi(x) \coloneqq \abs{\{ t < T \mid \hat{x}_t \neq x[t] \}}, \qquad
  S_\Pi(x) \coloneqq \max_{0 \leq t < T} \abs{s_t}
\]

For $T=0$, we set $S_\Pi(x):=0$ by convention.

\subsection{Efficiency Criteria}

We evaluate performance against the inherent structural complexity of the
individual sequence, rather than a probabilistic prior.

\begin{definition}
\label{def:word_complexity}
  A \emph{word complexity measure} is a function $C : \Sigma^* \to \mathbb{N}$
  that satisfies the following conditions:
  \begin{itemize}
    \item \textbf{(Polynomial bound.)} There exists a polynomial $p \in
      \mathbb{N}[x]$ s.t.\ for all $u \in \Sigma^*$,
      \begin{equation}
        \label{eq_word_com_bound}
        C(u) \leq p(\abs{u})
      \end{equation}
    \item \textbf{(Approximate monotonicity.)} There exists a polynomial $q \in
      \mathbb{N}[x, y]$ s.t.\ for all $u \in \Sigma^*$ and $k \leq \abs{u}$,
      \begin{equation}
        \label{eq:mono}
        C(u[:k]) \leq q(\log \abs{u}, \log C(u)) \cdot (C(u) + 1)
      \end{equation}
  \end{itemize}
\end{definition}

\begin{definition}[\textbf{Statistical Efficiency.}]
\label{def_stat_eff}
The predictor $\Pi$ is \emph{statistically efficient} with respect to $C$ if
the number of mistakes is quasilinear in the complexity: $M_\Pi(x) \leq (C(x) +
1) \cdot \poly(\log \abs{x})$ for all $x \in \Sigma^*$.
\end{definition}

\begin{definition}[\textbf{Computational Efficiency.}]
\label{def_comp_eff}
The predictor $\Pi$ is \emph{computationally feasible} if $\hat{x}_t =
\mathcal{P}(s_t)$ and $s_{t+1} = \mathcal{U}(s_t, x[t])$ are computable in time
$\poly(\abs{s_t})$, and $S_\Pi(x)$ is bounded by $\poly(\abs{x})$. The
predictor is \emph{computationally efficient} with respect to $C$ if it's
computationally feasible and, for all $x \in \Sigma^*$:
\begin{equation}
  \label{eq:gen_com_bound}
  S_\Pi(x) \leq \poly(C(x), \log \abs{x})
\end{equation}
\end{definition}

We can thus think of $s_t$ as a compressed representation (not necessarily
lossless) of the history $x[:t]$, and we refer to inequalities of the form
\equationref{eq:gen_com_bound} as ``compression bounds''. In particular, if we 
consider an infinite sequence $x\in\Sigma^\omega$ s.t. $C(x[:t])$ grows 
polylogarithmically, then the space complexity and the per-round processing 
time of a computationally efficienct predictor are also 
polylogarithmic\footnote{Since per-round processing time is polynomial in $\abs{s_t}$ and 
$\abs{s_t}$ is polylogarithmic in $t$ due to inequality 
\equationref{eq:gen_com_bound}.}.

\section{Layered Ziplines}
\label{sec:ziplines}

In this section, we introduce a restricted class of the zipline programs of
\citep{Kosoy2026b} --- the \emph{layered} programs over a fixed finite set of
arities $R$ --- and the associated complexity measure $\LZLC_R$, which subsumes
RTL automaticity in every base generated by $R$ (\corollaryref{cor:multibase})
while being itself subsumed, up to logarithmic factors, by the synchronous
zipline complexity $\sZLC$ (\corollaryref{cor:szlc}). Its raison d'\^etre is
the associated predictor $P^3$, whose mistake bound \emph{and} state-size bound
are both quasilinear in the complexity (\theoremref{thm:p3}). Throughout this
section and the next, $\sigma \coloneqq \abs{\Sigma}$; the proofs omitted from
this section are given in \appendixref{app:proofs}.

\subsection{Layered Zipline Programs}

We start by recalling the relevant definitions of \citep{Kosoy2026b}. Given $r
\geq 2$ and $u_0, u_1 \dots u_{r-1} \in \Sigma^*$, the word $z \coloneqq
\zip(u_0, \dots, u_{r-1})$ is obtained by interleaving $u_0 \dots u_{r-1}$ in a
round-robin fashion, stopping just before the first requested symbol that does
not exist: if $\mu \coloneqq \min_{i < r} \abs{u_i}$ and $j < r$ is minimal
with $\abs{u_j} = \mu$, then
\begin{equation}
  \label{eq:zip}
  \abs{z} \coloneqq r \mu + j \quad \text{and} \quad z[i] \coloneqq u_{i \bmod
  r}[\floor{i / r}] \text{ for } i < \abs{z}
\end{equation}

\begin{definition}[recalled]
\label{def:zlp}
A \emph{zipline program} (ZLP) over $\Sigma$ is $P = (Q, q_0, \delta)$, where
$Q$ is a finite set, $q_0 \in Q$ and $\delta : Q \to \overline{Q}^*$, where
$\overline{Q} \coloneqq Q \sqcup \Sigma$. We require that:
\begin{itemize}
  \item For all $q \in Q$, $\ar(q) \geq 2$, where $\ar(q) \coloneqq
    \abs{\delta(q)}$ is called the \emph{arity} of $q$.
  \item $\delta$ is acyclic, i.e.\ there are no $p_0, p_1 \dots p_{l-1} \in Q$
    s.t.\ for all $i < l$, $p_{(i+1) \bmod l}$ appears in $\delta(p_i)$.
  \item $q_0$ is the unique element of $Q$ that doesn't appear in $\delta(q)$
    for any $q \in Q$.
\end{itemize}
We define $\val_P : \overline{Q} \to \Sigma^*$ recursively by
\begin{itemize}
  \item For $a \in \Sigma$, $\val_P(a) \coloneqq a$.
  \item For $q \in Q$, $\val_P(q) \coloneqq \zip(\val_P(\delta(q)[0]), \dots,
    \val_P(\delta(q)[\ar(q) - 1]))$.
\end{itemize}
The \emph{value} of $P$ is $\val(P) \coloneqq \val_P(q_0)$ and the \emph{size}
of $P$ is $\abs{P} \coloneqq \sum_{q \in Q} \ar(q)$.
\end{definition}

Recall also that a ZLP $P$ is called \emph{synchronous} when there exists $\phi
: Q \to \mathbb{N}$ s.t.\ $\phi(q_0) = 1$ and, whenever $q'$ appears in
$\delta(q)$, $\phi(q') = \ar(q) \cdot \phi(q)$; and that the synchronous
zipline complexity $\sZLC(x)$ is the minimal size of a synchronous ZLP $P$
s.t.\ $x \sqsubseteq \val(P)$.

\begin{definition}
\label{def:layered}
A ZLP $H = (Q, q_0, \delta)$ is called \emph{layered} when there exists $h : Q
\to \mathbb{N}$ (a \emph{layering}) satisfying:
\begin{enumerate}
  \item $h(q_0) = 0$;
  \item for all $q, q' \in Q$, if $q'$ appears in $\delta(q)$ then $h(q') =
    h(q) + 1$;
  \item for all $q, q' \in Q$, if $h(q) = h(q')$ then $\ar(q) = \ar(q')$.
\end{enumerate}
\end{definition}

Since $q_0$ is the unique source of the (acyclic) graph of $H$, every $q \in Q$
is reachable from $q_0$, and condition 2 forces $h(q)$ to equal the length of
every path from $q_0$ to $q$. In particular, the layering is unique, and the
\emph{layers} $Q_j \coloneqq h^{-1}(j)$ are non-empty exactly for $j < d(H)$,
where $d(H) \coloneqq 1 + \max_{q \in Q} h(q)$ is the \emph{depth} of $H$. By
condition 3, we may define the \emph{arity profile} $\ar_H(j) \coloneqq \ar(q)$
for any $q \in Q_j$, and the associated \emph{scales}
\[
  m_H(j) \coloneqq \prod_{i < j} \ar_H(i), \quad 0 \leq j \leq d(H)
\]

Finally, given a finite set $R \subseteq \{2, 3, 4, \dots\}$, we say that $H$
is \emph{$R$-adic} when $\ar_H(j) \in R$ for all $j < d(H)$.

\begin{proposition}
\label{prop:sync}
Every layered ZLP $H$ is synchronous, with $\phi(q) = m_H(h(q))$.
\end{proposition}

\begin{proof}
We have $\phi(q_0) = m_H(0) = 1$, and if $q'$ appears in $\delta(q)$ then
$h(q') = h(q) + 1$, whence $\phi(q') = m_H(h(q) + 1) = \ar_H(h(q)) \cdot
m_H(h(q)) = \ar(q) \cdot \phi(q)$.
\end{proof}

By \equationref{eq:zip}, a layered ZLP computes its value in the manner of an
automaton reading the digits of the position index: writing $t < m_H(d(H))$ in
the mixed-radix numeration system with radices $\ar_H(0), \ar_H(1), \dots$,
i.e.\ $t = \sum_{j < d(H)} t_j m_H(j)$ with $t_j < \ar_H(j)$, and descending
from $q_0$ by choosing, at layer $j$, the child indexed by the digit $t_j$, one
arrives at the terminal $\val(H)[t]$. This is made precise in
\lemmaref{lem:descent} (\appendixref{app:proofs}). The digits are read least
significant first, which is the reading order of RTL automaticity, cf.\ Section
3 of \citep{Kosoy2026b}; below we make this connection precise.

\subsection{Layered Zipline Complexity}

\begin{definition}
\label{def:lzlc}
The \emph{width} of a layered ZLP $H$ is $w(H) \coloneqq \max_{j < d(H)}
\abs{Q_j}$. Given a finite non-empty set $R \subseteq \{2, 3, 4, \dots\}$ and $x \in
\Sigma^*$, the \emph{$R$-adic layered zipline complexity} of $x$ is
\[
  \LZLC_R(x) \coloneqq \min\{w(H) \mid H \text{ is an } R\text{-adic layered
  ZLP s.t.\ } x \sqsubseteq \val(H)\}
\]
\end{definition}

\begin{proposition}
\label{prop:lzlc-wcm}
For every finite non-empty $R \subseteq \{2, 3, 4, \dots\}$, the function $\LZLC_R$ is a
word complexity measure (\definitionref{def:word_complexity}). Moreover, for
all $x, y \in \Sigma^*$:
\begin{enumerate}
  \item $1 \leq \LZLC_R(x) \leq \max(\abs{x} - 1, 1)$;
  \item if $x \sqsubseteq y$ then $\LZLC_R(x) \leq \LZLC_R(y)$.
\end{enumerate}
\end{proposition}

The minimal-width programs for a given word can be badly behaved: values of
individual vertices may be truncated by shorter siblings (see
\equationref{eq:zip}), and the depth may be much larger than $\log \abs{x}$.
The following definition captures the absence of truncation.

\begin{definition}
\label{def:exact}
A layered ZLP $H$ is called \emph{exact} when $\abs{\val_H(q)} = m_H(d(H)) /
m_H(h(q))$ for every $q \in Q$. In particular, if $H$ is exact then
$\abs{\val(H)} = m_H(d(H))$.
\end{definition}

\begin{lemma}
\label{lem:exact}
A layered ZLP $H$ is exact if and only if no vertex outside the last layer
$Q_{d(H) - 1}$ has a terminal child.
\end{lemma}

\begin{proof}
Suppose no vertex outside the last layer has a terminal child. All children of
a vertex of the last layer are terminal (a non-terminal child would belong to
the empty layer $Q_{d(H)}$), so, by upward induction on layers using
\equationref{eq:zip}, the value of any $q \in Q_j$ is the zip of $\ar(q)$ words
of equal length $m_H(d(H)) / m_H(j + 1)$, whence $\abs{\val_H(q)} = m_H(d(H)) /
m_H(j)$. Conversely, if some $q \in Q_j$ with $j < d(H) - 1$ has a terminal
child then, by \equationref{eq:zip}, $\abs{\val_H(q)} \leq 2 \ar(q) - 1 <
\ar(q) \cdot m_H(d(H)) / m_H(j + 1) = m_H(d(H)) / m_H(j)$, where the strict
inequality holds because $m_H(d(H)) / m_H(j + 1) \geq 2$; so $H$ is not exact.
\end{proof}

The following lemma lets us assume exactness and logarithmic depth, at a small
cost in width.

\begin{lemma}[Regularization]
\label{lem:reg}
Let $n \geq 2$, $x \in \Sigma^n$, and let $H$ be an $R$-adic layered ZLP with
$x \sqsubseteq \val(H)$. Then there exists an $R$-adic layered ZLP $H'$, with
vertex set $Q'$ and layers $Q'_j$, s.t.:
\begin{enumerate}
  \item $x \sqsubseteq \val(H')$;
  \item $d(H') \leq \ceil{\log n}$ and $m_{H'}(d(H')) \geq n$;
  \item $w(H') \leq w(H) + \sigma$;
  \item $H'$ is exact; in particular, $\abs{\val(H')} = m_{H'}(d(H'))$.
\end{enumerate}
Moreover, if all arities of $H$ are equal to $k$ then so are those of $H'$, and
$d(H') = \ceil{\log_k n}$.
\end{lemma}

\begin{corollary}
\label{cor:szlc}
For every finite non-empty $R \subseteq \{2, 3, 4, \dots\}$ with $r \coloneqq \max R$,
every $n \geq 2$ and every $x \in \Sigma^n$:
\[
  \ARC(x) \leq \sZLC(x) \leq r \ceil{\log n} \cdot (\LZLC_R(x) + \sigma)
\]
\end{corollary}

\begin{proof}
Apply \lemmaref{lem:reg} to a width-optimal program for $x$. The resulting $H'$
satisfies $\abs{H'} = \sum_{j < d(H')} \abs{Q'_j} \cdot \ar_{H'}(j) \leq
\ceil{\log n} \cdot (\LZLC_R(x) + \sigma) \cdot r$, it is synchronous by
\propositionref{prop:sync}, and $x \sqsubseteq \val(H')$; hence the second
inequality. The first inequality is Proposition 4.9 of \citep{Kosoy2026b}.
\end{proof}

\subsection{Relation to RTL Automaticity}
\label{sec:rtl}

Next, we compare $\LZLC_R$ with RTL automaticity, showing that the two are
essentially equivalent (up to logarithmic factors) when $R$ consists of a
single base, while general $R$ yields a strictly more flexible measure. Recall
the setting of Section 3 of \citep{Kosoy2026b}: fix an integer base $k \geq 2$;
for $m \in \mathbb{N}$ and $0 \leq t < k^m$, let $\langle t \rangle_k^m \in
[k]^m$ denote the base-$k$ representation of $t$ padded with leading zeros to
length $m$, and let $(\cdot)^{\Rt}$ denote string reversal. A
\emph{Deterministic Finite Automaton with Output} (DFAO) with input alphabet
$[k]$ is $M = (Q_M, q_M, \delta_M, \tau_M)$, computing $M(u) \coloneqq
\tau_M(\delta_M^*(q_M, u))$ for $u \in [k]^*$. The \emph{RTL
$k$-automaticity}\footnote{The \emph{LTR $k$-automaticity} $\AC_k^{\Lt}(x)$ of
\citep{Kosoy2026a} is defined identically, except that the automaton reads the
digits most significant first: $(\langle t \rangle_k^m)^{\Rt}$ is replaced by
$\langle t \rangle_k^m$.} of $x \in \Sigma^*$ is
\[
  \AC_k^{\Rt}(x) \coloneqq \min\{\abs{Q_M} \mid \exists m \in \mathbb{N} :
  \abs{x} \leq k^m \text{ and } \forall t < \abs{x} : x[t] = M((\langle t
  \rangle_k^m)^{\Rt})\}
\]

\begin{lemma}[Radix refinement]
\label{lem:radix-refinement}
Let $k \geq 2$, write $k = \prod_{i < l} f_i$ with $l \geq 1$ and $f_i \geq 2$,
and let $M = (Q_M, q_M, \delta_M, \tau_M)$ be a DFAO over $[k]$. Suppose that
$m \geq 1$, $x \in \Sigma^n$, $n \leq k^m$, and
\[
  x[t] = M((\langle t \rangle_k^m)^{\Rt})
\]
for every $t < n$. Then there is an exact layered ZLP $H$ with $x \sqsubseteq
\val(H)$, $d(H) = l m$, and $\ar_H(j) = f_{j \bmod l}$ for every $j < l m$,
whose width satisfies
\[
  w(H) \leq \abs{Q_M} \max_{j < l} \prod_{i < j} f_i.
\]
\end{lemma}

\begin{proposition}
\label{prop:rtl}
Let $k \geq 2$, $n \geq 2$ and $x \in \Sigma^n$. Then:
\begin{enumerate}
  \item $\LZLC_R(x) \leq \AC_k^{\Rt}(x)$ for every finite $R$ with $k \in R$;
  \item $\AC_k^{\Rt}(x) \leq \ceil{\log_k n} \cdot (\LZLC_{\{k\}}(x) + \sigma)
    + \sigma$.
\end{enumerate}
\end{proposition}

The extra freedom of $\LZLC_R$ manifests once we consider several arities. For
a finite $R \subseteq \{2, 3, 4, \dots\}$, denote by $R^!$ the multiplicative
closure of $R$: the set of all finite products of elements of $R$ (including
the empty product $1$).

\begin{corollary}
\label{cor:multibase}
Let $R$ be finite, $x \in \Sigma^*$ and $k \in R^!$ with $k \geq 2$. Then
$\LZLC_R(x) \leq (k / 2) \cdot \AC_k^{\Rt}(x)$.
\end{corollary}

Thus the single measure $\LZLC_R$ simultaneously dominates RTL automaticity in
all the (infinitely many) bases of $R^!$; recall from Appendix D of
\citep{Kosoy2026b} that automaticity itself is highly sensitive to the choice
of base (cf.\ also the mix-automatic sequences of \citep{Endrullis2013}).

\section{The $P^3$ Predictor}
\label{sec:p3}

We now present the promised predictor for $\LZLC_R$, which we call \emph{Path
Plurality of Pluralities} ($P^3$). It generalizes the Tabular Plurality (TP)
algorithm of \citep{Kosoy2026b}, which handles the case $R = \{k\}$. The idea
behind TP is the following. Suppose that $x \sqsubseteq \val(H)$ for an exact
layered ZLP $H$ (\definitionref{def:exact}) all of whose arities equal $k$. Two
positions congruent modulo $k^j$ have the same $j$ least significant base-$k$
digits, so their descents pass through the same vertex $v$ of layer $j$, and
\lemmaref{lem:descent} gives $x[t] = \val_H(v)[\floor{t / k^j}]$ for every such
position $t$; consequently, the word read off the positions congruent to $q$
modulo $k^j$ --- an arithmetic subsequence of $x$ --- is, for every residue $q$,
a prefix of the value of one of the at most $w(H)$ vertices of layer $j$. TP
maintains, for every $j$, a partition of the residues modulo $k^j$ into groups
whose subsequences the observations so far have not shown to differ, together
with tentative tables recording, for each group modulo $k^j$ and each digit $i
< k$, the group modulo $k^{j+1}$ that contains its residues $q + i k^j$. In the
present setting, the arity profile of the program is unknown: the moduli that
play the role of $k^j$ are the scales $m_H(0), m_H(1), \dots$ of an unknown
$R$-adic program, which form some path through the multiplicative closure $R^!$
of $R$ (\sectionref{sec:rtl}) that starts at $1$ and multiplies by an element
of $R$ at each step. $P^3$ therefore maintains TP-style tables over a finite
fragment $S \subset R^!$ of this set and, when predicting, aggregates over all
maximal paths of $S$ by a weighted plurality vote, each path casting the
forecast that TP would make along it, itself a plurality --- whence the name.

\sectionref{sec:p3-model}, \sectionref{sec:p3-update} and
\sectionref{sec:p3-predict} give an overview of $P^3$, and
\sectionref{sec:p3-guarantees} its guarantees; the complete pseudocode, with a
line-by-line discussion, is given in \appendixref{app:p3-desc}.

\subsection{The Model}
\label{sec:p3-model}

The predictor receives $R$ as a parameter. For $m \geq 1$ and a residue $q \in
[m]$, the \emph{subsequence} of $q$ modulo $m$ is the arithmetic subsequence
$x[q] x[q + m] x[q + 2 m] \dots$ of $x$, read off the positions congruent to
$q$ modulo $m$; the subsequence of a position is that of its residue, and we
refer to the position $m \ell + q$ as \emph{row} $\ell$ of the subsequence of
$q$ modulo $m$, so that the rows of a subsequence enumerate its entries in
order. For $r \geq 2$, the subsequence of $q$ modulo $m$ is the zip
(\equationref{eq:zip}) of the subsequences of $q + i m$ modulo $m r$, $i < r$;
we say that the residues $q + i m$ modulo $m r$ \emph{refine} the residue $q$
modulo $m$. The internal state of $P^3$ consists of:
\begin{itemize}
  \item The time index $n \in \mathbb{N}$ --- the number of positions observed
    so far --- and a capacity parameter $c_{\max} \in \mathbb{N}_+$, which
    bounds the number of states kept at any scale (see below).
  \item A finite set of \emph{scales} $S \subset R^!$ with $1 \in S$, the scale
    $1$ being called the \emph{root}, and a set of \emph{edges} $\mathcal{E}
    \subseteq S \times R$, where $(m, r) \in \mathcal{E}$ requires $m r \in S$;
    every $m \in S \setminus \{1\}$ is required to be the target $m = m' r$ of
    at least one edge $(m', r) \in \mathcal{E}$. The directed graph on $S$ with
    an edge $m \to m r$ for every $(m, r) \in \mathcal{E}$ is the \emph{scale
    fragment}; its \emph{sinks} are the scales with no outgoing edges.
  \item For each $m \in S$: a set of \emph{states} $Q_m \subseteq [m]$ ---
    residues modulo $m$ whose subsequences the observations so far have shown
    to be pairwise distinct --- together with a partial \emph{output function}
    $\tau_m : Q_m \to \Sigma$, recording $x[q]$, the first symbol of the
    subsequence of $q$. Every residue modulo $m$ is tentatively
    \emph{identified} with a state, the model treating their subsequences as
    identical; a state is the minimal residue identified with it, in particular
    $0 \in Q_m$.
  \item For each $(m, r) \in \mathcal{E}$: a \emph{transition table} $\delta_m^r
    : Q_m \times [r] \to Q_{m r}$, where $\delta_m^r(q, i)$ is the state with
    which the refined residue $q + i m$ modulo $m r$ is identified: the minimal
    state of $Q_{m r}$ whose subsequence is consistent with all observations of
    the subsequence of $q + i m$ modulo $m r$.
\end{itemize}
We call a scale $m \in S$ \emph{beyond the horizon} when the positions observed
so far are all smaller than $m$. At such a scale, no subsequence has been
observed beyond its row $0$, so the states are simply the first occurrences of
the distinct symbols observed so far; the output functions are maintained only
beyond the horizon, which is the only place where they are consulted
(\sectionref{app:p3-update}). Exactly as for $R = \{k\}$ above, if $x
\sqsubseteq \val(H)$ for an exact $R$-adic layered $H$, the subsequences at
scale $m_H(j)$ are prefixes of the values of the at most $w(H)$ vertices of
layer $j$ (\lemmaref{lem:reg} provides exactness at the cost of $\sigma$ extra
width). For $R = \{k\}$, this is essentially the state of TP, $S$ being a chain
of powers of $k$.

\textbf{Subroutines.} The prediction and update routines share two helper
subroutines, Eval and Diverge (\algorithmref{alg:eval} and
\algorithmref{alg:diverge} in \sectionref{app:p3-subroutines}). Eval$(t)$
computes the model's current guess of $x[t]$, for $t \in \mathbb{N}$: it
follows a path of the fragment from the root, moving at the edge $(m, r)$ from
its current state $q \in Q_m$ to $\delta_m^r(q, \floor{t / m} \bmod r)$ --- the
state identified with the residue of $t$ modulo $m r$ --- and returns the
output function at the sink where it halts. The path is selected arbitrarily;
by the data-consistency invariant (claim 1 of \lemmaref{lem:p3-inv}), the value
returned at an already-observed position does not depend on this choice.
Diverge$()$ computes, for every scale $m \in S$ and every pair of states $q, q'
\in Q_m$, the \emph{divergence value} $d_m(q, q') \in \mathbb{N} \cup
\{\infty\}$: a row at which the subsequences of $q$ and $q'$ have been observed
to differ, while agreeing, as far as the model can tell, on all earlier rows
(claim 3 of \lemmaref{lem:p3-inv}; $d_m(q, q) = \infty$). The computation is by
dynamic programming over the scale fragment, from the sinks downwards. At a
sink, distinct states are declared to diverge already at row $0$, which is
sound because sinks lie beyond the horizon, where states are first occurrences
of distinct symbols (claim 4 of \lemmaref{lem:p3-inv}). At any other scale,
since row $r \ell' + i$ of the subsequence of $q$ modulo $m$ is row $\ell'$ of
the subsequence of $q + i m$ modulo $m r$, we let $d_m(q, q') \coloneqq \min (r
\cdot d_{m r}(\delta_m^r(q, i), \delta_m^r(q', i)) + i)$, the minimum ranging
over the edges $(m, r) \in \mathcal{E}$ and the digits $i < r$, with $r \cdot
\infty + i \coloneqq \infty$. The divergence tables are recomputed when needed
rather than stored, trading per-round time for a stronger compression bound
(cf.\ Remark A.6 of \citep{Kosoy2026b}).

\subsection{Update}
\label{sec:p3-update}

The initial state is the \emph{empty state} --- no scales and no edges, with
$c_{\max} = 1$ --- at which the prediction is an arbitrary fixed symbol. Upon
observing $x[n]$, the state is updated in four stages.

\textbf{Initialization.} At the empty state, the update creates the minimal
nonempty state: $S = \{1\}$, whose sole state $0 \in Q_1$ represents the entire
sequence, with $\tau_1(0) \coloneqq x[n]$, and the time index is reset, so that
positions are henceforth indexed relative to the current one. This happens at
round $0$ and after every restart (see below).

\textbf{Preprocessing.} When the time index reaches a scale, $n \in S$, the
fragment grows: for every $r \in R$ with $(n, r) \notin \mathcal{E}$, the scale
$n r$ is adjoined to $S$ if not already present, inheriting $Q_{n r} \coloneqq
Q_n$ and $\tau_{n r} \coloneqq \tau_n$, and the edge $(n, r)$ is added to
$\mathcal{E}$ with the transitions $\delta_n^r(q, 0) \coloneqq q$ --- the
residue $q$ modulo $n r$ is a state, as $Q_{n r} = Q_n$ --- and $\delta_n^r(q,
i) \coloneqq 0$ for $i \geq 1$: the subsequences of $q + i n$ modulo $n r$ with
$i \geq 1$ start at positions $\geq n$, hence are as yet entirely unobserved,
and are tentatively identified with the minimal state. The position $n$ itself
is the first to test these identifications, later in the same round. This is
the only mechanism by which the fragment acquires scales or edges.

\textbf{Repair.} A single pass over the scales $m \in S \setminus \{1\}$, in
increasing order, restores the consistency of the model with the new
observation by the smallest local repair available. At each edge $(m', r) \in
\mathcal{E}$ into $m$ whose source residue $n \bmod m'$ is a state, the entry
$q' \coloneqq \delta_{m'}^r(n \bmod m', \floor{n / m'} \bmod r)$ of its table
is examined. If the residue $n \bmod m$ is not itself a state, this entry
identifies it, tentatively, with a smaller state $q'$. Row $\ell \coloneqq
\floor{n / m}$ of the subsequence of $q'$ is the position $m \ell + q' < n$,
already observed; if the symbol there, Eval$(m \ell + q')$, disagrees with
$x[n]$, the identification is refuted, and the entry is \emph{rerouted} to the
minimal state $\tilde{q} \in Q_m$ that is compatible with the recorded
divergences, $d_m(\tilde{q}, q') \geq \ell$, and carries the new observation at
the current row, Eval$(m \ell + \tilde{q}) = x[n]$. If no such state exists,
the data has separated the subsequence of $n$ from that of every compatible
state, and the \emph{state} $n \bmod m$ \emph{is created}, inheriting the
outgoing transitions of $q'$ verbatim and, when $m > n$, recording $\tau_m(n)
\coloneqq x[n]$. Finally, once $n \bmod m$ is a state, every examined entry is
set to $n \bmod m$. A creation does not by itself restore consistency: it turns
$n \bmod m$ into a state, whose inherited transitions face the same test at the
larger scales of the pass, so a single mismatch \emph{cascades} upward along
the residues of $n$ until it is resolved by rerouting or anchored at the sinks.
For $R = \{k\}$, this is the update rule of TP.

\textbf{Post-processing.} Finally, every scale $m$ with $\abs{Q_m} > c_{\max}$
is deleted, together with the scales thereby disconnected from the root and,
iterating as needed, every scale $m \leq n$ that has lost all of its outgoing
edges --- so that every sink of the fragment stays beyond the horizon; the
edges incident to a deleted scale are deleted with it. If the root itself is
deleted, $c_{\max}$ is doubled and the predictor restarts, re-initializing on
the current symbol. The restarts divide the processing of the input into
\emph{epochs}, the $i$-th of which runs with the constant capacity $c_{\max} =
2^i$; owing to this doubling trick, no prior knowledge of the complexity of $x$
is required.

The update rule preserves the following invariants --- the analogues of those
of TP (Lemma 3.3 of \citep{Kosoy2026b}) --- which fix the semantics of the
model; the proof is given in \appendixref{app:p3-correctness}. The indexing is
local to an epoch: if an epoch begins at position $b$ of the original input, we
re-index the suffix $x[b], x[b+1], \dots$ from $0$ and again denote it by $x$,
so that the $n$ in the lemma is the last observed index of the current epoch.

\begin{lemma}
\label{lem:p3-inv}
Consider the state of $P^3$ in some epoch after updating on the prefix
$x[:n+1]$. The following statements hold.
\begin{enumerate}
  \item \textbf{Data consistency.} For every $t \leq n$, every run of
    Eval$(t)$ returns $x[t]$, regardless of how $r$ is selected in line 3 of
    \algorithmref{alg:eval}.
  \item \textbf{Routing.} For every $(m, r) \in \mathcal{E}$, $q \in Q_m$ and
    $i < r$: $\delta_m^r(q, i) \leq q + i m$, with equality whenever $q + i m
    \in Q_{m r}$. Moreover, routing is hereditary: every $u \in Q_{m r}$
    satisfies $u \bmod m \in Q_m$ --- so, by the equality above, $\delta_m^r(u
    \bmod m, \floor{u / m}) = u$.
  \item \textbf{Distinguishability.} For every $m \in S$ and distinct $q_1, q_2
    \in Q_m$, the divergence value $\ell \coloneqq d_m(q_1, q_2)$ computed by
    Diverge satisfies $m \ell + \max(q_1, q_2) \leq n$ and $x[m \ell + q_1]
    \neq x[m \ell + q_2]$.
  \item \textbf{Horizon.} For every $m \in S$ with $m > n$, the states at scale
    $m$ are precisely the first occurrences of the distinct symbols of
    $x[:n+1]$ --- $Q_m = \{t \leq n \mid \forall t' < t : x[t'] \neq x[t]\}$
    --- and $\tau_m(q) = x[q]$ for every $q \in Q_m$. Moreover, every sink of
    the scale fragment (a scale with no outgoing edges) satisfies $m > n$.
    Consequently, $\tau_m$ is injective on $Q_m$ at every sink $m$.
\end{enumerate}
\end{lemma}

\subsection{Prediction}
\label{sec:p3-predict}

Predict is invoked to forecast $x[n]$ after the state has been updated on
$x[:n]$. Call the pair $(m, q)$ \emph{active} when $m \in S$ and $q = n \bmod m
\in Q_m$. For an active pair and an edge $(m, r) \in \mathcal{E}$, let $i
\coloneqq \floor{n / m} \bmod r$ be the digit of $n$ at this edge and $q'
\coloneqq \delta_m^r(q, i)$, so that $q' \leq q + i m = n \bmod m r$ by routing
(claim 2 of \lemmaref{lem:p3-inv}); we call $(m, q, r)$ a \emph{leaf edge} when
this inequality is strict --- when the residue $n \bmod m r$ is not a state of
its own but is identified with the smaller state $q'$. These are the points at
which the model's treatment of $n$ is tentative. At a leaf edge, $P^3$
forecasts the current symbol by a plurality vote among states, exactly as TP
does at the single such point of its chain. With $\ell \coloneqq \floor{n / m
r}$ the current row, the candidates are the states $\tilde{q} \in Q_{m r}$
whose subsequences have not been observed to diverge from that of $q'$ before
row $\ell$, i.e.\ $d_{m r}(q', \tilde{q}) \geq \ell$; the vote of $\tilde{q}$
is the number of pairs $(\hat{q}, i') \in Q_m \times [r]$ with
$\delta_m^r(\hat{q}, i') = \tilde{q}$ whose row-$\ell$ position $m r \ell + i'
m + \hat{q}$ precedes $n$ --- the refined residues identified with $\tilde{q}$
that have already been traversed within the current row; and the forecast is
Eval$(m r \ell + q^*)$, the symbol at row $\ell$ of the subsequence of the
winner $q^*$.

The forecasts of the leaf edges along the different paths of the fragment are
aggregated by mass. Let $\eta \geq 0$ be the unique solution of $\sum_{r \in R}
r^{-\eta} = 1$, and give an edge of arity $r$ the weight $r^{-\eta}$ (for
$\abs{R} = 1$, $\eta = 0$ and all weights equal $1$). The weights are the step
probabilities of a random walk on $R^!$ that starts at the root and multiplies
its current scale by $r$ with probability $r^{-\eta}$; an \emph{infinite path}
is an infinite sequence of such steps, and we call the probability of a set of
infinite paths its \emph{mass}. Call an infinite path \emph{unblocked} if the
fragment contains every edge that it traverses before reaching a sink, and
\emph{blocked} otherwise; a blocked path is one that attempts an edge deleted
by post-processing. Prediction is the plurality vote in which every unblocked
path casts the forecast of the first leaf edge it encounters, and the votes for
a symbol are counted by mass; a path that encounters no leaf edge abstains,
which happens only when $n \in S$, for the paths whose walk reaches the sink
$n$. \algorithmref{alg:p3-predict} computes the vote by a single pass over the
scales in decreasing order, in time polynomial in the size of the state. Since
$\eta$ is in general irrational, the implementation works with a rational
approximation of it and fixed-point arithmetic (\sectionref{app:p3-predict}),
which the mistake analysis accounts for. For $R = \{k\}$, $\eta = 0$, the
fragment is a chain, and the vote is the forecast of the chain's first leaf
edge, as in TP.

\subsection{Guarantees}
\label{sec:p3-guarantees}

\begin{theorem}
\label{thm:p3}
Fix a finite nonempty $R \subseteq \{2, 3, 4, \dots\}$ and let $r \coloneqq
\max R$. The predictor $P^3 = P^3(R)$ is statistically and computationally
efficient with respect to $\LZLC_R$. Specifically, for any $n \geq 2$ and $x
\in \Sigma^n$, denoting $s \coloneqq \LZLC_R(x)$:
\begin{enumerate}
  \item \textbf{Statistical efficiency:} the number of mistakes is bounded by
    \begin{equation}
      \label{eq:p3-mistakes}
      M_{P^3}(x) = O(r \log r \cdot (s + \sigma) \log(r (s + \sigma)) \cdot
      \log n \cdot \log(r n))
    \end{equation}
  \item \textbf{Computational efficiency:} the predictor satisfies the
    compression bound
    \begin{equation}
      \label{eq:p3-state}
      S_{P^3}(x) = O((s + \sigma) \cdot (\abs{R} r \log(s + \sigma) + \log(r
      n)) \cdot (\log(2 r n))^{\abs{R}})
    \end{equation}
\end{enumerate}
\end{theorem}

The constants in \equationref{eq:p3-mistakes} and \equationref{eq:p3-state} are
absolute, and the time complexity of each round is polynomial in the size of
the state --- including the root-finding and the fixed-point evaluation of the
prediction vote, specified in \sectionref{app:p3-predict}.

Both bounds of \theoremref{thm:p3} are quasilinear in $s$; an exponential
dependence on $\abs{R}$ enters only through the factor $(\log(2 r
n))^{\abs{R}}$ of \equationref{eq:p3-state}, which bounds the number of scales
$\abs{S}$. Whether this dependence can be avoided, we leave as an open problem.
In particular, for an infinite sequence $x \in \Sigma^\omega$ s.t.\
$\LZLC_R(x[:n])$ grows polylogarithmically in $n$, $P^3$ makes
polylogarithmically many mistakes while running in polylog space and
quasilinear cumulative time: the regime promised in \sectionref{sec:intro}.
Comparing with previous results: by \corollaryref{cor:szlc}, the ARC
predictor of \citep{Kosoy2026b} already achieves a mistake bound of $O(r (s +
\sigma) (\log n)^2)$ on words of low $\LZLC_R$ --- the contribution of
\theoremref{thm:p3} is the compression bound \equationref{eq:p3-state}, which
that predictor entirely lacks; conversely, for $R = \{k\}$, combining
\theoremref{thm:p3} with \propositionref{prop:rtl} yields guarantees
comparable, up to logarithmic factors, to those established for TP in Theorem
3.4 of \citep{Kosoy2026b}.

\textbf{Proof outline.} The invariants of \lemmaref{lem:p3-inv} are established
in \appendixref{app:p3-correctness}, and \appendixref{app:p3-proof} derives the
bounds from them. First, every retained scale has $O(s + \sigma)$ states
(\sectionref{app:p3-state-bound}): within an epoch, let $G$ be an exact
$R$-adic layered ZLP of width at most $2(s + \sigma)$ for the suffix being
processed; by distinguishability, the states at the scale $m_G(j)$ inject into
layer $j$ of $G$, and those at $m_G(d(G))$ into $\Sigma$, so once $c_{\max}
\geq 2(s + \sigma)$ the chain of scales of $G$ is never deleted and no further
restart occurs, whence $c_{\max} < 4(s + \sigma)$ throughout. For the mistake
bound (\sectionref{app:p3-stat}), the unblocked paths are the experts of a
plurality vote under the prior given by mass: within an epoch, their set only
shrinks, while its mass stays at least $(r n)^{-\eta}$; a single edge forecasts
wrongly at most $O(r c_{\max} \log(r c_{\max}))$ times, by a coupling to the
branch prediction game of \citep{Kosoy2026b}; and a path is charged only at the
$O(\log n)$ edges of its own prefix. A lemma on plurality voting under such
budgets (\sectionref{app:budgeted-plurality}) bounds the mistakes of an epoch,
and summing over the epochs, whose capacities form a geometric progression,
gives \equationref{eq:p3-mistakes}. Finally, the fragment has at most $(\log(2
r n))^{\abs{R}}$ scales, each carrying $O(s + \sigma)$ states of $O(\log(r n))$
bits and at most $\abs{R}$ tables of $O(r (s + \sigma))$ entries of $O(\log(s +
\sigma))$ bits each; this gives \equationref{eq:p3-state}
(\sectionref{app:p3-compression}).

\section*{AI Disclosure}

The core mathematical content and the text of the introduction were created by 
the author with only minimal AI assistance. The other sections of the paper 
were mostly written by AI based on the author's detailed research notes and guidance, 
with some back-and-forth and small manual corrections. The notes contained the algorithms and high-level 
proof sketches but not all the proof details. In the process, the AI
identified and proposed corrections for some minor technical errors in the notes, and also 
fleshed out certain technical details, e.g. the finite-precision treatment 
of the $\eta$ parameter. All AI-written material was subsequently reviewed by the author,
who takes full responsibility for its correctness. Multiple versions of Anthropic Claude
and OpenAI ChatGPT were involved, but Anthropic Claude Fable 5 was used the most.

\acks{This work was supported by the Advanced Research+Invention Agency (ARIA) of the United Kingdom, Survival and Flourishing Corp, and Coefficient Giving in San Francisco, California.}

\bibliography{references}

\appendix

\section{Proofs for \texorpdfstring{\sectionref{sec:ziplines}}{Section 3}}
\label{app:proofs}

The following lemma is the basic link between layered programs and automata
reading the digits of the position index (\sectionref{sec:ziplines}). By its
claim 2, every position $t < \abs{\val(H)}$ satisfies $t < m_H(d(H))$, so that
the numbers $t_j$ appearing in it are exactly the digits of $t$ in the
mixed-radix numeration system with radices $\ar_H(0), \ar_H(1), \dots$, i.e.\
$t = \sum_{j < d(H)} t_j m_H(j)$.

\begin{lemma}[Descent]
\label{lem:descent}
Let $H$ be a layered ZLP and let $t < \abs{\val(H)}$. For $j < d(H)$, denote
$t_j \coloneqq \floor{t / m_H(j)} \bmod \ar_H(j)$. Define the \emph{descent} of
$t$ in $H$ by $q^0 \coloneqq q_0$ and, recursively, $c^{j+1} \coloneqq
\delta(q^j)[t_j]$ and $q^{j+1} \coloneqq c^{j+1}$, as long as $c^{j+1} \in Q$;
this is well-defined, since $q^j \in Q_j$ by \definitionref{def:layered}, so
that $t_j < \ar(q^j)$. Then:
\begin{enumerate}
  \item whenever $q^j$ is defined (i.e.\ $c^1, \dots, c^j \in Q$), $\floor{t /
    m_H(j)} < \abs{\val_H(q^j)}$ and
    \begin{equation}
      \label{eq:descent}
      \val(H)[t] = \val_H(q^j)[\floor{t / m_H(j)}]
    \end{equation}
  \item the descent stops at a terminal: there is $J(t) < d(H)$ with $c^1,
    \dots, c^{J(t)} \in Q$ and $c^{J(t)+1} \in \Sigma$; moreover, $t < m_H(J(t)
    + 1)$ and $\val(H)[t] = c^{J(t)+1}$.
\end{enumerate}
\end{lemma}

\begin{proof}
By \equationref{eq:zip}, for any $q \in Q$ and $i < \abs{\val_H(q)}$, with $r
\coloneqq \ar(q)$,
\begin{equation}
  \label{eq:zip-step}
  \val_H(q)[i] = \val_H(\delta(q)[i \bmod r])[\floor{i / r}]
\end{equation}
where the position on the right-hand side is smaller than the length of the
corresponding value, since \equationref{eq:zip} defines $z[i]$ for exactly the
$i < \abs{z}$.

Consider a $j$ for which $q^j$ is defined and claim 1 holds, and put $r
\coloneqq \ar_H(j) = \ar(q^j)$. Then
\begin{equation}
  \label{eq:descent-step}
  \begin{aligned}
    \val(H)[t] &= \val_H(q^j)[\floor{t / m_H(j)}] \\
    &= \val_H(\delta(q^j)[\floor{t / m_H(j)} \bmod r])[\floor{\floor{t /
      m_H(j)} / r}] \\
    &= \val_H(c^{j+1})[\floor{t / m_H(j+1)}],
  \end{aligned}
\end{equation}
and the position $\floor{t / m_H(j+1)}$ is smaller than
$\abs{\val_H(c^{j+1})}$. Here the first equality is \equationref{eq:descent}
for $j$; the second is \equationref{eq:zip-step} applied to $q \coloneqq q^j$
and $i \coloneqq \floor{t / m_H(j)}$, which is admissible since $i <
\abs{\val_H(q^j)}$ by claim 1 for $j$, and which also gives the bound on the
position; and the third holds since $\floor{t / m_H(j)} \bmod r = t_j$ by
definition, so that $\delta(q^j)[t_j] = c^{j+1}$, while $\floor{\floor{t /
m_H(j)} / r} = \floor{t / m_H(j+1)}$, as $m_H(j+1) = m_H(j) r$.

Claim 1 now follows by induction on $j$: it holds for $j = 0$, since $q^0 =
q_0$, $m_H(0) = 1$ and $t < \abs{\val(H)} = \abs{\val_H(q_0)}$; and if it holds
for $j$ and $q^{j+1}$ is defined, i.e.\ $c^{j+1} = q^{j+1}$, then
\equationref{eq:descent-step} and the accompanying bound are claim 1 for $j +
1$.

For claim 2: since $q^j \in Q_j$ and $Q_{d(H)} = \emptyset$, $q^{d(H)}$ is
undefined, so the recursion stops at some $J \coloneqq J(t) < d(H)$, with $c^1,
\dots, c^J \in Q$ and $c^{J+1} \in \Sigma$. By claim 1,
\equationref{eq:descent-step} is available for $j = J$, and the accompanying
bound reads $\floor{t / m_H(J+1)} < \abs{\val_H(c^{J+1})} = 1$. Hence $t <
m_H(J+1)$, and \equationref{eq:descent-step} gives $\val(H)[t] =
\val_H(c^{J+1})[0] = c^{J+1}$.
\end{proof}

\begin{proofof}{\propositionref{prop:lzlc-wcm}}
Claim 2 is immediate from \definitionref{def:lzlc}: if $x \sqsubseteq y
\sqsubseteq \val(H)$ then $x \sqsubseteq \val(H)$. For claim 1, the lower bound
holds because $Q_0 = \{q_0\} \neq \emptyset$. For the upper bound, we construct
an $R$-adic layered ZLP $H$ with $x \sqsubseteq \val(H)$ and $w(H) \leq
\max(\abs{x} - 1, 1)$. Let $r \coloneqq \min R$, $N \coloneqq \max(\abs{x},
2)$, and let $d \geq 1$ be minimal with $r^d \geq N$. The vertices of layer $j
< d$ are the elements of $[r^j]$ (a separate copy for each layer), the root is
$0 \in [r^0]$, and for $q \in [r^j]$ and $i < r$:
\begin{itemize}
  \item if $j < d - 1$, then $\delta(q)[i] \coloneqq q + i r^j \in [r^{j+1}]$;
  \item if $j = d - 1$, then $\delta(q)[i]$ is the terminal $x[q + i r^{d-1}]$
    when $q + i r^{d-1} < \abs{x}$, and an arbitrary terminal otherwise.
\end{itemize}
For $j < d - 1$, every $q' \in [r^{j+1}]$ is the child $\delta(q' \bmod
r^j)[\floor{q' / r^j}]$, so the root is the unique source, and $H$ is an
$\{r\}$-adic, hence $R$-adic, layered ZLP of depth $d$ and width $w(H) =
r^{d-1}$. Only the vertices of its last layer have terminal children, so $H$ is
exact by \lemmaref{lem:exact}, and $\abs{\val(H)} = r^d \geq \abs{x}$.

Now fix $t < \abs{x}$ and consider its descent in $H$ (\lemmaref{lem:descent}),
whose digits are $t_j = \floor{t / r^j} \bmod r$. By induction on $j$, $q^j = t
\bmod r^j$ for all $j < d$: indeed, $q^0 = 0$, and $q^{j+1} = q^j + t_j r^j =
(t \bmod r^j) + (\floor{t / r^j} \bmod r) r^j = t \bmod r^{j+1}$. Hence $c^d =
x[q^{d-1} + t_{d-1} r^{d-1}] = x[t \bmod r^d] = x[t]$, as $t < \abs{x} \leq
r^d$, and claim 2 of \lemmaref{lem:descent} gives $\val(H)[t] = x[t]$. Thus $x
\sqsubseteq \val(H)$, and
\[
  \LZLC_R(x) \leq w(H) = r^{d-1} \leq N - 1 = \max(\abs{x} - 1, 1),
\]
where $r^{d-1} \leq N - 1$ holds since $r^{d-1} = 1 \leq N - 1$ if $d = 1$, and
$r^{d-1} < N$ by the minimality of $d$ if $d \geq 2$.

Finally, claims 1 and 2 give \equationref{eq_word_com_bound} (with $p \coloneqq
x + 1$) and \equationref{eq:mono} (with $q \coloneqq 1$), so $\LZLC_R$ is a
word complexity measure.
\end{proofof}

\begin{proofof}{\lemmaref{lem:reg}}
By claim 2 of \lemmaref{lem:descent}, every $t < \abs{\val(H)}$ satisfies $t <
m_H(d(H))$; hence $m_H(d(H)) \geq \abs{\val(H)} \geq n$, and
\[
  d^* \coloneqq \min \{j \mid m_H(j) \geq n\}
\]
is well-defined, with $1 \leq d^* \leq d(H)$ (as $m_H(0) = 1 < n$). By the
minimality of $d^*$, $m_H(d^* - 1) < n$; since every arity is at least $2$,
$m_H(d^* - 1) \geq 2^{d^* - 1}$, so $d^* - 1 < \log n$ and $d^* \leq \ceil{\log
n}$. In the uniform-arity case, $k^{d^* - 1} < n \leq k^{d^*}$, i.e.\ $d^* =
\ceil{\log_k n}$.

We now construct $H'$, of depth $d^*$, by cutting $H$ below layer $d^* - 1$:
each child of that layer is replaced by the first symbol of its value, and each
terminal child of an earlier layer is replaced by a padding vertex of constant
value. For $c \in \overline{Q}$, write $\mathrm{fs}(c) \coloneqq \val_H(c)[0]$
(values are non-empty: terminals have length $1$, and by \equationref{eq:zip} a
zip of non-empty words has length at least $2$). The vertices of $H'$ are
$\bigcup_{j < d^*} Q_j$ together with padding vertices $\gamma_a^{(j)}$ for $a
\in \Sigma$ and $1 \leq j < d^*$, where the vertices of $Q_j$ and the
$\gamma_a^{(j)}$ are placed in layer $j$; unreachable vertices are discarded.
Define the productions as follows:
\begin{itemize}
  \item if $q \in Q_j$ and $j < d^* - 1$, retain each non-terminal child
    $\delta(q)[i]$, and replace each terminal child $\delta(q)[i] = a \in
    \Sigma$ by $\gamma_a^{(j+1)}$;
  \item if $q \in Q_{d^* - 1}$, put $\delta'(q)[i] \coloneqq
    \mathrm{fs}(\delta(q)[i])$;
  \item for $a \in \Sigma$, let $\delta'(\gamma_a^{(j)})$ consist of $\ar_H(j)$
    copies of $\gamma_a^{(j+1)}$ when $j < d^* - 1$, and let
    $\delta'(\gamma_a^{(d^* - 1)})$ consist of $\ar_H(d^* - 1)$ copies of $a$.
\end{itemize}
Every child of a vertex of layer $j$ lies in layer $j + 1$ or is terminal, and
all vertices of layer $j$ have arity $\ar_H(j)$. So $H'$ is a layered ZLP with
$d(H') = d^*$ and $\ar_{H'}(j) = \ar_H(j)$ for $j < d^*$ --- so it is $R$-adic,
and all its arities equal $k$ when those of $H$ do --- and its layers satisfy
$\abs{Q'_j} \leq \abs{Q_j} + \sigma \leq w(H) + \sigma$. Every child outside its
last layer is non-terminal, so it is exact by \lemmaref{lem:exact}, and
$m_{H'}(d(H')) = m_H(d^*) \geq n$. This proves claims 2--4 and the
uniform-arity depth formula.

It remains to verify claim 1, i.e.\ that $\val(H')[t] = x[t]$ for every $t <
n$. Fix such $t$. Since $t < n \leq \abs{\val(H)}$ and $t < n \leq
\abs{\val(H')}$, \lemmaref{lem:descent} applies to the descents of $t$ in both
$H$ and $H'$, and these use the same digits $t_0, \dots, t_{d^* - 1}$, the two
arity profiles agreeing through layer $d^* - 1$. Let $J \coloneqq J(t)$ be the
layer at which the descent in $H$ stops and $a \coloneqq c^{J+1}$ its terminal,
so that $\val(H)[t] = a$ by claim 2 of \lemmaref{lem:descent}. Both descents
start at $q_0$, and as long as the descent in $H$ takes a non-terminal child at
a layer below $d^* - 1$, the descent in $H'$ takes the same child.
\begin{itemize}
  \item If $J < d^* - 1$, then the descent in $H'$ reaches $q^J$ and there
    takes the child $\delta'(q^J)[t_J] = \gamma_a^{(J+1)}$, whose descendants
    along any digits are $\gamma_a^{(J+2)}, \dots, \gamma_a^{(d^* - 1)}$ and
    then the terminal $a$. So $\val(H')[t] = a$ by claim 2 of
    \lemmaref{lem:descent}.
  \item If $J \geq d^* - 1$, then both descents reach the same $q \coloneqq
    q^{d^* - 1} \in Q_{d^* - 1}$. With $c \coloneqq \delta(q)[t_{d^* - 1}]$,
    the descent in $H'$ stops at the terminal $\delta'(q)[t_{d^* - 1}] =
    \mathrm{fs}(c)$, so $\val(H')[t] = \mathrm{fs}(c)$. If $c$ is terminal,
    then $J = d^* - 1$ and $\val(H)[t] = a = c = \mathrm{fs}(c)$. Otherwise $c
    = q^{d^*}$ is defined, and \equationref{eq:descent} for $j = d^*$ gives
    $\val(H)[t] = \val_H(c)[\floor{t / m_H(d^*)}] = \val_H(c)[0] =
    \mathrm{fs}(c)$, as $t < n \leq m_H(d^*)$.
\end{itemize}
In either case $\val(H')[t] = \val(H)[t] = x[t]$.
\end{proofof}

\begin{proofof}{\lemmaref{lem:radix-refinement}}
Put $F_j \coloneqq \prod_{i < j} f_i$, so $F_0 = 1$ and $F_l = k$. The program
$H$ simulates $M$, a block of $l$ consecutive layers reading one base-$k$ digit
of the position, one factor $f_j$ at a time. Its vertices at layer $h l + j$,
where $h < m$ and $j < l$, are the pairs $(q, u) \in Q_M \times [F_j]$ (a
separate copy for each layer), $q$ being meant as the state of $M$ after
reading $h$ base-$k$ digits and $u$ as the portion of the next digit assembled
from its first $j$ sub-digits; the root is $q_0 \coloneqq (q_M, 0)$ at layer
$0$, and unreachable vertices are discarded. For a vertex $(q, u)$ at layer $h
l + j$ and $v \in [f_j]$, put $d \coloneqq u + v F_j \in [F_{j+1}]$ and define
\[
  \delta((q, u))[v] \coloneqq \begin{cases}
    (q, d) \text{ at layer } h l + j + 1 & \text{if } j < l - 1, \\
    (\delta_M(q, d), 0) \text{ at layer } (h + 1) l & \text{if } j = l - 1
      \text{ and } h < m - 1, \\
    \tau_M(\delta_M(q, d)) \in \Sigma & \text{if } j = l - 1 \text{ and } h = m
      - 1.
  \end{cases}
\]
Thus $H$ is a layered ZLP of depth $l m$ with $\ar_H(h l + j) = f_j$, and no
terminal occurs before its last layer, so $H$ is exact by
\lemmaref{lem:exact}; in particular, $\abs{\val(H)} = m_H(l m) = k^m \geq n$.
At layer $h l + j$ there are at most $\abs{Q_M} F_j$ vertices, which gives the
width bound.

It remains to show that $\val(H)[t] = x[t]$ for every $t < n$. Let $D_h
\coloneqq \floor{t / k^h} \bmod k$ be the base-$k$ digits of $t$, so that
$(\langle t \rangle_k^m)^{\Rt} = D_0 D_1 \dots D_{m-1}$, and let $t_0, \dots,
t_{l m - 1}$ be its digits in the mixed-radix system of $H$
(\lemmaref{lem:descent}). These are related as follows: writing each $D_h$ in
the mixed-radix system with radices $f_0, \dots, f_{l-1}$, i.e.\ $D_h = \sum_{j
< l} e_{h, j} F_j$ with $e_{h, j} < f_j$, and substituting into $t = \sum_{h <
m} D_h k^h$ gives $t = \sum_{h < m} \sum_{j < l} e_{h, j} k^h F_j$; since $k^h
F_j = m_H(h l + j)$, this is a representation of $t$ in the mixed-radix system
of $H$, which is unique, so $e_{h, j} = t_{h l + j}$. That is, the digits $t_{h
l}, \dots, t_{h l + l - 1}$ read by the $h$-th block of layers are the
sub-digits of $D_h$:
\[
  D_h = \sum_{j < l} t_{h l + j} F_j.
\]
By induction along the descent of $t$, the vertex reached at layer $h l + j$ is
\[
  q^{h l + j} = \left(\delta_M^*(q_M, D_0 \dots D_{h-1}), \sum_{i < j} t_{h l +
  i} F_i\right).
\]
Indeed, the second coordinate accumulates the sub-digits of $D_h$ until, at $j
= l - 1$, the assembled digit is $d = D_h$ and the first coordinate advances to
$\delta_M(\delta_M^*(q_M, D_0 \dots D_{h-1}), D_h) = \delta_M^*(q_M, D_0 \dots
D_h)$. At the last layer, $h = m - 1$ and $j = l - 1$, the descent therefore
stops at the terminal $\tau_M(\delta_M^*(q_M, D_0 \dots D_{m-1})) = M((\langle
t \rangle_k^m)^{\Rt}) = x[t]$, and claim 2 of \lemmaref{lem:descent} gives
$\val(H)[t] = x[t]$.
\end{proofof}

\begin{proofof}{\propositionref{prop:rtl}}
For claim 1, let $M$ realize $\AC_k^{\Rt}(x)$, with $m$ witnessing the
realization; then $m \geq 1$, as $k^m \geq n \geq 2$. Apply
\lemmaref{lem:radix-refinement} with $l = 1$ and $f_0 = k$, so that the maximum
in its width bound is the empty product $1$. The resulting $\{k\}$-adic, hence
$R$-adic, layered ZLP $H$ has $x \sqsubseteq \val(H)$ and $w(H) \leq
\abs{Q_M}$, proving $\LZLC_R(x) \leq \AC_k^{\Rt}(x)$.

For claim 2, let $H$ realize $s \coloneqq \LZLC_{\{k\}}(x)$ and let $H'$ be
obtained from \lemmaref{lem:reg}, so that $d(H') = m \coloneqq \ceil{\log_k
n}$, all arities of $H'$ equal $k$, $\abs{Q'} \leq m (s + \sigma)$, and $H'$ is
exact with $\abs{\val(H')} = k^m \geq n$. Define a DFAO $M$ over $[k]$ with
state set $Q' \cup \{f_a \mid a \in \Sigma\}$, initial state $q_0$, transitions
\[
  \delta_M(q, i) \coloneqq \begin{cases}
    \delta'(q)[i] & \text{if } q \in Q' \text{ and } \delta'(q)[i] \in Q', \\
    f_{\delta'(q)[i]} & \text{if } q \in Q' \text{ and } \delta'(q)[i] \in
      \Sigma, \\
    f_a & \text{if } q = f_a,
  \end{cases}
\]
and output $\tau_M(f_a) \coloneqq a$ (and arbitrary on $Q'$). Fix $t < n$. The
$j$-th letter of $(\langle t \rangle_k^m)^{\Rt}$ is $\floor{t / k^j} \bmod k =
t_j$, the $j$-th digit of the descent of $t$ in $H'$ (\lemmaref{lem:descent},
which applies as $t < \abs{\val(H')}$). Hence the run of $M$ on this word
follows the descent: its state after $j$ letters is $q^j$ as long as $q^j$ is
defined, and after $J(t) + 1 \leq m$ letters it is $f_{c^{J(t)+1}} =
f_{\val(H')[t]} = f_{x[t]}$ (claim 2 of \lemmaref{lem:descent}), where it stays
for the remaining letters. So $M((\langle t \rangle_k^m)^{\Rt}) = x[t]$ for
every $t < n$, and $M$ witnesses $\AC_k^{\Rt}(x) \leq \abs{Q'} + \sigma \leq m
(s + \sigma) + \sigma$.
\end{proofof}

\begin{proofof}{\corollaryref{cor:multibase}}
If $\abs{x} \leq 1$, then $\LZLC_R(x) \leq 1 \leq (k / 2) \cdot \AC_k^{\Rt}(x)$
by claim 1 of \propositionref{prop:lzlc-wcm} and $\AC_k^{\Rt}(x) \geq 1$.
Otherwise, let $M$ realize $\AC_k^{\Rt}(x)$, with $m$ witnessing the
realization; then $m \geq 1$, as $k^m \geq \abs{x} \geq 2$. Since $k \in R^!$
and $k \geq 2$, we may write $k = f_0 f_1 \dots f_{l-1}$ with $l \geq 1$, all
$f_i \in R$ and, placing the largest factor last, $f_{l-1} = \max_{i<l} f_i$.
Applying \lemmaref{lem:radix-refinement} gives a layered ZLP $H$ with $x
\sqsubseteq \val(H)$, which is $R$-adic since $\ar_H(j) = f_{j \bmod l} \in R$,
and whose width is at most
\[
  \abs{Q_M} \max_{j < l} \prod_{i < j} f_i = \abs{Q_M} \prod_{i < l - 1} f_i =
  \abs{Q_M} k / f_{l-1} \leq \abs{Q_M} k / 2,
\]
the maximum being attained at $j = l - 1$. Hence $\LZLC_R(x) \leq (k / 2) \cdot
\AC_k^{\Rt}(x)$.
\end{proofof}

\section{Algorithm Description}
\label{app:p3-desc}

In this appendix, we give the complete pseudocode of the $P^3$ predictor, with
a line-by-line discussion, completing the overview of \sectionref{sec:p3},
whose notation and terminology --- scales, edges, states, subsequences and
rows, the horizon --- we use throughout. \appendixref{app:p3-correctness} and
\appendixref{app:p3-proof} then prove \theoremref{thm:p3}.

\subsection{Subroutines}
\label{app:p3-subroutines}

The helper subroutines Eval and Diverge, shared by the prediction and update
routines and described in \sectionref{sec:p3-model}, are
\algorithmref{alg:eval} and \algorithmref{alg:diverge}. In
\algorithmref{alg:eval}, the path taken through the fragment is selected
arbitrarily in line 3. In \algorithmref{alg:diverge}, line 2 declares the
subsequences of distinct states of a sink to diverge already at row $0$, and
the recurrence of line 6 is evaluated with the convention $r \cdot \infty + i
\coloneqq \infty$.

\begin{algorithm}[htbp]
\floatconts{alg:eval}
{\caption{Eval$(t)$}}
{\begin{alglines}
  \item $m \gets 1$;\hspace{0.5em} $q \gets 0$
  \item \textbf{while} ${\{r \in R \mid (m, r) \in \mathcal{E}\} \neq
    \emptyset}$
  \item \algind{1} select $r \in R$ s.t.\ $(m, r) \in \mathcal{E}$
  \item \algind{1} $q \gets \delta_m^r(q, \floor{t / m} \bmod r)$
  \item \algind{1} $m \gets m r$
  \item \textbf{return} $\tau_m(q)$
\end{alglines}}
\end{algorithm}

\begin{algorithm}[htbp]
\floatconts{alg:diverge}
{\caption{Diverge$()$}}
{\begin{alglines}
  \item $A \gets \{m \in S \mid \forall r \in R : (m, r) \notin \mathcal{E}\}$
  \item \textbf{for} $m \in A$ and $q, q' \in Q_m$: \textbf{if} $q = q'$
    \textbf{then} $d_m(q, q') \gets \infty$ \textbf{else} $d_m(q, q') \gets 0$
  \item \textbf{while} $A \neq S$
  \item \algind{1} select $m \in S \setminus A$ s.t.\ $\forall r \in R : (m, r)
    \in \mathcal{E} \Rightarrow m r \in A$
  \item \algind{1} \textbf{for} $q, q' \in Q_m$
  \item \algind{2} $d_m(q, q') \gets \min_{r \in R : (m, r) \in \mathcal{E}} \
    \min_{i < r} (r \cdot d_{m r}(\delta_m^r(q, i), \delta_m^r(q', i)) + i)$
  \item \algind{1} $A \gets A \cup \{m\}$
  \item \textbf{return} $\{d_m\}_{m \in S}$
\end{alglines}}
\end{algorithm}

\subsection{Initialization}
\label{app:p3-init}

The initial state $s_{\mathrm{init}}$ of the protocol
(\sectionref{sec:protocol}) is the \emph{empty state}: no scales and no edges
($S = \emptyset$, $\mathcal{E} = \emptyset$), with $c_{\max} = 1$ and the time
index immaterial; the prediction function outputs an arbitrary fixed symbol
there, so round $0$ costs at most one mistake. The routine Initialize
(\algorithmref{alg:p3-init}) is the branch of Update that handles this state
(line 1 of \algorithmref{alg:p3-update}): invoked with the observed symbol $a =
x[0]$, it creates the minimal nonempty state. The fragment $S = \{1\}$ consists
of the single scale $1$, whose sole state $0 \in Q_1$ --- the minimal (indeed,
the only) residue modulo 1 --- represents the entire sequence, and the output
function records the symbol $x[0]$ at its row 0. No transition tables exist
yet, as the fragment has no edges; the fragment first grows at time $n = 1$,
when $1 \in S$ triggers preprocessing (\sectionref{app:p3-preprocessing}) and
the scales $R$ are adjoined. In particular, immediately after initialization
every run of Eval$(0)$ returns $\tau_1(0) = x[0]$ at once, and the invariants of
\lemmaref{lem:p3-inv} hold vacuously.

\begin{algorithm}[htbp]
\floatconts{alg:p3-init}
{\caption{$P^3$ Initialize$(a)$}}
{\begin{alglines}
  \item $n \gets 0$;\hspace{0.5em} $S \gets \{1\}$;\hspace{0.5em} $\mathcal{E}
    \gets \emptyset$
  \item $Q_1 \gets \{0\}$;\hspace{0.5em} $\tau_1(0) \gets a$
\end{alglines}}
\end{algorithm}

The capacity parameter is deliberately left untouched by
\algorithmref{alg:p3-init}: it starts at $c_{\max} = 1$ in $s_{\mathrm{init}}$
and is modified only by the doubling trick of post-processing
(\sectionref{app:p3-post}). When post-processing empties the fragment, the
predictor restarts by invoking Initialize with the symbol observed in the
current round; resetting the time index in line 1 re-indexes positions relative
to the current one --- the shared increment of line 19 of
\algorithmref{alg:p3-update} then leaves $n = 1$, exactly as after round $0$
--- so that the remaining suffix is processed as a fresh sequence under the
doubled capacity.

\subsection{Preprocessing}
\label{app:p3-preprocessing}

The routine Preprocess (\algorithmref{alg:p3-preprocess}) is invoked at the
beginning of every update, before the main cascade
(\sectionref{app:p3-update}), and grows the fragment as described in
\sectionref{sec:p3-update}, creating the missing edges out of the scale reached
by the time index. Note that the guard of line 2 is on the edge rather than on
the scale: for $\abs{R} \geq 2$, the scale $n r$ may predate the edge (with $R
= \{2, 3\}$, for instance, the scale $6$ is adjoined at time $2$, as $2 \cdot
3$, while the edge $3 \to 6$ is only created at time $3$), in which case only
the transitions are missing and line 3 is skipped.

\begin{algorithm}[htbp]
\floatconts{alg:p3-preprocess}
{\caption{$P^3$ Preprocess$()$}}
{\begin{alglines}
  \item \textbf{if} $n \in S$
  \item \algind{1} \textbf{for} $r \in R$ s.t.\ $(n, r) \notin \mathcal{E}$
  \item \algind{2} \textbf{if} $n r \notin S$ \textbf{then} $S \gets S \cup \{n
    r\}$;\hspace{0.5em} $Q_{n r} \gets Q_n$;\hspace{0.5em} $\tau_{n r} \gets
    \tau_n$
  \item \algind{2} $\mathcal{E} \gets \mathcal{E} \cup \{(n, r)\}$
  \item \algind{2} \textbf{for} $q \in Q_n$
  \item \algind{3} $\delta_n^r(q, 0) \gets q$
  \item \algind{3} \textbf{for} $1 \leq i < r$: \hspace{0.4em} $\delta_n^r(q,
    i) \gets 0$
\end{alglines}}
\end{algorithm}

The transitions follow the semantics of \sectionref{sec:p3-model}. Refining the
residue $q$ modulo $n$ by the digit $i = 0$ yields the residue $q + 0 \cdot n =
q$ modulo $n r$, so line 6 routes each state to its own residue --- as the
routing invariant (claim 2 of \lemmaref{lem:p3-inv}) mandates whenever $q \in
Q_{n r}$. The subsequences of $q + i n$ modulo $n r$ with $i \geq 1$, by
contrast, start at positions $q + i n \geq n$, hence are as yet entirely
unobserved; every state of $Q_{n r}$ is consistent with this empty evidence,
and line 7 tentatively routes each such residue to the minimal state, $0 \in
Q_{n r}$. (The residue $0$ belongs to $Q_m$ for every $m \in S$: it is placed
in $Q_1$ by \algorithmref{alg:p3-init}, copied by line 3 on scale creation,
never removed by the update, and discarded by post-processing only together
with its scale.)

When the scale $n r$ predates the edge, it was created from a different
predecessor and has since evolved independently, so it is not immediate that
line 6 is even well-defined --- that $q \in Q_{n r}$ for every $q \in Q_n$.
This is where the horizon invariant (claim 4 of \lemmaref{lem:p3-inv}) comes
in. When Preprocess runs, the observed prefix is $x[:n]$, so both $n$ and $n r$
are beyond the horizon; every observed position $t < n$ has $\floor{t / n} = 0$,
so the subsequences at these scales have been observed at row $0$ only ---
words of length at most $1$, distinguishable only by their single symbols.
Accordingly, by claim 4, the states of every such scale are exactly the first
occurrences of the distinct symbols of $x[:n]$: in particular $Q_{n r} = Q_n$
and $\tau_{n r} = \tau_n$ hold already, the assignments of line 3 would have
been no-ops, and the same transitions serve both cases of the guard.

The timing is as lazy as possible. Creating the edge earlier would gain
nothing: all evidence available at scale $n r$ before time $n$ sits in row $0$
of the $i = 0$ refinements, which is exactly what $\tau_{n r}$ records.
Position $n$ itself, whose digit at scale $n$ is $\floor{n / n} \bmod r = 1$,
is the first to leave the $i = 0$ refinements, and the tentative transitions of
line 7 are put to the test by the main update of the very same round. Data
consistency (claim 1 of \lemmaref{lem:p3-inv}) is undisturbed: a run of
Eval$(t)$ with $t < n$ that crosses a new edge does so through $\delta_n^r(q,
0) = q$ of line 6, halts at $n r$ --- a scale beyond the horizon has no
outgoing edges, as those would only be created at time $n r$ --- and reads
$\tau_{n r}(q) = \tau_n(q)$, the same symbol it would have read at scale $n$
before the edge existed. Finally, for $R = \{k\}$ the fragment is a chain of
powers of $k$ whose interior edges were created at their own scales' rounds, so
the guard of line 2 fires exactly when $n$ reaches the maximal scale;
\algorithmref{alg:p3-preprocess} then reduces to the preprocessing step of TP,
extending the chain by one power of $k$.

\subsection{Predict}
\label{app:p3-predict}

The prediction routine of $P^3$ is \algorithmref{alg:p3-predict}, invoked to
predict $x[n]$ after the state has been updated on $x[:n]$; it computes the
vote of \sectionref{sec:p3-predict}, with the edge weights $r^{-\eta}$, where
$\eta \geq 0$ is the unique solution of $\sum_{r \in R} r^{-\eta} = 1$.

Solving for $\eta$ is a one-dimensional root-finding problem. The function
$f(t) \coloneqq \sum_{r \in R} r^{-t}$ is continuous, strictly decreasing and
convex on $[0, \infty)$, with $f(0) = \abs{R}$ and $f(t) \to 0$ as $t \to
\infty$, so the root $\eta$ exists and is unique. Every arity being at least
$2$, $f(t) \leq \abs{R} 2^{-t}$, whence $f(\ceil{\log \abs{R}}) \leq 1 \leq
f(0)$: the interval $[0, \ceil{\log \abs{R}}]$ brackets the root, with integer
endpoints. Throughout the rest of this subsection, put $r_{\max} \coloneqq \max
R$.

The finite-precision implementation at the end of this subsection will need a
rational approximation $\tilde{\eta}$ of $\eta$, accurate to polynomially many
bits. The obstacle to producing one by textbook bisection is the sign test:
whether $f(t) > 1$ at a rational $t$ is the comparison with $1$ of an algebraic
number that can lie arbitrarily close to $1$ --- or equal it: for $R = \{2, 3,
6\}$, $\eta = 1$ --- and no efficient such test is apparent. The subroutine
Bisect$(\epsilon)$ (\algorithmref{alg:p3-bisect}), where $\epsilon$ is the
requested accuracy, therefore bisects on approximate values of $f$: in place of
an exact straddle of the root, its bracket maintains the relaxed straddle $f(l)
\geq 1 - \epsilon / 2$, $f(h) \leq 1 + \epsilon / 2$, which an evaluation of
$f$ to accuracy $\epsilon / 2$ suffices to preserve, and the loop runs until
the bracket is too short for $f$ to stray from $1$ anywhere on it.

\begin{algorithm}[htbp]
\floatconts{alg:p3-bisect}
{\caption{Bisect$(\epsilon)$}}
{\begin{alglines}
  \item $l \gets 0$;\hspace{0.5em} $h \gets \ceil{\log \abs{R}}$
  \item \textbf{while} $h - l > \epsilon / (2 \abs{R} \ceil{\log r_{\max}})$
  \item \algind{1} $t \gets (l + h) / 2$
  \item \algind{1} compute $\tilde{f} \in \mathbb{Q}$ with $\abs{\tilde{f} -
    f(t)} \leq \epsilon / 2$
  \item \algind{1} \textbf{if} $\tilde{f} \geq 1$ \textbf{then} $l \gets t$
    \textbf{else} $h \gets t$
  \item \textbf{return} $l$
\end{alglines}}
\end{algorithm}

Every comparison that Bisect performs is between rational numbers: the
endpoints, integers at line 1, remain dyadic rationals under the halving of
line 3, the threshold of line 2 is rational along with $\epsilon$, and the
estimate $\tilde{f}$ of line 4 is rational by specification. The irrational $f$
enters only through the approximate evaluation of line 4, a routine
computation, $f(t)$ being a finite sum of exponentials taken at a dyadic
rational. The validity of the procedure --- for every rational $\epsilon \in
(0, 1 / 2]$, it halts within polynomially many iterations and returns a
rational $\tilde{\eta} \geq 0$ with $\abs{\tilde{\eta} - \eta} \leq 3 \epsilon$
--- is proved as \lemmaref{lem:p3-bisect} (\sectionref{app:p3-exponent}).

\begin{algorithm}[htbp]
\floatconts{alg:p3-predict}
{\caption{$P^3$ Predict$()$}}
{\begin{alglines}
  \item $\{d_m\} \gets$ Diverge$()$
  \item $\mathcal{A} \gets \{(m, q) \mid m \in S, q \in Q_m, q \equiv n
    \pmod{m}\}$
  \item $\mathcal{L} \gets \{(m, q, r) \mid (m, q) \in \mathcal{A}, (m, r) \in
    \mathcal{E}, \delta_m^r(q, \floor{n / m} \bmod r) \not\equiv n \pmod{m r}\}$
  \item \textbf{for} $(m, q, r) \in \mathcal{L}$
  \item \algind{1} $\ell \gets \floor{n / m r}$;\hspace{0.5em} $q' \gets
    \delta_m^r(q, \floor{n / m} \bmod r)$
  \item \algind{1} $D \gets \{\tilde{q} \in Q_{m r} \mid d_{m r}(q',
    \tilde{q}) \geq \ell\}$
  \item \algind{1} \textbf{for} $\tilde{q} \in D$: \hspace{0.4em}
    $\nu(\tilde{q}) \gets \abs{\{(\hat{q}, i) \in Q_m \times [r] \mid
    \delta_m^r(\hat{q}, i) = \tilde{q},\ i m + \hat{q} < n \bmod m r\}}$
  \item \algind{1} $q^* \gets \argmax_{\tilde{q} \in D} \nu(\tilde{q})$
  \item \algind{1} $\phi(m, q, r) \gets$ Eval$(m r \ell + q^*)$
  \item \textbf{for} $m \in S$ in decreasing order
  \item \algind{1} \textbf{if} $\{r \in R \mid (m, r) \in \mathcal{E}\} =
    \emptyset$ \textbf{then} $\theta(m) \gets 1$ \textbf{else} $\theta(m) \gets
    \sum_{r \in R : (m, r) \in \mathcal{E}} r^{-\eta} \cdot \theta(m r)$
  \item \algind{1} \textbf{for} $(m, q) \in \mathcal{A}$ and $a \in \Sigma$
  \item \algind{2} $\begin{aligned}[t]
      \pi(m, q)[a] \gets{} & \sum_{(m, q, r) \in \mathcal{L} : \phi(m, q, r) =
        a} \, r^{-\eta} \cdot \theta(m r) \\
      & + \sum_{r \in R : (m, r) \in \mathcal{E},\, (m, q, r) \notin
        \mathcal{L}} \, r^{-\eta} \cdot \pi(m r, \delta_m^r(q, \floor{n / m}
        \bmod r))[a]
    \end{aligned}$
  \item \textbf{return} $\argmax_{a \in \Sigma} \pi(1, 0)[a]$
\end{alglines}}
\end{algorithm}

Let us unpack \algorithmref{alg:p3-predict}. Line 2 collects the \emph{active
pairs} and line 3 the \emph{leaf edges} of \sectionref{sec:p3-predict}: $(m, q)
\in \mathcal{A}$ when the residue $n \bmod m$ is itself a state of scale $m$,
in which case $q = n \bmod m$ (so each scale contributes at most one active
pair, and scale $1$ always contributes $(1, 0)$), and $(m, q, r) \in
\mathcal{L}$ when the routing inequality $\delta_m^r(q, i) \leq q + i m = n
\bmod m r$ (claim 2 of \lemmaref{lem:p3-inv}) is strict at the digit $i =
\floor{n / m} \bmod r$. Whether installed as a preprocessing default
(\sectionref{app:p3-preprocessing}) or by a rerouting of some earlier update
(\sectionref{app:p3-update}), the identification at a leaf edge is consistent
with all past observations but not determined by them. At these edges, lines
5--9 forecast the current symbol by a plurality vote among states, exactly as
in TP. The candidates $D$ are the states whose subsequences have not been
observed to diverge from that of the routed target $q'$ before the current row
$\ell$ (claim 3 of \lemmaref{lem:p3-inv}); in particular $q' \in D$, so $D \neq
\emptyset$. The vote $\nu(\tilde{q})$ counts the refined residues that the
model routes to $\tilde{q}$ and that have already been traversed within the
current row --- the residues $i m + \hat{q}$ modulo $m r$ with $m r \ell + (i m
+ \hat{q}) < n$. The forecast $\phi(m, q, r)$ is read off the winner $q^*$ via
Eval.

Lines 10--13 aggregate these forecasts over the whole fragment. Recall from
\sectionref{sec:p3-predict} that an infinite path from the root is
\emph{unblocked} if the fragment contains every edge that it traverses before
reaching a sink, and \emph{blocked} otherwise; since every scale below the
horizon retains an outgoing edge (\sectionref{app:p3-post}), a blocked path is
one that attempts an edge deleted by post-processing. Line 11 computes, by the
same traversal as Diverge, the \emph{unblocked mass} $\theta(m)$ of every scale
$m \in S$: the probability that the random walk, started at $m$, reaches a sink
of the fragment through edges of $\mathcal{E}$. At a sink nothing has been
blocked --- its outgoing edges are merely not created yet, and round $m$ will
create all of them at once (\sectionref{app:p3-preprocessing}) --- so
$\theta(m) = 1$; elsewhere $\theta(m)$ accumulates, over the steps that the
fragment still offers, the step weight times the unblocked mass of the target,
a deleted step contributing nothing. Like the divergence tables, $\theta$ is
recomputed at every call and is not part of the persistent state. The recursion
of line 13 is well-founded: at an edge with $(m, q, r) \notin \mathcal{L}$, the
transition target is $n \bmod m r$ itself, so $(m r, n \bmod m r)$ is an active
pair, whose vector $\pi$ was computed at an earlier iteration of line 10.
Unwinding, every unblocked path casts the forecast of the \emph{first} leaf
edge it encounters, and the votes are counted by mass: a leaf edge $(m, r)$
receives the weight $m^{-\eta} \cdot r^{-\eta} \theta(m r)$ through each path
of the fragment from the root to $m$ --- all of which consist of active pairs
and non-leaf edges, by heredity (claim 2 of \lemmaref{lem:p3-inv}) --- and
altogether the mass of the unblocked paths traversing it. Blocked paths are
thereby excluded from the vote even when the deleted edge lies beyond the leaf
edge at which they would have voted; this is the point of the factor $\theta(m
r)$, without which a leaf edge would carry the mass $(m r)^{-\eta}$ of
\emph{all} paths through it. The total mass $\theta(1)$ is accounted for in
$\pi(1, 0)$ except for the paths that encounter no leaf edge, which abstain.
This happens in exactly one way: when $n \in S$, the residue of $n$ modulo each
of its divisors is $0$, a state at every scale, so the descent stays exact ---
the digits of $n$ at these scales are $0$, and $\delta(0, 0) = 0$ by routing
--- all the way to the active pair $(n, 0)$, where it finds no outgoing edge,
the scale $n$ being still a sink when $x[n]$ is predicted (its edges are
created by the update of the same round). The rationale for the plurality rule
--- a coupling to the branch prediction game of \citep{Kosoy2026b} --- is
deferred to \sectionref{app:p3-stat}, and so is the reason for confining the
vote to the unblocked paths: the mistake analysis treats them as the experts of
a budgeted plurality vote (\sectionref{app:budgeted-plurality}), each able to
err only boundedly often, and it needs the set of voters to shrink
monotonically between restarts. It does, because an edge out of $m$ is created
only in the round $n = m$ (\sectionref{app:p3-preprocessing}) --- a deleted
edge therefore stays absent until the predictor restarts --- while the edges
that Preprocess creates lead into sinks, extending the unblocked paths without
blocking any. Finally, for $R = \{k\}$ the fragment is a single chain, $\eta =
0$ and $\theta \equiv 1$: $\pi(1, 0)$ is then the indicator vector of the
forecast at the chain's first leaf edge --- the breakpoint at which TP halts
its descent --- and \algorithmref{alg:p3-predict} essentially reduces to the
prediction routine of TP. The one difference arises at the rounds $n = k^j$,
where the chain ends at the pair $(n, 0)$ and the vote is empty:
\algorithmref{alg:p3-predict} returns an arbitrary symbol, while TP, whose
descent stops one level short of the maximal scale, forecasts there as at any
other breakpoint --- at most one mistake per power of $k$ either way.

\textbf{Finite precision.} For $\abs{R} \geq 2$, the exponent $\eta$ and the
weights $r^{-\eta}$ admit, in general, no exact finite representation, so we
read lines 10--13 of \algorithmref{alg:p3-predict} as defining ideal real
quantities $\theta$ and $\pi$. The implementation runs lines 1--9 unchanged ---
every quantity there is an integer or a symbol, computed exactly --- and
replaces lines 10--14 by \algorithmref{alg:p3-predict-fp}, which computes
fixed-point approximations $\hat{\theta}$ and $\hat{\pi}$ in place of $\theta$
and $\pi$ and returns a symbol maximizing $\hat{\pi}(1, 0)$, ties broken
arbitrarily. There, $n$ is the time index of the round, the \emph{unit} of the
arithmetic is the $u = 2^{-F}$ of line 1, and $\floor{x}_u \coloneqq u \floor{x
/ u}$ is $x$ rounded down to a multiple of $u$ --- so every product of lines 5
and 7 carries a rounding error smaller than $u$, while the sums, of multiples
of $u$, are exact.

\begin{algorithm}[htbp]
\floatconts{alg:p3-predict-fp}
{\caption{$P^3$ Predict$()$, lines 10--14 as implemented}}
{\begin{alglines}
  \item $F \gets (\ceil{\log r_{\max}} + 1) \cdot \ceil{\log(2 r_{\max} n)} +
    2$;\hspace{0.5em} $u \gets 2^{-F}$;\hspace{0.5em} $\epsilon \gets u / (16
    \ceil{\log r_{\max}})$
  \item $\hat{\eta} \gets$ Bisect$(\epsilon) + 3 \epsilon$
  \item \textbf{for} $r \in R$: \hspace{0.4em} compute $\rho_r \in [0,
    r^{-\hat{\eta}}] \cap \mathbb{Q}$ with $r^{-\hat{\eta}} - \rho_r \leq
    2^{-F - 1}$
  \item \textbf{for} $m \in S$ in decreasing order
  \item \algind{1} \textbf{if} $\{r \in R \mid (m, r) \in \mathcal{E}\} =
    \emptyset$ \textbf{then} $\hat{\theta}(m) \gets 1$ \textbf{else}
    $\hat{\theta}(m) \gets \sum_{r \in R : (m, r) \in \mathcal{E}}
    \floor{\rho_r \cdot \hat{\theta}(m r)}_u$
  \item \algind{1} \textbf{for} $(m, q) \in \mathcal{A}$ and $a \in \Sigma$
  \item \algind{2} $\begin{aligned}[t]
      \hat{\pi}(m, q)[a] \gets{} & \sum_{(m, q, r) \in \mathcal{L} : \phi(m, q,
        r) = a} \, \floor{\rho_r \cdot \hat{\theta}(m r)}_u \\
      & + \sum_{r \in R : (m, r) \in \mathcal{E},\, (m, q, r) \notin
        \mathcal{L}} \, \floor{\rho_r \cdot \hat{\pi}(m r, \delta_m^r(q,
        \floor{n / m} \bmod r))[a]}_u
    \end{aligned}$
  \item \textbf{return} $\argmax_{a \in \Sigma} \hat{\pi}(1, 0)[a]$
\end{alglines}}
\end{algorithm}

Lines 1--3 calibrate the arithmetic; \lemmaref{lem:p3-fp-accuracy}
(\sectionref{app:p3-exponent}) certifies the accuracy achieved: $F = O(\log
r_{\max} \cdot \log(r_{\max} n))$, the $\hat{\eta}$ of line 2 is rational with
$0 \leq \hat{\eta} - \eta \leq 6 \epsilon$, and every \emph{rounded weight}
$\rho_r$ of line 3 underestimates its ideal counterpart by at most the unit ---
$\rho_r \in [0, r^{-\eta}]$ with $r^{-\eta} - \rho_r \leq u$. All the numbers
involved can be taken with $O(F)$ bits, so the time of the round remains
polynomial in the size of the state; and the rounding is fully accounted for in
the mistake analysis: by \lemmaref{lem:p3-precision}, the returned symbol
maximizes the ideal vote $\pi(1, 0)$ up to an additive error of $(r_{\max}
n)^{-\eta} / 2$, a tolerance which the budgeted plurality vote of
\sectionref{app:budgeted-plurality} is designed to absorb.

\subsection{Update}
\label{app:p3-update}

Upon receiving $a = x[n]$, the algorithm runs Update
(\algorithmref{alg:p3-update}); together with Predict, it realizes the protocol
of \sectionref{sec:protocol} --- Update is the state-update function
$\mathcal{U}$ and Predict the state-prediction function $\mathcal{P}$, up to
the binary encoding of the state --- where $n$ denotes the time index stored in
the state (\sectionref{sec:p3-model}), which line 19 advances at the end of
every round. Line 1 dispatches the empty state to Initialize
(\sectionref{app:p3-init}), and line 18 hands the updated fragment to
post-processing (\sectionref{app:p3-post}); in between, after Preprocess
(\algorithmref{alg:p3-preprocess}) has extended the fragment and Diverge has
refreshed the divergence tables (line 2), the pass of lines 3--17 performs the
repairs described in \sectionref{sec:p3-update}. (In a round that fires line 1,
the rest of the pass is vacuous: Preprocess finds $n = 0 \notin S$, the loop of
line 3 has no scale to visit, and post-processing finds nothing to delete; line
19 then ends the round with the time index at $1$.)

\begin{algorithm}[htbp]
\floatconts{alg:p3-update}
{\caption{$P^3$ Update$(a)$}}
{\begin{alglines}
  \item \textbf{if} $S = \emptyset$ \textbf{then} Initialize$(a)$
  \item Preprocess$()$;\hspace{0.5em} $\{d_m\} \gets$ Diverge$()$
  \item \textbf{for} $m \in S \setminus \{1\}$ in increasing order
  \item \algind{1} $\mathcal{I} \gets \{(m', r) \in \mathcal{E} \mid m' r = m,
    n \bmod m' \in Q_{m'}\}$;\hspace{0.5em} $\ell \gets \floor{n / m}$
  \item \algind{1} \textbf{if} $n \bmod m \notin Q_m$
  \item \algind{2} \textbf{for} $(m', r) \in \mathcal{I}$
  \item \algind{3} $q' \gets \delta_{m'}^r(n \bmod m', \floor{n / m'} \bmod r)$
  \item \algind{3} \textbf{if} Eval$(m \ell + q') \neq a$
  \item \algind{4} \textbf{if} there is $\tilde{q} \in Q_m$ with $d_m(\tilde{q},
    q') \geq \ell$ and Eval$(m \ell + \tilde{q}) = a$
  \item \algind{5} $\delta_{m'}^r(n \bmod m', \floor{n / m'} \bmod r) \gets$
    the minimal such $\tilde{q}$
  \item \algind{4} \textbf{else}
  \item \algind{5} $Q_m \gets Q_m \cup \{n \bmod m\}$
  \item \algind{5} \textbf{if} $m > n$ \textbf{then} $\tau_m(n) \gets a$
  \item \algind{5} \textbf{for} $(m, r') \in \mathcal{E}$ and $j \in [r']$:
    \hspace{0.4em} $\delta_m^{r'}(n \bmod m, j) \gets \delta_m^{r'}(q', j)$
  \item \algind{5} \textbf{break}
  \item \algind{1} \textbf{if} $n \bmod m \in Q_m$
  \item \algind{2} \textbf{for} $(m', r) \in \mathcal{I}$: \hspace{0.4em}
    $\delta_{m'}^r(n \bmod m', \floor{n / m'} \bmod r) \gets n \bmod m$
  \item Post-process$(a)$
  \item $n \gets n + 1$
\end{alglines}}
\end{algorithm}

\textbf{Detection.} Line 4 collects the incoming edges of the scale $m$ whose
source pair is active --- $(m', r) \in \mathcal{E}$ with $m' r = m$ and $n
\bmod m' \in Q_{m'}$, the edge-level counterpart of the active pairs of
\algorithmref{alg:p3-predict} --- together with the row $\ell = \floor{n / m}$
that the current position occupies at scale $m$. The set $\mathcal{I}$ is
evaluated only when the loop reaches $m$: a source residue may have entered
$Q_{m'}$ through a state creation at an earlier iteration, and this is the
mechanism by which a repair propagates through the fragment (see below). In the
branch of line 5 the residue $n \bmod m$ is not a state, so every edge of
$\mathcal{I}$ routes the residue $n \bmod m$ to some state $q' < n \bmod m$
(claim 2 of \lemmaref{lem:p3-inv}) --- a tentative identification of the
residue with the smaller state; in the terminology of
\sectionref{app:p3-predict}, a leaf edge. Row $\ell$ of the subsequence of $q'$
is then the position $m \ell + q' < m \ell + (n \bmod m) = n$: already
observed, hence evaluated identically by every run of Eval (claim 1). Line 8
tests this recorded value against the fresh observation, and the
identification is refuted --- calling for a repair at this edge --- exactly
when the two differ. (The divergence tables of line 2 stay valid throughout the
pass: before the loop reaches $m$, every table entry touched and every state
created belongs to a scale smaller than $m$, while $d_m$ is determined by the
transitions out of the scales at least $m$ and the state sets at their sinks.)

\textbf{Correction by rerouting.} The preferred remedy (lines 9--10) is to
reroute the refuted transition to another state of $Q_m$, subject to two
compatibility conditions, as in TP. The divergence condition $d_m(\tilde{q},
q') \geq \ell$ protects the past: by minimality of representatives, every
position evaluated through the modified entry has residue at least $n \bmod m$,
so the observed ones lie at rows below $\ell$ --- exactly the rows on which the
subsequences of $\tilde{q}$ and $q'$ have been evaluating identically
(\sectionref{sec:p3-model}). The same condition confines the candidates to
observed rows: by distinguishability (claim 3 of \lemmaref{lem:p3-inv}, applied
to the state after the previous round --- preprocessing does not disturb it),
the divergence value of a pair of distinct states points at an observed
difference, $m \cdot d_m(\tilde{q}, q') + \tilde{q} \leq n - 1$, so
$d_m(\tilde{q}, q') \geq \ell$ forces $\tilde{q} < n \bmod m$. In particular a
committed target never exceeds the residue it represents, as the routing
invariant requires, and row $\ell$ of every candidate's subsequence is the past
position $m \ell + \tilde{q} < n$, on which every run of Eval agrees (claim 1).
The second condition, Eval$(m \ell + \tilde{q}) = a$, is therefore the single
comparison $x[m \ell + \tilde{q}] = a$: it certifies that the candidate's
subsequence carries the new observation at the current row. Among the states
satisfying both conditions, line 10 commits the minimal one, following the
minimal-consistent-transition principle of \sectionref{app:p3-preprocessing}.

We call the routing property \emph{hereditary} when, for every edge $(p, r) \in
\mathcal{E}$ and every state $u \in Q_{p r}$, the reduced residue $u \bmod p$
belongs to $Q_p$. Equivalently, every state at an edge's target has its parent
residue stored at the edge's source. We refer to this implication as
\emph{heredity} below.

\textbf{Correction by state creation.} When no state of $Q_m$ passes both
tests, the data has separated the subsequence of $n$ from the subsequence of
every state compatible with $q'$, and line 12 installs its canonical
representative $n \bmod m$. An output value is recorded (line 13) only beyond
the horizon, where $n \bmod m = n$ and the symbol is the fresh observation
itself; as observed in \sectionref{app:p3-preprocessing}, these are precisely
the creations at the sinks, and the assignment keeps the states of every sink
the first occurrences of the distinct observed symbols, as claim 4 of
\lemmaref{lem:p3-inv} requires. A creation at a scale below the horizon ---
that is, $m < n$; the case $m = n$ cannot arise, as the residue $n \bmod n = 0$
is always a state --- leaves $\tau_m$ untouched, because the new entry could
never be consulted. Indeed, outputs are read in exactly two places: by Eval, at
the halting scale of its walk, which is a sink and hence beyond the horizon;
and by Preprocess, which copies $\tau_m$ onto fresh target scales in the round
$n = m$ itself. A scale below the time index never returns beyond the horizon,
and never becomes a sink either --- it retains an outgoing edge, or is deleted
by post-processing (\sectionref{app:p3-post}) --- so neither reader ever
revisits it. The output functions are thus maintained, and constrained by
\lemmaref{lem:p3-inv}, only beyond the horizon. The outgoing transitions of the
new state (line 14) are inherited verbatim from its former representative $q'$,
which keeps every past evaluation through the new state intact. Verbatim
inheritance is sound only if no refined residue $(n \bmod m) + j m$ is already
a state at the target scale --- such a state would demand the entry route to
it (claim 2), not to $q'$'s target. This is exactly the heredity part of claim
2: a state $u \in Q_{m r'}$ has $u \bmod m \in Q_m$, whereas $n \bmod m$ enters
$Q_m$ only now. In a chain, heredity is immediate --- a state refining the
residue $n \bmod m$ could only have been created at a round at which $n \bmod m
\in Q_m$ already held. In the scale fragment, a state refining $n \bmod m$
might seem obtainable through a \emph{different} predecessor of $m r'$, but the
update itself forestalls this: at such a round, the residue $n \bmod m$ is
still identified with a smaller representative, whose entries supply a
candidate passing both tests of line 9 --- by data consistency, the recorded
rows of that candidate's subsequence are the very data of the subsequence being
separated --- so the update reroutes rather than creates. Consequently, a state
is created at a scale only when every edge of $\mathcal{E}$ into it has an
active source; this is proved together with \lemmaref{lem:p3-inv} in
\appendixref{app:p3-correctness}. The same transfer of evidence operates
between the edges of $\mathcal{I}$ themselves: had an earlier edge passed the
test of line 8 or committed a reroute, its target would pass both tests of line
9 at every later edge --- the subsequences of sibling targets agree below row
$\ell$, by data consistency through the respective routes --- and the search
would succeed. A creation therefore occurs, if at all, at the first edge
examined, before any sibling edge has acted, and only when the recorded
row-$\ell$ value of every sibling target's subsequence is likewise refuted.
Line 15 then abandons the remaining edges of $\mathcal{I}$, deferring them to
line 17.

\textbf{Restoring routing equality.} Once $n \bmod m \in Q_m$ --- whether from
an earlier round or from line 12 --- the routing invariant demands equality: at
every edge of $\mathcal{I}$, the entry refining the residue of $n$ must land on
the state $n \bmod m$ exactly, and line 17 enforces this. When the state
predates the round, the assignments are no-ops: by heredity, every source
residue $n \bmod m'$ was in $Q_{m'}$ after the previous round as well, and
claim 2 already places its transition at $n \bmod m$. The assignments do real
work only in the wake of line 12, committing the creating edge together with
the edges that line 15 abandoned --- by the first-edge property above, this is
all of $\mathcal{I}$, none of it yet modified in this round. Their soundness is
uniform: every observed position evaluated through such an entry sits at a row
$\ell' < \ell$ (the minimality argument above), and at these rows nothing
changes --- a run of Eval at the observed position $m \ell' + (n \bmod m) < n$
can be routed through any edge of $\mathcal{I}$, so by claim 1 the subsequences
of the old target and of $n$ carry the same recorded symbol $x[m \ell' + (n
\bmod m)]$ at row $\ell'$. These commits could not have been left to the search
of line 9: at the candidate $n \bmod m$, the forecast test evaluates the
position $m \ell + (n \bmod m) = n$ itself, and mid-pass Eval$(n)$ still
returns the refuted value --- the repairs that will turn it to $a$ belong to
later iterations (see below).

\textbf{The cascade.} A creation at an interior scale does not restore
consistency at $n$ by itself: the inherited transitions make every run of
Eval$(n)$ produce the same refuted values as before. What it does do is turn
the residue $n \bmod m$ into a state, injecting the edges out of $m$ into the
sets $\mathcal{I}$ of later iterations, where the inherited transitions face
line 8 in their turn. A single mismatch thus cascades upward along the residues
of $n$ --- through every branch of the fragment it activates --- until it is
resolved by rerouting, or anchored at the sinks, where the created states
record the fresh symbol (line 13). The two resolutions are the alternatives of
claim 4 of \lemmaref{lem:p3-inv} at the sinks, where $\ell = 0$ and the
divergence condition is vacuous: when $a$ has occurred before, the minimal
compatible state is its recorded first occurrence and rerouting succeeds; when
$a$ is a fresh symbol, no stored value and hence no forecast anywhere in the
fragment can equal it, every test of line 9 fails, the cascade sweeps the
entire fragment, and every sink acquires the first-occurrence state $n$ with
output $a$. Finally, for $R = \{k\}$ the fragment is a chain: each $\mathcal{I}$
is a singleton, line 17 merely commits the transition to a freshly created
state, and, post-processing aside --- TP keeps every state it ever creates and
has no counterpart of the capacity control --- \algorithmref{alg:p3-update}
essentially reduces to the update routine of TP.

\subsection{Post-Processing}
\label{app:p3-post}

After the update cascade, post-processing --- the subroutine Post-process
(\algorithmref{alg:p3-post}), invoked at line 18 of
\algorithmref{alg:p3-update} with the symbol observed in the current round ---
performs the deletions and the restart described in
\sectionref{sec:p3-update}.

\begin{algorithm}[htbp]
\floatconts{alg:p3-post}
{\caption{$P^3$ Post-process$(a)$}}
{\begin{alglines}
  \item $\mathcal{D} \gets \{m \in S \mid \abs{Q_m} > c_{\max}\}$
  \item \textbf{while} $\mathcal{D} \neq \emptyset$
  \item \algind{1} $S \gets S \setminus \mathcal{D}$;\hspace{0.5em}
    $\mathcal{E} \gets \{(m, r) \in \mathcal{E} \mid m \in S, m r \in S\}$
  \item \algind{1} $\mathcal{D} \gets \{m \in S \setminus \{1\} \mid \forall
    (m', r) \in \mathcal{E} : m' r \neq m\}$ \\
    \algind{2} $\cup \{m \in S \mid m \leq n, \forall r \in R : (m, r) \notin
    \mathcal{E}\}$
  \item \textbf{if} $S = \emptyset$ \textbf{then} $c_{\max} \gets 2
    c_{\max}$;\hspace{0.5em} Initialize$(a)$
\end{alglines}}
\end{algorithm}

\textbf{Deletion by capacity.} The capacity parameter is consulted only here
--- neither Predict nor the update cascade reads $c_{\max}$ --- and a state,
once created, is never deleted on its own: the unit of deletion is the scale,
with all of its attached data. Removing $m$ from $S$ discards $Q_m$ and
$\tau_m$, and the restriction of $\mathcal{E}$ in line 3 discards the
transition tables of the incident edges (\sectionref{sec:p3-model}). The sweep
of line 1 examines the trail of the cascade just completed: at the end of the
previous round every scale satisfied $\abs{Q_m} \leq c_{\max}$, preprocessing
only copies the state set of an existing scale (line 3 of
\algorithmref{alg:p3-preprocess}), and line 12 of \algorithmref{alg:p3-update}
fires at most once per scale, so line 1 collects exactly the scales at which
the current cascade pushed the state count to $c_{\max} + 1$. Deleting them
restores the bound $\abs{Q_m} \leq c_{\max}$ throughout the fragment by the
time the round ends --- the uniform capacity on which the per-path mistake
analysis of \sectionref{app:p3-stat} relies.

\textbf{Deletion by structure.} The deletions of line 3 may leave the fragment
malformed, and line 4 collects the scales violating either of two structural
requirements, to be deleted in the next pass; the loop runs until no violation
remains. The first set of line 4 restores the requirement of
\sectionref{sec:p3-model} that every non-root scale be the target of an edge;
since the source of an edge is a smaller scale than its target, iterating this
criterion prunes exactly the scales no longer reachable from the root --- among
unreachable scales, a minimal one is the target of no edge at all --- so that
every surviving scale can still be visited by a run of Eval. The second set
deletes every scale at or below the time index that has no outgoing edges,
restoring the sink condition of claim 4 of \lemmaref{lem:p3-inv}: sinks are
where the runs of Eval halt and read the output functions, and where the
recursion of Diverge is grounded (line 2 of \algorithmref{alg:diverge}), and
both readings are sound only beyond the horizon, where the states are the first
occurrences of the distinct observed symbols and the output functions are
maintained. This is the deletion promised in \sectionref{app:p3-update}: a
scale overtaken by the time index either retains an outgoing edge or is
deleted, and its output function is never consulted again. The loop terminates,
as each pass strictly shrinks $S$; and since a deletion removes edges and never
adds any, it can create new violations but cannot cure one, so a scale
collected at line 4 stays violating until deleted --- the surviving fragment is
thus the unique maximal subset of $S$ satisfying both requirements, whichever
order or grouping the deletions are performed in. Note also that when line 1
collects nothing, post-processing deletes nothing at all: the cascade never
removes an edge, and preprocessing adjoins a scale only together with an edge
into it, while equipping the scale reached by the time index, if there is one,
with its outgoing edges in the same round
(\sectionref{app:p3-preprocessing}); the criteria of line 4 can thus first fire
only in the wake of a capacity deletion.

\textbf{Consistency across deletions.} Pruning modifies no surviving table, so
routing and heredity (claim 2) and the description of the surviving scales
beyond the horizon (claim 4) are inherited directly. Data consistency (claim 1)
persists by a truncation argument. A run of Eval$(t)$, $t \leq n$, on the
pruned state selects edges that the unpruned state offered as well, and halts
at a pair $(m, q)$ whose scale is a sink of the pruned fragment --- beyond the
horizon, by the sink condition just restored. On the unpruned state, the run
could have continued past $m$, but without ever changing state: the digit of
$t$ at any scale $m' > n$ is $\floor{t / m'} \bmod r = 0$ whatever the arity,
and $\delta_{m'}^r(q, 0) = q$, since by routing $\delta_{m'}^r(q, 0) \leq q + 0
\cdot m'$, with equality because $q$ belongs to $Q_{m' r}$ --- the state sets of
the scales beyond the horizon all being equal (claim 4). Every continuation
would therefore have halted at some sink $M$ of the unpruned fragment still at
the state $q$, returning $\tau_M(q) = x[q] = \tau_m(q)$ (claim 4 at both
scales); the truncated run thus returns the same symbol as the runs it
truncates --- $x[t]$, by data consistency before the deletions.
Distinguishability (claim 3) requires one extra step: deleting edges changes
the set of continuations over which Diverge minimizes, but pinning turns the
old observed witness into a witness along every surviving continuation. Hence
the new divergence occurs no later than the old one, which preserves the
required bound; see \sectionref{app:p3-pruning-proof}.

\textbf{The restart.} The root is deleted only last: it is immune to the sweep
of line 1 --- $\abs{Q_1} = 1 \leq c_{\max}$, the residue $0$ being the only
residue modulo $1$ --- and to the first criterion of line 4, so it is deleted
only upon losing all of its outgoing edges, at a round $n \geq 1$ (at round
$0$, the sole scale $1$ is itself beyond the horizon). And nothing outlives the
root: among any surviving scales, a minimal one would be a non-root scale that
is the target of no edge --- a violation the loop does not leave behind. Line 5
therefore fires exactly when the deletions have consumed the root, and the
predictor restarts: the capacity is doubled, and Initialize
(\algorithmref{alg:p3-init}) re-creates the minimal nonempty state from the
symbol observed in the current round, re-indexing positions so that the current
round plays the part of round $0$ for the remaining suffix
(\sectionref{app:p3-init}); the shared increment of line 19 of
\algorithmref{alg:p3-update} then ends the round with the time index at $1$,
exactly as after a true round $0$. The restarts divide the processing of an
input into \emph{epochs}: an epoch begins at a round at which Initialize runs
--- round $0$, by line 1 of \algorithmref{alg:p3-update}, or a round whose
post-processing fires line 5 --- and extends up to and including the next such
round, whose prediction and update cascade are still computed from the states
of the epoch, only the re-initialization at its end belonging to the next
epoch. Thus the epoch beginning at round $b$ processes the suffix of the input
from position $b$ as a fresh sequence, its first symbol being the one recorded
by Initialize; and the capacity is constant throughout an epoch, being doubled
exactly when an epoch ends in a restart, so that the $i$-th epoch, $i = 0, 1,
\dots$, runs with $c_{\max} = 2^i$. For an illustration, suppose the number of
distinct observed symbols exceeds $c_{\max}$ for the first time. Every scale
beyond the horizon holds one state per distinct symbol (claim 4), so the sweep
of line 1 deletes all of them at once, and the loop then unravels the remainder
--- its maximal scale is, at every pass, a sink at or below the time index ---
forcing a restart; indeed, no nonempty fragment obeying both the cap and the
sink condition exists until the cap is raised. Finally, for $R = \{k\}$,
post-processing is all-or-nothing. In a chain, deleting any scale disconnects
every scale above it and leaves the scale below it with no outgoing edges ---
at or below the time index, the next-to-maximal scale of the chain being
precisely the round at which the maximal one was adjoined --- so a capacity
violation anywhere unravels the whole chain. TP, which keeps every state it
ever creates (\sectionref{app:p3-update}), has no counterpart of this routine:
for $R = \{k\}$, \algorithmref{alg:p3-post} amounts to the doubling
meta-algorithm wrapped around TP, and the capacity control it enforces --- with
the restarts it forces --- is the price of aggregating, in the general case,
over all the maximal paths of the fragment (\sectionref{app:p3-stat}).

\section{Model Correctness}
\label{app:p3-correctness}

The proof of \theoremref{thm:p3} is divided between this appendix and
\appendixref{app:p3-proof}. This appendix proves \lemmaref{lem:p3-inv} --- the
invariants of the persistent state, which fix the semantics of the model.
\appendixref{app:p3-proof} then bounds the number of states at every scale in
terms of $s = \LZLC_R(x)$ (\sectionref{app:p3-state-bound}), and proves the
mistake bound \equationref{eq:p3-mistakes} (\sectionref{app:p3-stat}) and the
compression bound \equationref{eq:p3-state}
(\sectionref{app:p3-compression}).

Recall the indexing convention of \lemmaref{lem:p3-inv}, which is local to an
epoch (\sectionref{app:p3-post}): if an epoch begins at position $b$ of the
original input, we re-index the suffix $x[b], x[b+1], \dots$ from $0$ and again
denote it by $x$. Thus the $n$ in the lemma is the last observed index of the
current epoch; immediately after that update, the counter stored by the
algorithm is $n + 1$. A restart begins a new epoch and hence changes this
re-indexing.

The proof must account for two features absent from the chain-shaped state of
TP: a scale can have several incoming edges, and Eval may continue from a scale
along several different paths. We therefore split the proof into four parts.
First we give a pathwise semantics to Eval and Diverge. We then establish a
strengthened invariant for the unpruned update cascade, deduce from it the four
claims of \lemmaref{lem:p3-inv} for the configuration that the cascade hands to
post-processing, and finally verify that post-processing and restarts preserve
them.

\subsection{Pathwise Evaluation and Divergence}
\label{app:p3-path-semantics}

Fix a nonempty scale fragment $(S, \mathcal{E})$. A \emph{continuation} from $m
\in S$ is a sequence
\[
  \gamma = (m_0, r_0, m_1, r_1, \dots, r_{L-1}, m_L)
\]
such that $m_0 = m$, $(m_j, r_j) \in \mathcal{E}$ and $m_{j+1} = m_j r_j$ for
every $j < L$, and $m_L$ is a sink. We denote the set of continuations from $m$
by $\mathcal{C}_m$. It is nonempty: scales strictly increase along edges, and
the fragment is finite, so every directed path can be extended to a sink.

An \emph{entry path} to $m$ is a directed path in the scale fragment from $1$
to $m$. If $\alpha$ is an entry path to $m$ and $\gamma \in \mathcal{C}_m$, we
write $\alpha \gamma \in \mathcal{C}_1$ for their concatenation. Every $\zeta
\in \mathcal{C}_1$ fixes all the choices in Eval; for $t \in \mathbb{N}$, we
denote the resulting output by $\Eval_\zeta(t)$. Thus $\Eval_{\alpha
\gamma}(t)$ decomposes the pathwise evaluation according to its entry path
$\alpha$ to $m$ and continuation $\gamma$ from $m$.

\textbf{Routing along an entry path.} Let
\[
  \alpha = (a_0, s_0, a_1, s_1, \dots, s_{J-1}, a_J)
\]
be an entry path to $m$, so $a_0 = 1$, $a_J = m$, $(a_j, s_j) \in \mathcal{E}$
and $a_{j+1} = a_j s_j$ for $j < J$. For $u \in [m]$, define its \emph{route
along} $\alpha$, denoted $\rho_\alpha(u) \in Q_m$, by setting $p_0 \coloneqq 0$
and
\begin{equation}
  \label{eq:p3-route}
  p_{j+1} \coloneqq \delta_{a_j}^{s_j}(p_j, \floor{u / a_j} \bmod s_j), \quad
  \rho_\alpha(u) \coloneqq p_J.
\end{equation}
Thus $\rho_\alpha(u)$ is the endpoint of the transition recursion along
$\alpha$ on the mixed-radix digits of $u$.

For $\gamma \in \mathcal{C}_m$, $q \in Q_m$ and $\ell \in \mathbb{N}$, define
the \emph{pathwise partial evaluation} $\PEval_\gamma(q, \ell)$ as follows. Set
$q_0 \coloneqq q$ and $\ell_0 \coloneqq \ell$, and, for $j < L$, set
\[
  q_{j+1} \coloneqq \delta_{m_j}^{r_j}(q_j, \ell_j \bmod r_j), \quad \ell_{j+1}
  \coloneqq \floor{\ell_j / r_j}.
\]
Then
\begin{equation}
  \label{eq:p3-peval}
  \PEval_\gamma(q, \ell) \coloneqq \tau_{m_L}(q_L).
\end{equation}
Thus $\PEval_\gamma(q, \ell)$ is the terminal output of the transition
recursion initialized at state $q$ and row $\ell$ of scale $m$ and indexed by
the continuation $\gamma$.

The following elementary facts about the route map will be used repeatedly.

\begin{lemma}
\label{lem:p3-route-map}
Let $\alpha$ be an entry path to $m$.
\begin{enumerate}
  \item If the upper-bound part of routing holds along $\alpha$, then
    $\rho_\alpha(u) \leq u$ for every $u \in [m]$.
  \item If the full routing claim of \lemmaref{lem:p3-inv} holds along
    $\alpha$, then $\rho_\alpha(q) = q$ for every $q \in Q_m$.
  \item For every continuation $\gamma \in \mathcal{C}_m$, $u \in [m]$ and
    $\ell \in \mathbb{N}$,
    \begin{equation}
      \label{eq:p3-route-factor}
      \Eval_{\alpha \gamma}(m \ell + u) = \PEval_\gamma(\rho_\alpha(u), \ell),
    \end{equation}
\end{enumerate}
\end{lemma}

\begin{proof}
For claim 1, we prove by induction that $p_j \leq u \bmod a_j$. This is an
equality for $j = 0$. If $i \coloneqq \floor{u / a_j} \bmod s_j$, then the
upper-bound property and the induction hypothesis give
\[
  p_{j+1} = \delta_{a_j}^{s_j}(p_j, i) \leq p_j + i a_j \leq (u \bmod a_j) + i
  a_j = u \bmod a_{j+1}.
\]
At $j = J$ this is $\rho_\alpha(u) \leq u$, since $u < m = a_J$.

For claim 2, apply heredity backwards along $\alpha$: $q \bmod a_j \in
Q_{a_j}$ for every $j$. Starting with $p_0 = q \bmod 1 = 0$, the equality part
of routing then gives inductively $p_j = q \bmod a_j$, and hence
$\rho_\alpha(q) = q$.

For claim 3, let $t \coloneqq m \ell + u$. Along $\alpha$, the digit used at
the edge $(a_j, s_j)$ is
\[
  \floor{t / a_j} \bmod s_j = \floor{u / a_j} \bmod s_j,
\]
because $a_{j+1}$ divides $m$. Hence the transition recursion along $\alpha$ is
exactly \equationref{eq:p3-route} and has endpoint $\rho_\alpha(u)$. The
successive digits and remaining rows after $m$ are exactly those in the
recursion \equationref{eq:p3-peval}, which proves
\equationref{eq:p3-route-factor}.
\end{proof}

For $q, q' \in Q_m$, let
\[
  D_\gamma(q, q') \coloneqq \min \{\ell \in \mathbb{N} \mid \PEval_\gamma(q,
  \ell) \neq \PEval_\gamma(q', \ell)\},
\]
where the minimum of the empty set is $\infty$. The next lemma gives the exact
semantics of Diverge. Its hypothesis on sinks is an immediate consequence of
the horizon invariant: at a sink, distinct states are first occurrences of
distinct symbols and hence have distinct outputs.

\begin{lemma}
\label{lem:p3-div-semantics}
Suppose that $\tau_M$ is injective on $Q_M$ for every sink $M$. Then the values
computed by Diverge satisfy
\begin{equation}
  \label{eq:p3-div-path}
  d_m(q, q') = \min_{\gamma \in \mathcal{C}_m} D_\gamma(q, q')
\end{equation}
for every $m \in S$ and $q, q' \in Q_m$. Consequently:
\begin{enumerate}
  \item if $d_m(q, q') \geq L$, then $\PEval_\gamma(q, \ell) =
    \PEval_\gamma(q', \ell)$ for every $\gamma \in \mathcal{C}_m$ and every
    $\ell < L$;
  \item if $d_m(q, q') < \infty$, then some $\gamma \in \mathcal{C}_m$
    satisfies $\PEval_\gamma(q, d_m(q, q')) \neq \PEval_\gamma(q', d_m(q,
    q'))$.
\end{enumerate}
\end{lemma}

\begin{proof}
We prove \equationref{eq:p3-div-path} by downward induction over the scale
fragment, in the same topological order used by \algorithmref{alg:diverge}. If
$m$ is a sink, $\mathcal{C}_m$ contains only the length-zero continuation. For
$q = q'$ its path divergence is $\infty$, while for $q \neq q'$ it is $0$ by
the injectivity of $\tau_m$. These are exactly the values assigned in line 2 of
\algorithmref{alg:diverge}.

Now let $m$ be a non-sink and assume the identity at every target $m r$ of an
edge out of $m$. Every $\gamma \in \mathcal{C}_m$ consists of a first edge $(m,
r)$ followed by some $\gamma' \in \mathcal{C}_{m r}$. For this fixed
continuation, partitioning the rows according to $\ell = r \ell' + i$ gives
\[
  D_\gamma(q, q') = \min_{i < r} (r \cdot D_{\gamma'}(\delta_m^r(q, i),
  \delta_m^r(q', i)) + i),
\]
with the convention $r \cdot \infty + i \coloneqq \infty$. Indeed, at a row $r
\ell' + i$, the first transition uses digit $i$, and the remaining row at scale
$m r$ is $\ell'$; this is exactly the recursion in
\equationref{eq:p3-peval}. Taking the minimum first over all tails $\gamma' \in
\mathcal{C}_{m r}$ and then over all outgoing edges, and applying the induction
hypothesis, yields
\[
  \min_{\gamma \in \mathcal{C}_m} D_\gamma(q, q') = \min_{r \in R : (m, r) \in
  \mathcal{E}} \min_{i < r} (r \cdot d_{m r}(\delta_m^r(q, i), \delta_m^r(q',
  i)) + i),
\]
which is line 6 of \algorithmref{alg:diverge}. This proves
\equationref{eq:p3-div-path}. The two consequences follow immediately: the
right-hand side is the earliest row at which \emph{some} continuation separates
the pair, so every continuation agrees at all earlier rows, and a finite
minimum is attained because the fragment and all of its transition tables are
finite.
\end{proof}

We next specialize the route factorization to stored states. The structural
conditions maintained by the algorithm guarantee that every surviving scale has
an entry path. The following observation explains the role of both clauses of
the routing invariant.

\begin{lemma}
\label{lem:p3-peval-lift}
Assume the routing claim of \lemmaref{lem:p3-inv}, and let $m \in S$, $q \in
Q_m$, $\ell \in \mathbb{N}$, $\gamma \in \mathcal{C}_m$. For every entry path
$\alpha$ to $m$,
\[
  \Eval_{\alpha \gamma}(m \ell + q) = \PEval_\gamma(q, \ell).
\]
\end{lemma}

\begin{proof}
Claims 2 and 3 of \lemmaref{lem:p3-route-map} give
\[
  \Eval_{\alpha \gamma}(m \ell + q) = \PEval_\gamma(\rho_\alpha(q), \ell) =
  \PEval_\gamma(q, \ell).
\]
\end{proof}

Combined with data consistency, the lift pins the partial evaluations of
stored states to the data itself, wherever the positions involved have already
been observed. We use this observation constantly in what follows, under the
name \emph{pinning}.

\begin{corollary}[Pinning]
\label{cor:p3-pinning}
Assume routing and data consistency through an index $N$ (i.e.\ every run of
Eval$(t)$ with $t \leq N$ returns $x[t]$). Let $m \in S$, $q \in Q_m$ and $\ell
\in \mathbb{N}$ be s.t.\ $m \ell + q \leq N$. Then, for every $\gamma \in
\mathcal{C}_m$,
\[
  \PEval_\gamma(q, \ell) = x[m \ell + q].
\]
\end{corollary}

\begin{proof}
Fix an entry path $\alpha$ to $m$. By \lemmaref{lem:p3-peval-lift} and data
consistency at the position $m \ell + q \leq N$,
\[
  \PEval_\gamma(q, \ell) = \Eval_{\alpha \gamma}(m \ell + q) = x[m \ell + q].
\]
\end{proof}

In particular, pinning translates the pathwise divergence values into
differences in the data.

\begin{corollary}
\label{cor:p3-div-witness}
Assume routing and data consistency through an index $N$, and suppose that
$\tau_M$ is injective on $Q_M$ for every sink $M$. Let $m \in S$ and distinct
$q_1, q_2 \in Q_m$, and put $\ell \coloneqq d_m(q_1, q_2)$. If $\ell < \infty$
and $m \ell + \max(q_1, q_2) \leq N$, then
\[
  x[m \ell + q_1] \neq x[m \ell + q_2].
\]
\end{corollary}

\begin{proof}
By \lemmaref{lem:p3-div-semantics}, some continuation $\gamma \in
\mathcal{C}_m$ satisfies
\[
  \PEval_\gamma(q_1, \ell) \neq \PEval_\gamma(q_2, \ell).
\]
For $h \in \{1, 2\}$ we have $m \ell + q_h \leq m \ell + \max(q_1, q_2) \leq
N$, so pinning (\corollaryref{cor:p3-pinning}) identifies $\PEval_\gamma(q_h,
\ell)$ with $x[m \ell + q_h]$; hence $x[m \ell + q_1] \neq x[m \ell + q_2]$.
\end{proof}

Pinning also makes divergence values with observed witnesses robust: such a
value is determined by the data alone, so it agrees between any two
configurations that both evaluate the observed prefix correctly.

\begin{corollary}[Stability]
\label{cor:p3-div-stability}
Let two configurations of the model share the scale fragment, and suppose that
each of them satisfies the hypotheses of \corollaryref{cor:p3-div-witness} for
the same index $N$: routing and data consistency through $N$, with $\tau_M$
injective on $Q_M$ at every sink $M$. Let $m \in S$, let $q_1 \neq q_2$ be
states of $Q_m$ in both configurations, and suppose that $\ell \coloneqq
d_m(q_1, q_2)$, computed in the first configuration, satisfies $m \ell +
\max(q_1, q_2) \leq N$. Then $d_m(q_1, q_2) = \ell$ in the second
configuration as well.
\end{corollary}

\begin{proof}
By \corollaryref{cor:p3-div-witness} in the first configuration, $x[m \ell +
q_1] \neq x[m \ell + q_2]$. For every $s \leq \ell$ and $h \in \{1, 2\}$, the
position $m s + q_h \leq m \ell + \max(q_1, q_2) \leq N$ is observed, so
pinning applies in both configurations: $\PEval_\gamma(q_h, s) = x[m s + q_h]$
for every $\gamma \in \mathcal{C}_m$. In the first configuration, claim 1 of
\lemmaref{lem:p3-div-semantics} turns $d_m(q_1, q_2) = \ell$ into agreement at
every row below $\ell$, and pinning turns the agreement into the data
identities $x[m s + q_1] = x[m s + q_2]$, $s < \ell$. Reading the pinned values
in the second configuration, every $\gamma \in \mathcal{C}_m$ agrees on the
pair at all rows below $\ell$ and differs at the row $\ell$ itself; hence
$D_\gamma(q_1, q_2) = \ell$ for every $\gamma$, and $d_m(q_1, q_2) = \ell$ by
\equationref{eq:p3-div-path}.
\end{proof}

In view of \corollaryref{cor:p3-div-witness}, the distinguishability part of
\lemmaref{lem:p3-inv} reduces to proving the numerical bound $m d_m(q_1, q_2) +
\max(q_1, q_2) \leq n$: the separating-symbol conclusion then follows
automatically. The pathwise agreement conclusion of
\lemmaref{lem:p3-div-semantics} is likewise exactly what justifies a rerouting
in line 10 of \algorithmref{alg:p3-update} when $d_m(\tilde{q}, q') \geq \ell$.

\subsection{The Strengthened Transient Invariant}
\label{app:p3-transient-inv}

We now prepare an induction over the transient configurations of Update.
Throughout an update round, $n$ denotes the new position and $a \coloneqq x[n]$
its observed symbol; the state before the round is therefore consistent with
$x[:n]$. The divergence tables $d_m$ used by the algorithm are the tables
computed once in line 2 of \algorithmref{alg:p3-update}. For a transient
configuration --- after some, but not necessarily all, iterations of line 3 ---
we write $\Delta_m$ for the tables that a fresh call to Diverge would compute
from that configuration. Thus $\Delta_m = d_m$ immediately after line 2, but
the two notations need not agree later in the cascade.

\textbf{Settled and saturated pairs.} The distinguishability claim of
\lemmaref{lem:p3-inv} cannot be maintained verbatim through the cascade. When
line 12 of \algorithmref{alg:p3-update} creates the state $u_m \coloneqq n
\bmod m$, line 14 copies the outgoing transitions of the old target $v_m
\coloneqq q'$, so the two states evaluate identically along every continuation,
and no divergence table computed from the transient model separates the pair
$\{u_m, v_m\}$; moreover, the divergence values between $u_m$ and the remaining
states of the scale depend on evaluations of the position $n$, which the
repairs of the same round --- reroutings and creations at larger scales ---
keep overwriting. The transient assertion therefore demands of each pair one of
two things: a witness that the data has already fixed, or a separation by the
current position itself.

For a transient configuration, a scale $m \in S$ and distinct states $q_1, q_2
\in Q_m$, we call the pair $\{q_1, q_2\}$ \emph{settled} when
\begin{equation}
  \label{eq:p3-settled}
  \begin{gathered}
    m \Delta_m(q_1, q_2) + \max(q_1, q_2) \leq n - 1 \\
    \text{and} \quad x[m \Delta_m(q_1, q_2) + q_1] \neq x[m \Delta_m(q_1, q_2)
    + q_2],
  \end{gathered}
\end{equation}
and \emph{saturated} when
\begin{equation}
  \label{eq:p3-saturated}
  \max(q_1, q_2) = n \bmod m \quad \text{and} \quad x[m \floor{n / m} + \min(q_1,
  q_2)] \neq a.
\end{equation}

Every position that \equationref{eq:p3-settled} mentions is observed, so,
wherever routing and past consistency are available, pinning reads the settled
property directly off the data, and \corollaryref{cor:p3-div-stability} makes
it immune to everything a repair can do to the tables. A saturated pair, by
contrast, mentions no divergence value at all. Its larger member is the residue
of the current position, so that $m \floor{n / m} + \max(q_1, q_2) = n$, and
\equationref{eq:p3-saturated} states that the data separates the pair at the
current row: the subsequence of the larger member carries the fresh symbol $a =
x[n]$ there, and that of the smaller member does not. This is exactly the
situation of a creation: routing gives $v_m \leq u_m$, equality being
impossible --- $v_m \in Q_m$, while line 5 requires $u_m \notin Q_m$ --- so
$\max(u_m, v_m) = u_m = n \bmod m$, and $x[m \floor{n / m} + v_m] \neq a$ is
the failed test of line 8. The verification that the repairs leave every pair
settled or saturated is carried out in \sectionref{app:p3-cascade-proof}; that
the completed cascade returns every saturated pair to the distinguishability
claim of \lemmaref{lem:p3-inv}, in \sectionref{app:p3-cascade-complete}.

\textbf{The loop invariant.} Let $\mathcal{B}$ be the set of scales whose
iterations in line 3 have been completed in the current round; because the
scales are processed in increasing order, $\mathcal{B}$ is an initial segment
of $S \setminus \{1\}$. For $k \in \mathbb{N}$, write
\[
  \mathcal{F}_k \coloneqq \{t < k \mid \forall t' < t : x[t'] \neq x[t]\}
\]
for the first occurrences of the distinct symbols in $x[:k]$. Immediately
before and after every iteration of line 3, we maintain:
\begin{enumerate}
  \item \textbf{(Past consistency.)} For every $t < n$ and $\zeta \in
    \mathcal{C}_1$, $\Eval_\zeta(t) = x[t]$.
  \item \textbf{(Routing.)} Both parts of the routing claim of
    \lemmaref{lem:p3-inv} hold for every current edge.
  \item \textbf{(Relaxed distinguishability.)} At every scale, every pair of
    distinct states is settled or saturated.
  \item \textbf{(Moving horizon.)} Every sink has scale strictly greater than
    $n$. For every $m \in S$ with $m > n$,
    \[
      Q_m = \begin{cases}
        \mathcal{F}_{n+1} & \text{if } m \in \mathcal{B}, \\
        \mathcal{F}_n & \text{if } m \notin \mathcal{B},
      \end{cases}
    \]
    and $\tau_m(q) = x[q]$ for all $q \in Q_m$.
\end{enumerate}

In particular, the output map is injective at every sink throughout the
cascade. Notice also that when the next iteration is at scale $m$, all
preceding iterations modified only state sets and transitions at scales
strictly smaller than $m$. Since the recurrence of Diverge at $m$ only inspects
$m$ and larger scales, its values on the states present at the beginning of the
round have not changed:
\begin{equation}
  \label{eq:p3-frozen-div}
  \Delta_m(q, q') = d_m(q, q') \quad \text{for all } q, q' \in Q_m.
\end{equation}
This explains why the tests in line 9 may use the frozen table of line 2.
States created during the round are, of course, not arguments of that table;
the pairs they form are classified directly in
\sectionref{app:p3-cascade-proof}.

\textbf{Initialization.} At the beginning of an epoch, Initialize$(x[0])$
produces $S = \{1\}$, $Q_1 = \{0\}$ and $\tau_1(0) = x[0]$. The unique run of
Eval$(0)$ halts at once and returns $x[0]$; routing is vacuous, there are no
distinct states, and the sole scale is a sink beyond the horizon $0$, with $Q_1
= \mathcal{F}_1$. Thus all four claims of \lemmaref{lem:p3-inv} hold after
initialization. The rest of the argument proceeds by induction over the
subsequent rounds of this epoch; the restart case is treated in
\sectionref{app:p3-pruning-proof}.

\textbf{Preprocessing.} Fix a round $n \geq 1$ and assume
\lemmaref{lem:p3-inv} after the preceding round, i.e.\ for the prefix $x[:n]$.
We prove that, after Preprocess and the call to Diverge in line 2, the loop
invariant holds with $\mathcal{B} = \emptyset$; in fact, every pair of distinct
states is settled.

If $n \notin S$, Preprocess changes nothing. The old data consistency and
routing claims persist. Because the fragment is unchanged, the freshly computed
$\Delta_m$ equals the old final divergence table, and the old
distinguishability claim --- all of whose positions are at most $n - 1$ ---
says exactly that every pair is settled. The old horizon claim says that a sink
has scale greater than $n - 1$; equality to $n$ is impossible because $n \notin
S$, so every sink is in fact greater than $n$. It also gives $Q_m =
\mathcal{F}_n$ and $\tau_m(q) = x[q]$ at every $m > n$. Hence the
moving-horizon claim holds.

Suppose now that $n \in S$. Before preprocessing, both $n$ and every $n r$, $r
\in R$, are beyond the old horizon. Thus
\begin{equation}
  \label{eq:p3-pre-horizon}
  Q_n = Q_{n r} = \mathcal{F}_n, \quad \tau_n = \tau_{n r}
\end{equation}
whenever $n r$ already belongs to $S$; if it does not, line 3 of
\algorithmref{alg:p3-preprocess} makes these identities true by definition.
Also, every target $n r$ is a sink: an edge out of a scale is created only when
the time index reaches that scale, and $n r > n$. Preprocess adds the edges
$(n, r)$ with
\begin{equation}
  \label{eq:p3-new-edge}
  \delta_n^r(q, 0) = q, \quad \delta_n^r(q, i) = 0 \quad (1 \leq i < r).
\end{equation}

We verify the four parts of the invariant.

First consider past consistency. Let $t < n$ and $\zeta \in \mathcal{C}_1$ in
the new fragment. If $\zeta$ uses no new edge, then $\Eval_\zeta(t) = x[t]$ by
old past consistency. Otherwise $\zeta$ has the form $\alpha (n, r, n r)$ for
an entry path $\alpha$ to $n$ and a new edge $(n, r)$; put
\[
  q \coloneqq \rho_\alpha(t).
\]
Since $t < n$, the digit at the new edge is $0$, so
\equationref{eq:p3-new-edge} maps $q$ to itself at the sink $n r$. Hence
\equationref{eq:p3-route-factor} and \equationref{eq:p3-pre-horizon} give
\[
  \Eval_\zeta(t) = \tau_{n r}(q) = \tau_n(q) = \Eval_\alpha(t) = x[t],
\]
where the last equality is old past consistency for the path $\alpha$, which
ended at the old sink $n$. Thus the new pathwise identity holds for every
$\zeta$.

Routing on old edges is unchanged. On a new edge, the upper bound is immediate
from \equationref{eq:p3-new-edge}. Moreover, $Q_{n r} = Q_n \subseteq [n]$:
hence a residue $q + i n$ belongs to $Q_{n r}$ only when $i = 0$, in which case
the transition equals $q$; and every $u \in Q_{n r}$ satisfies $u \bmod n = u
\in Q_n$. This proves both routing clauses.

For the moving horizon, \equationref{eq:p3-pre-horizon} shows that every scale
greater than $n$ has the state set $\mathcal{F}_n$ and the correct output map,
as required when $\mathcal{B} = \emptyset$. The only old sink that might fail
the strict inequality was $n$ itself, and it has just acquired all its outgoing
edges. All new sinks are of the form $n r > n$. Thus every sink is strictly
beyond the new position, and its output map is injective.

It remains to verify that every pair of distinct states is settled. First take
distinct states $q_1, q_2$ at a scale already present before preprocessing, and
let $R$ be their old divergence. By the induction hypothesis,
\begin{equation}
  \label{eq:p3-old-div-bound}
  m R + \max(q_1, q_2) \leq n - 1.
\end{equation}
By \lemmaref{lem:p3-div-semantics}, an old continuation separates the pair at
row $R$. If that continuation ended at a sink other than $n$, it is still a
continuation in the new fragment. If it ended at $n$, extend it by any one of
the new edges. At the two positions in \equationref{eq:p3-old-div-bound} the
remaining row upon reaching scale $n$ is $0$, and \equationref{eq:p3-new-edge}
together with \equationref{eq:p3-pre-horizon} preserves the old terminal
output. Hence the extended continuation still separates the pair at row $R$.
Applying \lemmaref{lem:p3-div-semantics} in the new fragment gives
$\Delta_m(q_1, q_2) \leq R$, so
\[
  m \Delta_m(q_1, q_2) + \max(q_1, q_2) \leq n - 1.
\]
Past consistency, routing and injectivity at the new sinks have already been
proved; \corollaryref{cor:p3-div-witness} therefore identifies the two outputs
at row $\Delta_m(q_1, q_2)$ with distinct observed symbols. This is
\equationref{eq:p3-settled}.

Finally, suppose that a scale $n r$ has just been adjoined. Its distinct states
$q_1, q_2 \in Q_{n r} = \mathcal{F}_n$ have different symbols, lie below $n$,
and, since $n r$ is a sink, satisfy $\Delta_{n r}(q_1, q_2) = 0$. Thus the
pairs of the new scale are settled as well. We have proved the loop invariant
immediately after line 2. At this moment $\Delta_m = d_m$ at every scale, which
is the initial case of \equationref{eq:p3-frozen-div}.

\subsection{Preservation by the Update Cascade}
\label{app:p3-cascade-proof}

The proof for a single iteration is longer than in the chain-shaped case of TP
because the scale being processed can have several incoming edges. We divide it
into four steps. The first isolates the genuinely new point: all active entries
into a scale carry the same past subsequence, and a state can be created only
after the cascade has activated \emph{every} incoming edge. The remaining steps
use this fact to treat rerouting, creation and the end of the loop,
respectively.

\subsubsection{Entry-Path Compatibility and the First-Edge Property}
\label{app:p3-cascade-entry}

Fix the iteration of line 3 at a scale $m$, and assume the loop invariant
before that iteration. Put
\[
  u \coloneqq n \bmod m, \quad \ell \coloneqq \floor{n / m}.
\]
The only nontrivial branch is $u \notin Q_m$. Call an incoming edge $e = (p, r)
\in \mathcal{E}$ with $p r = m$ \emph{active} in this iteration when $n \bmod p
\in Q_p$ at the beginning of the iteration at $m$; thus $\mathcal{I}$ is
exactly the set of active incoming edges. At this point the iteration at $p$ is
already complete. Hence an inactive incoming edge satisfies $n \bmod p \notin
Q_p$ even after that update; in particular, the iteration at $p$ did not create
the residue $n \bmod p$. For each $e = (p, r) \in \mathcal{I}$, define
\[
  v_e \coloneqq n \bmod p, \quad i_e \coloneqq \floor{n / p} \bmod r, \quad q_e
  \coloneqq \delta_p^r(v_e, i_e).
\]
Thus $v_e \in Q_p$, and $q_e$ is the target tested for the edge $e$ in line 8.
The next lemma records the compatibility shared by this family of targets. Its
last two claims are the points for which the multiplicity of entry paths
matters.

\begin{lemma}
\label{lem:p3-first-edge}
Suppose $u \notin Q_m$ at the beginning of the iteration at $m$.
\begin{enumerate}
  \item For every $e \in \mathcal{I}$, the corresponding target $q_e$ defined
    above satisfies $q_e < u$ and, for every continuation $\gamma \in
    \mathcal{C}_m$ and every $s < \ell$,
    \begin{equation}
      \label{eq:p3-sibling-data}
      \PEval_\gamma(q_e, s) = x[m s + u].
    \end{equation}
  \item Consequently, for all $e, f \in \mathcal{I}$,
    \begin{equation}
      \label{eq:p3-sibling-compat}
      d_m(q_e, q_f) \geq \ell.
    \end{equation}
  \item If there is an incoming edge $(p, r) \in \mathcal{E}$ with $p r = m$
    for which $n \bmod p \notin Q_p$ at the beginning of the iteration at $m$
    (equivalently, after the completed iteration at $p$), then the symbol $a$
    occurs in $x[:n]$, and line 12 is not executed for any edge of
    $\mathcal{I}$ during the iteration at scale $m$.
  \item Hence, if line 12 does fire, it does so while examining the first edge
    of $\mathcal{I}$ (in the order used by line 6), every edge into $m$ belongs
    to $\mathcal{I}$, and
    \begin{equation}
      \label{eq:p3-all-refuted}
      \Eval(m \ell + q_e) \neq a \quad \text{for every } e \in \mathcal{I}.
    \end{equation}
\end{enumerate}
\end{lemma}

\begin{proof}
Fix $e = (p, r) \in \mathcal{I}$ and an entry path $\alpha$ to $p$. Since $v_e
\in Q_p$, claim 2 of \lemmaref{lem:p3-route-map} gives
\[
  \rho_\alpha(v_e) = v_e.
\]
For $s < \ell$, let $t_s \coloneqq m s + u$. Because $p$ divides $m$,
\[
  t_s \bmod p = u \bmod p = v_e, \quad \floor{t_s / p} \bmod r = \floor{u / p}
  = i_e.
\]
Thus, for the entry path $\alpha e$ obtained by appending $e$ to $\alpha$,
\equationref{eq:p3-route} gives
\[
  \rho_{\alpha e}(u) = \delta_p^r(\rho_\alpha(v_e), i_e) = q_e.
\]
By \equationref{eq:p3-route-factor} and past consistency, for every $\gamma \in
\mathcal{C}_m$,
\[
  \PEval_\gamma(q_e, s) = \PEval_\gamma(\rho_{\alpha e}(u), s) = \Eval_{\alpha
  e \gamma}(t_s) = x[t_s],
\]
because $t_s < m \ell + u = n$. This proves
\equationref{eq:p3-sibling-data}. Routing also gives
\[
  q_e \leq v_e + i_e p = u.
\]
The inequality is strict because $q_e \in Q_m$ whereas $u \notin Q_m$.

\equationref{eq:p3-sibling-data} says that the partial evaluations of $q_e$ and
$q_f$ agree along every continuation at every row below $\ell$. By
\lemmaref{lem:p3-div-semantics}, their current divergence is therefore at least
$\ell$. The iteration has not changed $Q_m$ or any transition out of $m$, and
earlier iterations changed only smaller scales; hence
\equationref{eq:p3-frozen-div} identifies this current divergence with the
frozen value $d_m$ used by line 9. This proves
\equationref{eq:p3-sibling-compat}.

We next prove claim 3. Let $(p, r)$ be an incoming edge furnished by its
premise. Thus $n \bmod p \notin Q_p$ after the completed iteration at $p$: that
update did not create the residue $n \bmod p$. Choose an entry path from the
root to $p$, and let $k$ be its first scale for which, after the iteration at
$k$,
\[
  u_k \coloneqq n \bmod k \notin Q_k.
\]
Such a $k$ exists. If $h$ is its predecessor on the path, then $n \bmod h \in
Q_h$ after the iteration at $h$ (and the root residue $0$ is always present).
Hence the edge $(h, k / h)$ entering $k$ was active in the iteration at $k$.
Put $L \coloneqq \floor{n / k}$, and let $z_k$ be the final target of the
corresponding transition after that iteration. Since $u_k \notin Q_k$ after the
iteration, that iteration did not create it. Thus either the target already
passed line 8, or line 10 rerouted it to a state passing line 9. In both cases
$z_k < u_k$, and the compatibility test, together with the data consistency
then in force, gives the permanent data identities
\begin{equation}
  \label{eq:p3-frontier-data}
  x[k j + z_k] = x[k j + u_k] \quad (j < L), \quad x[k L + z_k] = a.
\end{equation}
Indeed, in the no-rerouting case the active entry itself routes the residue
$u_k$; in the rerouting case \lemmaref{lem:p3-div-semantics} supplies the
equality at every row below $L$. All indices in the first identity and the
index $k L + z_k$ are smaller than $n$; in particular, the second identity of
\equationref{eq:p3-frontier-data} places an occurrence of $a$ at a position
smaller than $n$, which proves the first assertion of claim 3.

Let $\hat{\alpha}$ be the entry path to $m$ obtained by appending the rest of
the chosen path and the inactive edge $(p, r)$ to its prefix $\alpha_k$ ending
at $k$. Set
\[
  w \coloneqq u - u_k + z_k.
\]
Since $k$ divides $m$, we have $u_k = n \bmod k = u \bmod k$. Write $b
\coloneqq \floor{u / k}$. Then
\[
  u = k b + u_k, \quad w = k b + z_k.
\]
Since $z_k < u_k$, we have $0 \leq w < u < m$, so $z \coloneqq
\rho_{\hat{\alpha}}(w)$ is defined. Claim 1 of \lemmaref{lem:p3-route-map}
gives
\[
  z \leq w < u.
\]

With $b$ as above, $L = (m / k) \ell + b$. Also,
\[
  m \ell + w = n - u_k + z_k = k L + z_k < n.
\]
For any $\gamma \in \mathcal{C}_m$, \equationref{eq:p3-route-factor}, current
past consistency and the second identity of
\equationref{eq:p3-frontier-data} therefore give
\[
  \PEval_\gamma(z, \ell) = \Eval_{\hat{\alpha} \gamma}(m \ell + w) = x[k L +
  z_k] = a.
\]
In particular, $\Eval(m \ell + z) = a$; this evaluation is of an old position
because $z < u$. For $s < \ell$, put $J_s \coloneqq (m / k) s + b < L$. Then $m
s + w = k J_s + z_k$, and \equationref{eq:p3-route-factor} together with the
first identity of \equationref{eq:p3-frontier-data} gives
\[
  \PEval_\gamma(z, s) = x[k J_s + z_k] = x[k J_s + u_k] = x[m s + u].
\]
Comparing with \equationref{eq:p3-sibling-data} and using
\lemmaref{lem:p3-div-semantics} once more, we obtain
\[
  d_m(z, q_e) \geq \ell
\]
for every active target $q_e$. Thus, for every edge $e \in \mathcal{I}$, either
line 8 passes or line 9 supplies $z$ as a candidate. Consequently, line 12 is
not executed for any edge of $\mathcal{I}$ during the iteration at scale $m$.
This proves claim 3.

It remains to establish the first-edge assertion. Before any edge of
$\mathcal{I}$ is processed, \equationref{eq:p3-sibling-compat} holds among all
of their targets. If an edge passes line 8, its target has row-$\ell$ value $a$
and is compatible through the preceding rows with every still-unprocessed
target. If instead line 10 reroutes it to some $z$, then $d_m(z, q_e) \geq
\ell$ and $\Eval(m \ell + z) = a$; transitivity of the pathwise equalities
below row $\ell$, together with \equationref{eq:p3-sibling-compat}, makes $z$
compatible with every still-unprocessed target as well. In either case, every
later refuted target has an admissible candidate in line 9 and cannot cause a
creation. A creation must therefore occur, if at all, on the first edge
examined.

By claim 3, such a creation implies that no incoming edge is inactive, so
$\mathcal{I}$ contains every edge into $m$. Finally, if some sibling target
$q_e$ had $\Eval(m \ell + q_e) = a$, then \equationref{eq:p3-sibling-compat}
would make $q_e$ a candidate for the first edge, contradicting the failure of
line 9 that triggers line 12. This proves \equationref{eq:p3-all-refuted}.
\end{proof}

\subsubsection{Preservation under Rerouting}
\label{app:p3-cascade-reroute}

We now verify that the loop invariant survives the iterations of line 3 whose
repairs are all reroutings. The iteration under consideration is at the scale
$m$, with $u = n \bmod m$ and $\ell = \floor{n / m}$ as in
\sectionref{app:p3-cascade-entry}; we extend the edge notation of that
subsubsection from $\mathcal{I}$ to all of $\mathcal{E}$: an arbitrary edge $e
= (p, r) \in \mathcal{E}$ has the target scale $m_e \coloneqq p r$, the source
residue $v_e \coloneqq n \bmod p$ and the digit $i_e \coloneqq \floor{n / p}
\bmod r$, so that $m_e = m$ for the edges of $\mathcal{I}$. Recall from claim 4
of \lemmaref{lem:p3-first-edge} that a creation, if one occurs, occurs at the
first edge of $\mathcal{I}$ examined, before any assignment of line 10, and
line 15 then abandons the loop; an iteration that executes line 10 therefore
never executes line 12, and its entire effect on the state is a sequence of
assignments of line 10, one at each refuted edge of $\mathcal{I}$, in the order
of examination. Note also that the frozen-divergence identity
\equationref{eq:p3-frozen-div} persists through these assignments: each of them
modifies an entry at a source scale $p < m$, which the recurrence of Diverge at
the scales at least $m$ never inspects. We prove that a single assignment of
line 10 preserves the loop invariant (\lemmaref{lem:p3-reroute-step}), then
record the effect of a complete creation-free iteration
(\lemmaref{lem:p3-reroute-iter}); state creation is taken up in
\sectionref{app:p3-cascade-create}.

The \emph{current entry} of an edge $e = (p, r) \in \mathcal{E}$ is the table
entry $\delta_p^r(v_e, i_e)$: the entry through which the transition recursion
of the position $n$ refines the residue of $n$ from scale $p$ to scale $m_e$
(the entry depends on the current position $n$, whence the name). The tests and
repairs of the iteration all happen at current entries: line 7 reads the
current entry of the edge under examination for the test of line 8 --- its
value is still the target $q_e$ of \sectionref{app:p3-cascade-entry}, since an
entry is modified only when its own edge is processed --- and the assignments
of lines 10 and 17 write current entries and nothing else. Accordingly, for a
path $\zeta \in \mathcal{C}_1$ containing $e$ --- so that $\zeta = \alpha e
\gamma$, where $\alpha$ is the prefix of $\zeta$ ending at $p$ and $\gamma \in
\mathcal{C}_{m_e}$ is the continuation of $\zeta$ from $m_e$ --- say that the
run of $\Eval_\zeta(t)$ \emph{crosses} the current entry of $e$ when
\[
  \rho_\alpha(t \bmod p) = v_e \quad \text{and} \quad \floor{t / p} \bmod r =
  i_e,
\]
i.e.\ when the transition recursion of $\Eval_\zeta(t)$ arrives at the scale
$p$ in the state $v_e$ and consumes the digit $i_e$ at the edge $e$. The
following lemma identifies the residue that a current entry refines, and bounds
the traffic through the entry: only positions whose residue modulo $m_e$ is at
least that of $n$ can cross it, observed positions cross it only at rows below
that of $n$, and the run of the position $n$ itself, when the edge's source
pair is active, cannot avoid it. The lemma is stated for an arbitrary
configuration arising during the round, and will serve again in the creation
analysis of \sectionref{app:p3-cascade-create}.

\begin{lemma}
\label{lem:p3-current-entry}
Let $e = (p, r) \in \mathcal{E}$ with target scale $m_e = p r$. Then
\begin{equation}
  \label{eq:p3-entry-residue}
  n \bmod m_e = v_e + p i_e.
\end{equation}
Moreover, let $\zeta = \alpha e \gamma \in \mathcal{C}_1$ contain $e$, where
$\alpha$ is the prefix of $\zeta$ ending at $p$ and $\gamma \in
\mathcal{C}_{m_e}$.
\begin{enumerate}
  \item Suppose the upper-bound part of routing holds along $\alpha$. If the
    run of $\Eval_\zeta(t)$, $t \in \mathbb{N}$, crosses the current entry of
    $e$, then $t \bmod m_e \geq n \bmod m_e$ and
    \begin{equation}
      \label{eq:p3-entry-factor}
      \Eval_\zeta(t) = \PEval_\gamma(\delta_p^r(v_e, i_e), \floor{t / m_e}).
    \end{equation}
    In particular, if $t < n$, then $\floor{t / m_e} < \floor{n / m_e}$.
  \item Suppose the full routing claim of \lemmaref{lem:p3-inv} holds along
    $\alpha$, and $v_e \in Q_p$. Then the run of $\Eval_\zeta(n)$ crosses the
    current entry of $e$; in particular, \equationref{eq:p3-entry-factor} holds
    at $t = n$, with the row $\floor{n / m_e}$.
\end{enumerate}
\end{lemma}

\begin{proof}
We first prove \equationref{eq:p3-entry-residue}. Denote $u' \coloneqq n \bmod
m_e$. Since $p$ divides $m_e$, we have $u' \bmod p = n \bmod p = v_e$; and
dividing $n = m_e \floor{n / m_e} + u'$ by $p$ gives $\floor{n / p} = r
\floor{n / m_e} + \floor{u' / p}$ with $\floor{u' / p} < r$, whence $i_e =
\floor{n / p} \bmod r = \floor{u' / p}$. Combining the two identities, $u' =
v_e + p i_e$, which is \equationref{eq:p3-entry-residue}.

Turning to the claims, write $t' \coloneqq t \bmod p$ and $\ell' \coloneqq
\floor{t / p}$, so that $t = p \ell' + t'$ and, since $m_e = p r$,
\[
  t \bmod m_e = t' + p \cdot (\ell' \bmod r).
\]
The continuation of $\zeta$ from $p$ is the edge $e$ followed by $\gamma$;
denote it $e \gamma \in \mathcal{C}_p$. Claim 3 of
\lemmaref{lem:p3-route-map}, applied to the entry path $\alpha$ and the
continuation $e \gamma$, gives $\Eval_\zeta(t) = \PEval_{e
\gamma}(\rho_\alpha(t'), \ell')$; and unfolding the first step of the recursion
\equationref{eq:p3-peval}, which consumes the digit $\ell' \bmod r$ at the edge
$e$ and leaves the row $\floor{\ell' / r} = \floor{t / m_e}$, we obtain
\begin{equation}
  \label{eq:p3-entry-unfold}
  \Eval_\zeta(t) = \PEval_\gamma(\delta_p^r(\rho_\alpha(t'), \ell' \bmod r),
  \floor{t / m_e}).
\end{equation}

For claim 1, suppose the run crosses the current entry: $\rho_\alpha(t') = v_e$
and $\ell' \bmod r = i_e$. Substituting the two crossing conditions into
\equationref{eq:p3-entry-unfold} yields
\equationref{eq:p3-entry-factor}. Claim 1 of \lemmaref{lem:p3-route-map} gives
$v_e = \rho_\alpha(t') \leq t'$, so
\[
  t \bmod m_e = t' + p \cdot (\ell' \bmod r) \geq v_e + p i_e = u'.
\]
Finally, if $t < n$, then
\[
  m_e \floor{t / m_e} = t - (t \bmod m_e) \leq t - u' < n - u' = m_e \floor{n /
  m_e},
\]
so $\floor{t / m_e} < \floor{n / m_e}$.

For claim 2, take $t = n$, so that $t' = v_e$ and $\ell' = \floor{n / p}$,
whence $\ell' \bmod r = i_e$. Since $v_e \in Q_p$ and the full routing claim
holds along $\alpha$, claim 2 of \lemmaref{lem:p3-route-map} gives
$\rho_\alpha(t') = \rho_\alpha(v_e) = v_e$. Thus both crossing conditions hold,
and \equationref{eq:p3-entry-factor} at $t = n$ follows by claim 1, whose
routing hypothesis is implied by the full claim.
\end{proof}

The workhorse of this subsubsection is the following preservation lemma for a
single assignment of line 10. Its third claim records what the repair achieves
at its own edge: every evaluation of the position $n$ routed through the
repaired edge now returns the observed symbol. (In the chain-shaped state of
TP, this already restores $\Eval(n) = a$ outright; in the scale fragment, the
runs of $n$ through the other branches remain unrepaired until their own scales
are reached, which is why the claim is confined to the paths through $e$.)

\begin{lemma}
\label{lem:p3-reroute-step}
Suppose that $u \notin Q_m$ and that, at some point of the iteration at $m$,
line 10 commits the state $\tilde{q} =: z$ at the edge $e = (p, r) \in
\mathcal{I}$. Assume that at this moment the four parts of the loop invariant
hold and $\Delta_m = d_m$ on $Q_m$. Then:
\begin{enumerate}
  \item $q_e < u$ and $z < u$; moreover $x[m \ell + q_e] \neq a$ and $x[m \ell
    + z] = a$ --- in particular, the symbol $a$ occurs in $x[:n]$;
  \item after the assignment $\delta_p^r(v_e, i_e) \gets z$, the four parts of
    the loop invariant hold and $\Delta_m = d_m$ on $Q_m$, with $\mathcal{B}$
    unchanged;
  \item in the configuration after the assignment, $\Eval_\zeta(n) = a$ for
    every $\zeta \in \mathcal{C}_1$ containing $e$.
\end{enumerate}
\end{lemma}

\begin{proof}
We refer to the configurations immediately before and immediately after the
assignment as the \emph{old} and the \emph{new} configuration, and decorate
configuration-dependent quantities accordingly when the distinction matters;
the hypotheses of the lemma concern the old one. As noted above, the current
entry of $e$ still holds the target $q_e = \delta_p^r(v_e, i_e)$ of
\sectionref{app:p3-cascade-entry} when line 7 reads it, and no entry is written
between that read and the assignment; the two configurations therefore differ
in exactly one table entry --- the current entry of $e$, overwritten from $q_e$
to $z$.

\textbf{Claim 1.} Since the target scale of $e$ is $m$,
\equationref{eq:p3-entry-residue} reads $u = v_e + p i_e$. The upper-bound part
of routing at the current entry therefore gives $q_e \leq v_e + i_e p = u$, and
the inequality is strict because $q_e \in Q_m$ while $u \notin Q_m$. Hence $m
\ell + q_e < m \ell + u = n$: the position evaluated by line 8 is observed,
past consistency makes every run of Eval$(m \ell + q_e)$ return $x[m \ell +
q_e]$, and the failure of the test of line 8 means exactly $x[m \ell + q_e]
\neq a$. In particular $z \neq q_e$: were the two equal, the run of Eval$(m
\ell + z)$ performed by line 9 would have returned $x[m \ell + q_e] \neq a$,
contrary to its success. The pair $\{z, q_e\}$ thus consists of two distinct
states of $Q_m$; it is not saturated --- the larger of two states cannot be the
residue $n \bmod m = u \notin Q_m$ --- so relaxed distinguishability makes it
settled. By the hypothesis $\Delta_m = d_m$ and the divergence condition
committed by line 9, $\Delta_m(z, q_e) = d_m(z, q_e) \geq \ell$, so
\equationref{eq:p3-settled} gives
\[
  m \ell + \max(z, q_e) \leq m \Delta_m(z, q_e) + \max(z, q_e) \leq n - 1 < n =
  m \ell + u,
\]
whence $z < u$. Hence $m \ell + z < n$ is an observed position as well: the run
of line 9 returned $x[m \ell + z]$, its success means $x[m \ell + z] = a$, and
the symbol $a$ occurs in $x[:n]$, at the position $m \ell + z$. This proves
claim 1.

\textbf{Claim 2.} We first record what the assignment does not change. It
creates and deletes no scale, no edge and no state: the fragment, the
continuation sets $\mathcal{C}_k$, the state sets, the output functions, the
sinks and $\mathcal{B}$ are the same in both configurations. The
moving-horizon part of the loop invariant therefore persists verbatim, and the
output maps remain injective at the sinks, so \lemmaref{lem:p3-div-semantics}
applies in both configurations. Moreover, for $\gamma \in \mathcal{C}_m$, the
recursion of $\PEval_\gamma$ reads only tables at scales visited by $\gamma$
--- all at least $m > p$ --- so $\PEval_\gamma$ is the same function in both
configurations. Likewise, for any scale $k > p$, a fresh run of Diverge
computes $\Delta_k$ from the transition tables out of the scales reachable from
$k$ --- all at least $k$ --- and from the state sets at the sinks, none of
which was modified; hence $\Delta_k^{\mathrm{new}} = \Delta_k^{\mathrm{old}}$
for every $k > p$, and in particular the identity $\Delta_m = d_m$ on $Q_m$
persists. It remains to verify the other three parts of the loop invariant.

\emph{Routing.} Only the current entry of $e$ changed. Its new value satisfies
$z < u = v_e + i_e p$ (claim 1 and \equationref{eq:p3-entry-residue}), which is
the upper-bound clause; the equality clause at this entry is vacuous, because
the refined residue $v_e + i_e p = u$ is not a state of $Q_m$; and heredity
constrains the state sets alone, which did not change. The routing clauses at
every other entry are untouched. Hence routing holds in the new configuration.

\emph{Past consistency.} Let $t < n$ and $\zeta \in \mathcal{C}_1$. The run of
$\Eval_\zeta(t)$ reads the modified entry only if $\zeta$ traverses $e$ --- say
$\zeta = \alpha e \gamma$, with $\alpha$ an entry path to $p$ and $\gamma \in
\mathcal{C}_m$ --- and the run crosses the current entry of $e$. The crossing
conditions depend only on the digits of $t$ and on the route $\rho_\alpha$,
which is computed from tables at scales smaller than $p$; hence they hold in
the new configuration iff they hold in the old one. If $\zeta$ does not
traverse $e$, or if the run does not cross the current entry, then every entry
the run reads is unmodified, the run is identical in the two configurations,
and $\Eval_\zeta(t) = x[t]$ by old past consistency. If the run crosses, claim
1 of \lemmaref{lem:p3-current-entry} --- whose hypothesis, the upper-bound part
of routing along $\alpha$, holds in both configurations --- evaluates the two
runs through \equationref{eq:p3-entry-factor} as
\[
  \Eval_\zeta^{\mathrm{old}}(t) = \PEval_\gamma(q_e, \floor{t / m}), \quad
  \Eval_\zeta^{\mathrm{new}}(t) = \PEval_\gamma(z, \floor{t / m}),
\]
at the common row $\floor{t / m} < \floor{n / m} = \ell$ (using $t < n$). Since
$\Delta_m(z, q_e) = d_m(z, q_e) \geq \ell$, claim 1 of
\lemmaref{lem:p3-div-semantics}, applied in the new configuration, makes
$\PEval_\gamma(z, s) = \PEval_\gamma(q_e, s)$ for every $s < \ell$; hence
$\Eval_\zeta^{\mathrm{new}}(t) = \Eval_\zeta^{\mathrm{old}}(t) = x[t]$ in this
case as well.

\emph{Relaxed distinguishability.} Routing and past consistency hold in both
configurations, the fragment is shared, and the output maps are injective at
the sinks; pinning (\corollaryref{cor:p3-pinning}) and stability
(\corollaryref{cor:p3-div-stability}) are therefore available, with $N = n -
1$. Fix a scale $k \in S$ and distinct states $w_1, w_2 \in Q_k$. If the pair
is settled in the old configuration, \corollaryref{cor:p3-div-stability}
preserves its divergence value, and the two clauses of
\equationref{eq:p3-settled} --- statements about the data at observed positions
--- persist with it: the pair is settled in the new configuration. If the pair
is saturated, \equationref{eq:p3-saturated} mentions only the data, the residue
$n \bmod k$ and the symbol $a$, none of which the assignment touches, and it
persists verbatim. Relaxed distinguishability thus holds at every scale, and
claim 2 is proved.

\textbf{Claim 3.} Let $\zeta = \alpha e \gamma \in \mathcal{C}_1$ contain $e$,
with $\alpha$ an entry path to $p$ and $\gamma \in \mathcal{C}_m$. By claim 2,
the full routing claim holds in the new configuration, and $v_e \in Q_p$
because $e \in \mathcal{I}$; claim 2 of \lemmaref{lem:p3-current-entry}
therefore shows that the run of $\Eval_\zeta(n)$ crosses the current entry of
$e$, and \equationref{eq:p3-entry-factor} gives
\[
  \Eval_\zeta(n) = \PEval_\gamma(\delta_p^r(v_e, i_e), \ell) =
  \PEval_\gamma(z, \ell) = x[m \ell + z] = a,
\]
the last two equalities by pinning at the observed position $m \ell + z < n$
and by claim 1.
\end{proof}

We can now dispose of every iteration of the main loop that does not create a
state.

\begin{lemma}
\label{lem:p3-reroute-iter}
Suppose the loop invariant holds immediately before the iteration of line 3 at
the scale $m$, and that this iteration does not execute line 12. Then the loop
invariant holds immediately after the iteration, with $\mathcal{B}$ replaced by
$\mathcal{B} \cup \{m\}$. Moreover:
\begin{enumerate}
  \item if $u \in Q_m$ when line 5 is evaluated, the iteration modifies
    nothing: the assignments of line 17 rewrite the values already stored;
  \item otherwise, $u \notin Q_m$ holds at the end of the iteration as well; in
    the configuration at the end of the iteration, the current entry of every
    edge $e \in \mathcal{I}$ holds a value $w_e$ with $w_e < u$ and $x[m \ell +
    w_e] = a$, and $\Eval_\zeta(n) = a$ for every $\zeta \in \mathcal{C}_1$
    containing an edge of $\mathcal{I}$.
\end{enumerate}
\end{lemma}

\begin{proof}
The set $\mathcal{I}$ and the row $\ell$ are computed by line 4 once, at the
beginning of the iteration; for the edges of $\mathcal{I}$ we use the notation
$v_e$, $i_e$, $q_e$ of \sectionref{app:p3-cascade-entry}. The iteration deletes
nothing and, line 12 not being executed, creates nothing either: the fragment,
the continuation sets, the state sets, the output functions and the sinks are
the same in every configuration of the iteration. We treat the two branches of
line 5 in turn.

\emph{The branch $u \in Q_m$.} Lines 6--15 are skipped, and line 16 finds $n
\bmod m = u \in Q_m$, so the iteration consists of the assignments of line 17,
one at the current entry of each edge of $\mathcal{I}$. Let $e = (m', r) \in
\mathcal{I}$, so that $v_e \in Q_{m'}$. By
\equationref{eq:p3-entry-residue}, the residue that the current entry of $e$
refines is $v_e + m' i_e = n \bmod m = u$, a state of $Q_m$; the equality
clause of routing therefore places the entry at $\delta_{m'}^r(v_e, i_e) = u$
already, and the assignment of line 17, which writes $n \bmod m = u$ into it,
rewrites the stored value. The iteration thus modifies nothing, which proves
claim 1, and the four parts of the loop invariant persist verbatim. Enlarging
$\mathcal{B}$ costs nothing either: the only part that consults $\mathcal{B}$
is the moving-horizon claim, which does so at the scales greater than $n$
alone, and here $m \leq n$ --- indeed, $m > n$ would give $u = n \bmod m = n$,
whereas the moving-horizon claim before the iteration ($m \notin \mathcal{B}$)
puts $Q_m = \mathcal{F}_n$, whose elements are smaller than $n$, contradicting
$u \in Q_m$. Hence the loop invariant holds after the iteration with
$\mathcal{B} \cup \{m\}$.

\emph{The branch $u \notin Q_m$: the effect of the iteration.} Suppose now that
line 5 finds $u \notin Q_m$. Since line 12 is not executed, the search of line
9 succeeds at every edge of $\mathcal{I}$ that the test of line 8 refutes, and
the entire effect of the iteration on the state is the resulting sequence of
assignments of line 10 --- one at the current entry of each refuted edge, in
the order used by line 6, and none at the edges that pass the test. These
writes do not collide: an edge into $m$ is determined by its source scale, so
the current entries of distinct edges of $\mathcal{I}$ lie in distinct tables,
and the current entry of an edge of $\mathcal{I}$ is thus written --- at most
once --- during the examination of that edge itself. Finally, $Q_m$ does not
change, so $u \notin Q_m$ holds at the end of the iteration as well, and line
16 skips the assignments of line 17; this is the first assertion of claim 2.

\emph{Chaining the reroute steps.} We claim that at every configuration of the
iteration the following hold: the four parts of the loop invariant, with the
$\mathcal{B}$ of the hypothesis, and the identity $\Delta_m = d_m$ on $Q_m$. At
the beginning of the iteration, the first is the hypothesis of the lemma and
the second is the frozen-divergence identity
\equationref{eq:p3-frozen-div}. The configuration changes only at the
assignments of line 10 and, together with $u \notin Q_m$, the two facts are
exactly the hypotheses of \lemmaref{lem:p3-reroute-step} at the first
assignment; claim 2 of that lemma re-establishes them immediately after the
assignment, hence they are met at the second assignment as well and, continuing
along the assignments in their order of execution, at every configuration of
the iteration --- in particular at its end, where the four parts of the loop
invariant hold with $\mathcal{B}$ unchanged. (If no edge of $\mathcal{I}$ is
refuted, the iteration does not modify the state, and there is nothing to
prove.)

\emph{The final entries.} For $e \in \mathcal{I}$, let $w_e$ denote the value
of the current entry of $e$ at the end of the iteration. We claim that
\[
  w_e < u \quad \text{and} \quad x[m \ell + w_e] = a.
\]
If the test of line 8 refuted $e$, then line 10 committed some state $z$ at $e$
and the assignment made $w_e = z$; claim 1 of \lemmaref{lem:p3-reroute-step}
gives $z < u$ and $x[m \ell + z] = a$. If instead $e$ passed the test, its
current entry was never written, so $w_e = q_e$, and $q_e < u$ by claim 1 of
\lemmaref{lem:p3-first-edge}; the position $m \ell + q_e < m \ell + u = n$ is
therefore observed and, by the past consistency in force when line 8 ran, the
run of Eval$(m \ell + q_e)$ that the test performed returned $x[m \ell + q_e]$,
so its passing means exactly $x[m \ell + q_e] = a$.

\emph{Enlarging $\mathcal{B}$.} As in the first branch, replacing $\mathcal{B}$
by $\mathcal{B} \cup \{m\}$ can matter only for the moving-horizon claim, and
only at the scale $m$ itself, whose prescribed state set changes from
$\mathcal{F}_n$ to $\mathcal{F}_{n+1}$ when $m > n$; at every other scale
beyond the horizon, neither the state set nor the membership in $\mathcal{B}$
has changed. Suppose then that $m > n$, so that $u = n$ and $\ell = 0$. The
state set $Q_m = \mathcal{F}_n$ and the outputs $\tau_m(q) = x[q]$ are
unchanged, and the sets $\mathcal{F}_n \subseteq \mathcal{F}_{n+1}$ differ at
most in a first occurrence at the position $n$; it therefore suffices to
exhibit an occurrence of the symbol $a$ in $x[:n]$. The scale $m$ is the target
of at least one edge, as \sectionref{sec:p3-model} requires: the requirement is
restored by post-processing at every round's end
(\sectionref{app:p3-post}), established by preprocessing for the scales it
adjoins, and undisturbed by the cascade, which deletes nothing. If some edge
into $m$ is inactive, claim 3 of \lemmaref{lem:p3-first-edge} furnishes the
occurrence. Otherwise every edge into $m$ is active, so $\mathcal{I} \neq
\emptyset$, and any $e \in \mathcal{I}$ places the occurrence at the observed
position $m \ell + w_e < n$, where $x[m \ell + w_e] = a$. In either case
$\mathcal{F}_{n+1} = \mathcal{F}_n$, the moving-horizon claim holds with $m$
adjoined to $\mathcal{B}$, and the loop invariant holds after the iteration
with $\mathcal{B} \cup \{m\}$.

\emph{Runs of the position $n$ through $\mathcal{I}$.} Finally, let $\zeta \in
\mathcal{C}_1$ contain an edge of $\mathcal{I}$ --- necessarily a unique one:
the edges of $\mathcal{I}$ all target the scale $m$, and scales increase
strictly along a path --- say $\zeta = \alpha e \gamma$, where $e = (p, r) \in
\mathcal{I}$, $\alpha$ is an entry path to $p$ and $\gamma \in \mathcal{C}_m$.
Consider the configuration at the end of the iteration: the routing part of the
loop invariant holds in it, and $v_e \in Q_p$ because the state sets did not
change, so claim 2 of \lemmaref{lem:p3-current-entry} shows that the run of
$\Eval_\zeta(n)$ crosses the current entry of $e$, whose value is now $w_e$.
Hence, by \equationref{eq:p3-entry-factor} and pinning
(\corollaryref{cor:p3-pinning}, with $N = n - 1$) at the observed position $m
\ell + w_e < n$,
\[
  \Eval_\zeta(n) = \PEval_\gamma(w_e, \ell) = x[m \ell + w_e] = a.
\]
\end{proof}

\subsubsection{Preservation under State Creation}
\label{app:p3-cascade-create}

It remains to treat the iterations of line 3 that execute line 12. Fix such an
iteration, at the scale $m$, with $u = n \bmod m$ and $\ell = \floor{n / m}$,
and recall the notation $v_e$, $i_e$, $q_e$ of
\sectionref{app:p3-cascade-entry} for the edges of $\mathcal{I}$; line 12
presupposes the branch $u \notin Q_m$ of line 5. By claim 4 of
\lemmaref{lem:p3-first-edge}, the creation fires while the first edge of
$\mathcal{I}$ is examined, before any assignment of line 10: at that moment the
configuration is still the one from the beginning of the iteration, every edge
into $m$ belongs to $\mathcal{I}$, and every target is refuted
(\equationref{eq:p3-all-refuted}). Line 15 then abandons the loop of line 6,
line 16 finds $n \bmod m \in Q_m$, and line 17 commits the current entry of
every edge of $\mathcal{I}$, none of which has been modified in this iteration.
In the interval between line 12 and these commits the invariant is genuinely
false --- the residue $u$ is a state, while the current entries still route it
to smaller targets, in violation of the equality clause of routing --- so,
unlike in \sectionref{app:p3-cascade-reroute}, we do not verify the invariant
assignment by assignment: the unit of verification is the iteration as a whole.

\begin{lemma}
\label{lem:p3-create-iter}
Suppose that the loop invariant holds immediately before the iteration of line
3 at the scale $m$, and that this iteration executes line 12. Then the loop
invariant holds immediately after the iteration, with $\mathcal{B}$ replaced by
$\mathcal{B} \cup \{m\}$.
\end{lemma}

\begin{proof}
We compare the \emph{old} configuration, from the beginning of the iteration,
with the \emph{new} one, from its end; as recorded above, lines 12--14 and the
assignments of line 17 are the iteration's entire effect on the state.
Throughout the proof, $q' \coloneqq q_{e^*}$ denotes the target of the first
edge $e^* \in \mathcal{I}$ examined by line 6 --- the state whose refuted test
at line 8 and failed search at line 9 fired line 12, and whose outgoing
transitions line 14 copies --- and $q' < u$ by claim 1 of
\lemmaref{lem:p3-first-edge}.

First, $m \neq n$: otherwise $u = n \bmod n = 0$ and, for any edge $e = (p, r)
\in \mathcal{I}$ --- the set is nonempty, line 12 having fired within the loop
over it --- the upper bound of routing and \equationref{eq:p3-entry-residue}
give $q_e \leq v_e + p i_e = u = 0$, so that $q_e = u$; but $q_e$, a value of a
transition table, is a state of $Q_m$ (\sectionref{sec:p3-model}), while $u
\notin Q_m$. Second, a scale beyond the horizon has no outgoing edges --- an
edge out of a scale is created only in the round in which the time index
reaches it (\sectionref{app:p3-preprocessing}) --- while, by the moving
horizon, every sink lies beyond it: hence $m > n$ if and only if $m$ is a sink.
The creation is thus either \emph{interior} ($m < n$, at a non-sink) or
\emph{at a sink} ($m > n$, with $u = n$ and $\ell = 0$), and the new
configuration differs from the old one exactly as follows:
\begin{enumerate}
  \item $Q_m$ acquires the state $u$ (line 12);
  \item at a sink, the output function is extended by $\tau_m(u) \coloneqq a$
    (line 13), and line 14 is vacuous;
  \item at an interior creation, line 13 is skipped, and the tables acquire the
    rows of the new state: $\delta_m^{r'}(u, j) \coloneqq \delta_m^{r'}(q', j)$
    for every $(m, r') \in \mathcal{E}$ and $j \in [r']$ (line 14);
  \item for every edge $e = (p, r)$ into $m$, the current entry is overwritten
    with the new state: $\delta_p^r(v_e, i_e) \gets u$ (line 17).
\end{enumerate}
The fragment, and hence the continuation sets and the sinks, are unchanged.

\emph{Locality.} At an interior creation, the recursion of $\PEval_\gamma(w,
s)$, for an old state $w \in Q_m^{\mathrm{old}}$ and $\gamma \in
\mathcal{C}_m$, starts at a row of an old state and moves to scales above $m$,
so it reads no modified entry: the values of unmodified entries are unchanged,
and the modified entries --- the rows of $u$, at the scale $m$ itself, and the
current entries, in the tables of the edges into $m$, which only a passage from
a source scale $p < m$ to $m$ reads --- are never encountered. Hence
$\PEval_\gamma(w, \cdot)$ is the same function in both configurations, for
every $w \in Q_m^{\mathrm{old}}$ and $\gamma \in \mathcal{C}_m$. Moreover, for
every $\gamma \in \mathcal{C}_m$ and $s \in \mathbb{N}$, the recursion of
$\PEval_\gamma^{\mathrm{new}}(u, s)$ reads, at its first step, the copied entry
$\delta_m^{r'}(u, s \bmod r') = \delta_m^{r'}(q', s \bmod r')$ and continues
from its value exactly as the recursion of $\PEval_\gamma(q', s)$ does after
its own first step, which reads the same entry; hence
\[
  \PEval_\gamma^{\mathrm{new}}(u, s) = \PEval_\gamma(q', s) \quad (\gamma \in
  \mathcal{C}_m,\ s \in \mathbb{N}),
\]
the right-hand side being configuration-independent. At a sink, instead,
$\mathcal{C}_m$ consists of the empty continuation, $\PEval_\gamma(w, s) =
\tau_m(w)$ is unchanged for old states, and $\PEval_\gamma^{\mathrm{new}}(u, s)
= \tau_m(u) = a$.

\emph{Routing.} At the current entry of an edge $e = (p, r) \in \mathcal{I}$,
the refined residue is $v_e + p i_e = u$ by
\equationref{eq:p3-entry-residue}, and the new value is $u$ itself: the upper
bound holds with equality, as the equality clause demands now that $u \in Q_m$.
Any other entry $\delta_p^r(q, i)$ of such a table is unchanged, and so is its
obligation: the residue it refines satisfies $q + i p \neq u$ --- the residue
determines the entry, since $q = (q + i p) \bmod p$ and $i = \floor{(q + i p) /
p}$ --- and hence lies in $Q_m^{\mathrm{new}}$ if and only if it lies in
$Q_m^{\mathrm{old}}$. Heredity at the edges into $m$ extends to the new state:
$u \bmod p = v_e \in Q_p$ for every edge $(p, r)$ into $m$, because every such
edge belongs to $\mathcal{I}$ --- this is where claim 4 of
\lemmaref{lem:p3-first-edge} is decisive --- and the source residues of the
edges of $\mathcal{I}$ are states by definition. At an interior creation,
consider the edges out of $m$. The new rows obey the upper bound,
$\delta_m^{r'}(u, j) = \delta_m^{r'}(q', j) \leq q' + j m < u + j m$; their
equality clauses are vacuous, because a state $w \in Q_{m r'}$ has $w \bmod m
\in Q_m^{\mathrm{old}}$ by old heredity, whereas $u \notin Q_m^{\mathrm{old}}$,
so no residue $u + j m$ is a state of $Q_{m r'}$; and heredity at these edges
is undisturbed, as $Q_{m r'}$ is unchanged and $Q_m$ only grew. Every other
edge, together with the state sets it references, is untouched. Routing thus
holds in the new configuration.

\emph{Past consistency.} Let $t < n$ and $\zeta \in \mathcal{C}_1$. A run of
Eval reads a row of $u$, or the output $\tau_m(u)$, only after arriving at the
scale $m$ in the state $u$, hence only after reading an entry of an edge into
$m$ whose value is $u$; the values of these tables are states of $Q_m$, so in
the old configuration there is no such entry, and in the new one the entries
with value $u$ are exactly the current entries. Suppose first that the run of
$\Eval_\zeta(t)$ does not cross the current entry of the edge of $\zeta$ into
$m$ --- the edge is unique, scales increasing strictly along a path, and this
case includes every $\zeta$ avoiding $m$; note also that the crossing
conditions are configuration-independent, as the route to a source scale $p$ is
computed from entries at scales below $p < m$, none of which is modified. Then
the run reads only entries whose values agree in the two configurations, it is
identical in both, and $\Eval_\zeta(t) = x[t]$ by old past consistency. Suppose
instead that the run crosses the current entry of $e = (p, r) \in \zeta$, $p r
= m$. At a sink this is impossible: claim 1 of
\lemmaref{lem:p3-current-entry} --- its hypothesis, the upper-bound part of
routing, holds in both configurations --- would give $\floor{t / m} < \floor{n
/ m} = 0$. At an interior creation, the same claim, applied in each
configuration, evaluates the run through \equationref{eq:p3-entry-factor} at
the common row $s \coloneqq \floor{t / m} < \ell$:
\[
  \Eval_\zeta^{\mathrm{old}}(t) = \PEval_\gamma(q_e, s), \quad
  \Eval_\zeta^{\mathrm{new}}(t) = \PEval_\gamma^{\mathrm{new}}(u, s) =
  \PEval_\gamma(q', s),
\]
where $\gamma \in \mathcal{C}_m$ is the continuation of $\zeta$ from $m$ and
the last equality is the locality identity. By
\equationref{eq:p3-sibling-data} --- applied once at the edge $e$ and once at
$e^*$ --- both partial evaluations equal $x[m s + u]$, so
$\Eval_\zeta^{\mathrm{new}}(t) = \Eval_\zeta^{\mathrm{old}}(t) = x[t]$, again
by old past consistency.

\emph{Moving horizon.} The sinks are unchanged and remain beyond $n$. At an
interior creation, no state set and no output function at a scale beyond the
horizon is touched, and enlarging $\mathcal{B}$ by the scale $m < n$ does not
affect the prescription, which consults membership in $\mathcal{B}$ at the
scales greater than $n$ alone: the claim persists verbatim. At a sink, the old
claim ($m \notin \mathcal{B}$) gives $Q_m^{\mathrm{old}} = \mathcal{F}_n$ and
$\tau_m(q) = x[q]$ on it, and the failure of the search of line 9 --- in which
the divergence condition $d_m(\tilde{q}, q') \geq \ell = 0$ holds vacuously, so
that every $\tilde{q} \in Q_m^{\mathrm{old}}$ was tested --- returned
$\Eval(\tilde{q}) \neq a$ for every $\tilde{q} \in \mathcal{F}_n$. By old past
consistency at the observed positions $\tilde{q} < n$, no first occurrence in
$x[:n]$ carries the symbol $a$; since every symbol occurring in $x[:n]$ does so
at a position of $\mathcal{F}_n$, the symbol $a$ does not occur in $x[:n]$ at
all, and $\mathcal{F}_{n+1} = \mathcal{F}_n \cup \{n\}$. Thus
$Q_m^{\mathrm{new}} = \mathcal{F}_n \cup \{u\} = \mathcal{F}_{n+1}$ --- exactly
what the claim prescribes at $m \in \mathcal{B} \cup \{m\}$ --- with $\tau_m(u)
= a = x[n]$ by line 13. Every other scale beyond the horizon, and its
membership in $\mathcal{B}$, is untouched. The moving-horizon claim thus holds
with $\mathcal{B} \cup \{m\}$; in particular, the output functions remain
injective at every sink, the symbols of $\mathcal{F}_{n+1}$ being distinct by
definition.

\emph{Relaxed distinguishability.} Routing, past consistency and injectivity at
the sinks now hold in both configurations, and the two share the scale
fragment; pinning (\corollaryref{cor:p3-pinning}) and stability
(\corollaryref{cor:p3-div-stability}) are available with $N = n - 1$, and
\lemmaref{lem:p3-div-semantics} applies in both configurations. Consider first
a pair of distinct old states $w_1, w_2 \in Q_k^{\mathrm{old}}$, at any scale
$k \in S$. If the pair is settled in the old configuration,
\corollaryref{cor:p3-div-stability} preserves its divergence value, and the two
clauses of \equationref{eq:p3-settled} --- statements about the data at
observed positions --- persist with it. If it is saturated,
\equationref{eq:p3-saturated} mentions only the data, the residue $n \bmod k$
and the symbol $a$, and persists verbatim. Since the state sets at the scales
other than $m$ are unchanged, the pairs that remain are those joining the new
state $u$ to a state $w \in Q_m^{\mathrm{old}}$.

At a sink, every such pair is saturated: $\max(w, u) = u = n$ is the residue $n
\bmod m$, and $x[m \cdot 0 + w] = x[w] \neq a$ was established with the moving
horizon. At an interior creation, consider first $w = q'$. Then $\max(u, q') =
u = n \bmod m$, and $x[m \ell + q'] \neq a$: by
\equationref{eq:p3-all-refuted} at the edge $e^*$, the run of Eval performed by
the test of line 8 did not return $a$, while by old past consistency at the
observed position $m \ell + q' < n$ it returned $x[m \ell + q']$. The pair
$\{u, q'\}$ is saturated. Finally, let $w \in Q_m^{\mathrm{old}} \setminus
\{q'\}$. The locality identities give $D_\gamma(u, w) = D_\gamma(q', w)$ for
every $\gamma \in \mathcal{C}_m$ in the new configuration, whence, minimizing
over $\gamma$ (\equationref{eq:p3-div-path}) and using the
configuration-independence of the partial evaluations of old states together
with \equationref{eq:p3-frozen-div},
\[
  \Delta_m^{\mathrm{new}}(u, w) = \Delta_m^{\mathrm{new}}(q', w) =
  \Delta_m^{\mathrm{old}}(q', w) = d_m(q', w) =: \ell'.
\]
The old pair $\{q', w\}$ is not saturated --- both of its members are states of
$Q_m^{\mathrm{old}}$, while $n \bmod m = u$ is not --- hence it is settled:
\[
  m \ell' + \max(q', w) \leq n - 1, \quad x[m \ell' + q'] \neq x[m \ell' + w].
\]
In particular $m \ell' \leq n - 1 < m (\ell + 1)$, so $\ell' \leq \ell$. If
$\ell' < \ell$, the pair $\{u, w\}$ is settled: its divergence value is
$\ell'$; the numeric clause holds, since $m \ell' + u \leq m (\ell - 1) + u = n
- m \leq n - 1$ and $m \ell' + w \leq m \ell' + \max(q', w) \leq n - 1$; and,
for the data clause, \equationref{eq:p3-sibling-data} at the edge $e^*$ and the
row $\ell'$, followed by pinning at the observed position $m \ell' + q' \leq n
- 1$, gives, for any $\gamma \in \mathcal{C}_m$,
\[
  x[m \ell' + u] = \PEval_\gamma(q', \ell') = x[m \ell' + q'] \neq x[m \ell' +
  w].
\]
If $\ell' = \ell$, the pair $\{u, w\}$ is saturated: the settled bound reads
$\max(q', w) \leq n - 1 - m \ell = u - 1$, so $w < u$ and $\max(u, w) = u = n
\bmod m$; and $w$ was examined by the search of line 9 --- its divergence
condition $d_m(w, q') = \ell' \geq \ell$ holds --- so the failure of the search
includes the failure of the test of $w$: $\Eval(m \ell + w) \neq a$ and, the
position $m \ell + w < n$ being observed, $x[m \ell + w] \neq a$.

Every pair of distinct states of the new configuration is thus settled or
saturated. All four parts of the loop invariant have been verified, the
enlargement of $\mathcal{B}$ having been accounted for in the moving-horizon
part, and the lemma is proved.
\end{proof}

\subsubsection{Completion of the Cascade}
\label{app:p3-cascade-complete}

The loop invariant has now been carried across every kind of iteration; what
remains is to chain the iterations together and to convert the invariant of the
completed loop into the four claims of \lemmaref{lem:p3-inv} --- for the
extended prefix $x[:n+1]$, in the configuration that the cascade hands to
post-processing. Two gaps separate the loop invariant from the four claims:
past consistency stops at the position $n - 1$, and relaxed distinguishability
tolerates saturated pairs. Consistency at the position $n$ itself
(\lemmaref{lem:p3-consistency-n}) closes both: it is the missing case of data
consistency, and it makes pinning available at the current row, turning the
saturation clause \equationref{eq:p3-saturated} into an observed difference of
subsequences. As in the preceding subsubsections, we consider a round $n \geq
1$ with observed symbol $a = x[n]$, assuming \lemmaref{lem:p3-inv} after the
preceding round; the \emph{post-cascade configuration} is the one reached when
the loop of line 3 of \algorithmref{alg:p3-update} has completed, immediately
before line 18. The passage from the post-cascade configuration to the
persistent state of the round is the business of
\sectionref{app:p3-pruning-proof}.

\begin{lemma}
\label{lem:p3-cascade-end}
Immediately before every iteration of the loop of line 3 of
\algorithmref{alg:p3-update}, and again when the loop has completed, the loop
invariant holds, with $\mathcal{B}$ the set of scales already processed; in
particular, it holds in the post-cascade configuration with $\mathcal{B} = S
\setminus \{1\}$.
\end{lemma}

\begin{proof}
The loop creates and deletes no scale and no edge --- its writes are the
assignments of lines 10, 12, 13, 14 and 17, to transition tables, state sets
and output functions --- so the scale fragment is the same in every
configuration of the cascade, and the loop performs exactly one iteration per
scale of $S \setminus \{1\}$, in increasing order. Immediately after line 2,
the loop invariant holds with $\mathcal{B} = \emptyset$: this is the conclusion
of the preprocessing analysis of \sectionref{app:p3-transient-inv}. And if the
invariant holds immediately before an iteration, at a scale $m$, then it holds
immediately after it with $\mathcal{B}$ replaced by $\mathcal{B} \cup \{m\}$:
by \lemmaref{lem:p3-reroute-iter} when the iteration does not execute line 12,
and by \lemmaref{lem:p3-create-iter} when it does. The lemma follows by
induction along the iterations.
\end{proof}

Recall from \sectionref{app:p3-cascade-reroute} the notation attached to an
arbitrary edge $e = (p, r) \in \mathcal{E}$ --- the target scale $m_e = p r$,
the source residue $v_e = n \bmod p$, the digit $i_e = \floor{n / p} \bmod r$
and the current entry $\delta_p^r(v_e, i_e)$ --- all of it determined by $e$
and the position $n$ alone. The next lemma reads the current entries off the
post-cascade configuration at the edges where the transition recursion of the
position $n$ steps off its exact residues; at such an edge, the entry points at
an observed occurrence of the fresh symbol at the current row.

\begin{lemma}
\label{lem:p3-final-entry}
In the post-cascade configuration, let $e = (p, r) \in \mathcal{E}$ be an edge
with $v_e \in Q_p$ and $n \bmod m_e \notin Q_{m_e}$. Then the current entry of
$e$ holds a value $w$ satisfying
\[
  w < n \bmod m_e \quad \text{and} \quad x[m_e \floor{n / m_e} + w] = a.
\]
\end{lemma}

\begin{proof}
Write $u' \coloneqq n \bmod m_e$ and $\ell' \coloneqq \floor{n / m_e}$. The
target of an edge belongs to $S$ and exceeds $1$, so the loop performed an
iteration at $m_e$.

\emph{The iteration at $m_e$.} A state set $Q_k$ changes during the cascade
only at the iteration at $k$: line 12 is the loop's only assignment to a state
set, and it adjoins the residue $n \bmod k$ of the scale being processed,
nothing ever being deleted. Since the scales are processed in increasing order
and $p < m_e$, the membership $v_e \in Q_p$ of the post-cascade configuration
was already in force when line 4 of the iteration at $m_e$ ran --- for $p \in S
\setminus \{1\}$ because the iteration at $p$ was complete, and for $p = 1$
because $v_e = 0$ and $Q_1 = \{0\}$ throughout --- so $e \in \mathcal{I}$ in
that iteration. Likewise, the only state that the iteration at $m_e$ could have
adjoined to $Q_{m_e}$ is $u'$ itself, which would then have survived to the
post-cascade configuration; since it did not, the iteration found $u' \notin
Q_{m_e}$ at line 5 and did not execute line 12. By
\lemmaref{lem:p3-cascade-end}, the loop invariant held immediately before that
iteration; claim 2 of \lemmaref{lem:p3-reroute-iter} therefore applies to it,
and at its end the current entry of $e$ held a value $w$ with $w < u'$ and
$x[m_e \ell' + w] = a$.

\emph{Stability.} The cascade writes transition-table entries at lines 10, 14
and 17 alone. At the iteration at a scale $k$, lines 10 and 17 write current
entries of edges of its set $\mathcal{I}$, which lie in the tables
$\delta_{p'}^{r'}$ with $p' r' = k$, while line 14 writes the rows of the state
just created in the tables of the edges out of $k$. The table $\delta_p^r$ is
therefore written only at the iteration at $p$ (line 14) and at the iteration
at $m_e$ (lines 10 and 17), and at no later point of the cascade; the current
entry of $e$ thus holds $w$ in the post-cascade configuration as well.
\end{proof}

We can now extend data consistency to the current position: in the post-cascade
configuration, every run of Eval at the position $n$ returns the symbol
observed there. Along any continuation from the root, either every residue of
$n$ en route is a state --- and then the run of $n$ tracks these residues into
a sink, where the moving horizon has recorded the fresh symbol --- or the run
steps off the residues of $n$ at some first edge, which
\lemmaref{lem:p3-final-entry} has pointed at an observed occurrence of that
symbol, and pinning takes over.

\begin{lemma}
\label{lem:p3-consistency-n}
In the post-cascade configuration, $\Eval_\zeta(n) = a$ for every $\zeta \in
\mathcal{C}_1$.
\end{lemma}

\begin{proof}
By \lemmaref{lem:p3-cascade-end}, the four parts of the loop invariant are
available, with $\mathcal{B} = S \setminus \{1\}$. Fix $\zeta = (m_0, r_0, m_1,
\dots, r_{L-1}, m_L) \in \mathcal{C}_1$ and write $u_j \coloneqq n \bmod m_j$
for $j \leq L$. The moving horizon places every sink beyond $n \geq 1$, so the
root is not a sink and $L \geq 1$; also $u_0 = 0 \in Q_1$. We distinguish two
cases.

\emph{Every residue along $\zeta$ is a state.} Suppose first that $u_j \in
Q_{m_j}$ for every $j \leq L$, and consider the last edge $e \coloneqq
(m_{L-1}, r_{L-1})$ of $\zeta$, with target scale $m_e = m_L$. Its source
residue $v_e = u_{L-1}$ is a state of $Q_{m_{L-1}}$, and routing holds along
the prefix of $\zeta$ ending at $m_{L-1}$; by claim 2 of
\lemmaref{lem:p3-current-entry}, the run of $\Eval_\zeta(n)$ crosses the
current entry of $e$ and, the continuation of $\zeta$ from the sink $m_L$ being
empty, \equationref{eq:p3-entry-factor} reduces to
\[
  \Eval_\zeta(n) = \tau_{m_L}(\delta_{m_{L-1}}^{r_{L-1}}(v_e, i_e)).
\]
By \equationref{eq:p3-entry-residue}, the residue refined by the current entry
is $u_L$ --- a state of $Q_{m_L}$, by the case hypothesis --- so the equality
clause of routing evaluates the entry at $u_L$ exactly. The sink $m_L$ lies
beyond $n$, whence $u_L = n$, and the moving horizon gives
\[
  \Eval_\zeta(n) = \tau_{m_L}(n) = x[n] = a.
\]

\emph{The run steps off the residues.} Otherwise, let $j^*$ be minimal with
$u_{j^*} \notin Q_{m_{j^*}}$; by the above, $1 \leq j^* \leq L$. The edge $e
\coloneqq (m_{j^* - 1}, r_{j^* - 1})$ has target scale $m_e = m_{j^*}$; its
source residue $v_e = u_{j^* - 1}$ is a state of $Q_{m_{j^* - 1}}$ by the
minimality of $j^*$, while $n \bmod m_e = u_{j^*} \notin Q_{m_e}$. Thus
\lemmaref{lem:p3-final-entry} evaluates the current entry of $e$ at a value $w$
with $w < u_{j^*}$ and $x[m_e \ell' + w] = a$, where $\ell' \coloneqq \floor{n
/ m_e}$; and, as in the previous case, claim 2 of
\lemmaref{lem:p3-current-entry} shows that the run of $\Eval_\zeta(n)$ crosses
that entry, \equationref{eq:p3-entry-factor} now reading $\Eval_\zeta(n) =
\PEval_\gamma(w, \ell')$, where $\gamma \in \mathcal{C}_{m_e}$ is the
continuation of $\zeta$ from $m_e$. The position $m_e \ell' + w$ is smaller
than $m_e \ell' + u_{j^*} = n$, hence observed; routing and past consistency
make pinning available with $N = n - 1$, and \corollaryref{cor:p3-pinning}
finishes:
\[
  \Eval_\zeta(n) = \PEval_\gamma(w, \ell') = x[m_e \ell' + w] = a.
\]
\end{proof}

The claims of \lemmaref{lem:p3-inv}, in the form they take before
post-processing, follow. Below, $\Delta_m$ denotes, as in
\sectionref{app:p3-transient-inv}, the divergence tables that a fresh call to
Diverge computes from the configuration under consideration --- here, the
post-cascade one.

\begin{lemma}
\label{lem:p3-pre-pruning}
In the post-cascade configuration, the four claims of \lemmaref{lem:p3-inv}
hold for the prefix $x[:n+1]$:
\begin{enumerate}
  \item for every $t \leq n$, every run of Eval$(t)$ returns $x[t]$, regardless
    of how $r$ is selected in line 3 of \algorithmref{alg:eval};
  \item both clauses of the routing claim hold at every edge of $\mathcal{E}$;
  \item for every $m \in S$ and distinct $q_1, q_2 \in Q_m$, the divergence
    value $\ell \coloneqq \Delta_m(q_1, q_2)$ satisfies $m \ell + \max(q_1,
    q_2) \leq n$ and $x[m \ell + q_1] \neq x[m \ell + q_2]$;
  \item for every $m \in S$ with $m > n$, $Q_m = \mathcal{F}_{n+1}$ and
    $\tau_m(q) = x[q]$ for every $q \in Q_m$; and every sink has scale greater
    than $n$.
\end{enumerate}
\end{lemma}

\begin{proof}
By \lemmaref{lem:p3-cascade-end}, the loop invariant holds in the post-cascade
configuration, with $\mathcal{B} = S \setminus \{1\}$.

\emph{Claims 2 and 4.} Claim 2 is the routing part of the loop invariant. A
scale $m > n \geq 1$ is a non-root scale, hence belongs to $\mathcal{B}$, and
the moving horizon prescribes $Q_m = \mathcal{F}_{n+1}$ and $\tau_m(q) = x[q]$
on it, keeping every sink beyond $n$; this is claim 4. In particular, the
output map of every sink is injective, the symbols at the positions of
$\mathcal{F}_{n+1}$ being pairwise distinct.

\emph{Claim 1.} The selections of line 3 of \algorithmref{alg:eval} drive the
run of Eval$(t)$ along edges from the root until it reaches a scale with no
outgoing edges: the scales and edges traversed form a continuation $\zeta \in
\mathcal{C}_1$, and the run returns $\Eval_\zeta(t)$. For $t < n$, past
consistency gives $\Eval_\zeta(t) = x[t]$; for $t = n$,
\lemmaref{lem:p3-consistency-n} gives $\Eval_\zeta(n) = a = x[n]$. Data
consistency thus holds through the index $N = n$ and --- together with routing
and the injectivity just noted --- makes pinning
(\corollaryref{cor:p3-pinning}) and \corollaryref{cor:p3-div-witness} available
with $N = n$.

\emph{Claim 3.} Let $m \in S$ and let $q_1 \neq q_2$ be states of $Q_m$; by
relaxed distinguishability, the pair is settled or saturated. A settled pair
satisfies claim 3 outright: \equationref{eq:p3-settled} bounds $m
\Delta_m(q_1, q_2) + \max(q_1, q_2)$ by $n - 1$ and exhibits the separating
symbols. Suppose then that the pair is saturated, and write $u \coloneqq
\max(q_1, q_2)$, $w \coloneqq \min(q_1, q_2)$ and $\ell \coloneqq \floor{n /
m}$: by \equationref{eq:p3-saturated}, $u = n \bmod m$ --- so that $m \ell + u =
n$ --- and $x[m \ell + w] \neq a$. Both positions $m \ell + u = n$ and $m \ell
+ w < n$ are within reach of pinning, so, for every $\gamma \in \mathcal{C}_m$,
\[
  \PEval_\gamma(u, \ell) = x[n] = a \neq x[m \ell + w] = \PEval_\gamma(w,
  \ell),
\]
that is, every continuation separates the pair at the row $\ell$: $D_\gamma(q_1,
q_2) \leq \ell$. By \lemmaref{lem:p3-div-semantics}, whose hypothesis is the
sink injectivity noted with claim 4, the value $\ell^* \coloneqq \Delta_m(q_1,
q_2)$ is the least of the pathwise divergences
(\equationref{eq:p3-div-path}), hence $\ell^* \leq \ell$, and the numeric
clause of claim 3 follows: $m \ell^* + \max(q_1, q_2) \leq m \ell + u = n$. And
since $\ell^* < \infty$, \corollaryref{cor:p3-div-witness} (with $N = n$) turns
the divergence value into separating symbols: $x[m \ell^* + q_1] \neq x[m
\ell^* + q_2]$.
\end{proof}

Together with the initialization paragraph of
\sectionref{app:p3-transient-inv}, \lemmaref{lem:p3-pre-pruning} carries the
induction over the rounds of an epoch as far as line 18 of
\algorithmref{alg:p3-update}. The remaining distance to \lemmaref{lem:p3-inv}
is post-processing: \sectionref{app:p3-pruning-proof} verifies that the
deletions of Post-process preserve the four claims, treats the re-indexing
performed by a restart, and closes the induction.

\subsection{Post-Processing and Completion}
\label{app:p3-pruning-proof}

We first isolate pruning from the restart that it may trigger. In the following
lemma, the \emph{unpruned configuration} is the post-cascade configuration of
\lemmaref{lem:p3-pre-pruning}, and an \emph{induced subconfiguration} means
that some scales and all incident edges are deleted, while every state set,
output map and transition table attached to a surviving scale or edge is left
unchanged.

\begin{lemma}
\label{lem:p3-pruning-stability}
Suppose that the unpruned configuration satisfies the four claims of
\lemmaref{lem:p3-pre-pruning}. Let $(\hat{S}, \hat{\mathcal{E}})$ be a nonempty
induced subconfiguration in which every scale other than $1$ is the target of a
surviving edge and every sink has scale greater than $n$. Then the four claims
of \lemmaref{lem:p3-inv} hold in the induced subconfiguration.
\end{lemma}

\begin{proof}
First note that every surviving scale is reachable from the root. Indeed, edges
strictly increase the scale. The minimal element of the nonempty set $\hat{S}$
must therefore be $1$: otherwise it would be a non-root scale with no incoming
edge. Repeatedly following a surviving incoming edge from any $m \in \hat{S}$
strictly decreases the scale and must end at $1$, producing an entry path.
Continuations exist as well, since the fragment is finite and acyclic.

\emph{Routing and the horizon.} Every surviving transition table and state set
is unchanged, so both clauses of routing, including heredity, restrict
immediately from the unpruned configuration to the surviving edges. Likewise,
if $m \in \hat{S}$ and $m > n$, \lemmaref{lem:p3-pre-pruning} gives $Q_m =
\mathcal{F}_{n+1}$ and $\tau_m(q) = x[q]$ for every $q \in Q_m$. The assumed
sink condition supplies the remaining part of the horizon claim. In particular,
the output map at every surviving sink is injective.

\emph{Data consistency.} Fix $t \leq n$ and a run of Eval$(t)$ in the induced
subconfiguration. Let $M$ be the sink at which it halts and $q \in Q_M$ the
state it reaches. By assumption $M > n$. The path followed up to $M$ also
exists in the unpruned fragment; extend it there by an arbitrary continuation
from $M$ to an old sink (using the empty continuation when $M$ was already a
sink). Every scale $k$ on the added tail satisfies $k \geq M > n$, and the old
horizon claim gives
\[
  Q_k = Q_{k r} = \mathcal{F}_{n+1}
\]
at every edge $(k, r)$ of that tail. Since $q \in Q_M$, induction along the
tail and the equality part of routing give $\delta_k^r(q, 0) = q$. Moreover,
the digit of $t$ at every such edge is $\floor{t / k} \bmod r = 0$. The
extended old run therefore keeps the state $q$ and returns $x[q]$ at its old
sink, while the truncated new run returns $\tau_M(q) = x[q]$. The two runs have
the same value, and the old data-consistency claim identifies that value with
$x[t]$. This proves data consistency after pruning, for every choice of path.

\emph{Distinguishability.} Write $d_m^{\mathrm{pre}}$ and
$d_m^{\mathrm{post}}$ for the divergence tables before and after the deletion.
Fix a surviving scale $m$ and distinct states $q_1, q_2 \in Q_m$, and put
\[
  \ell \coloneqq d_m^{\mathrm{pre}}(q_1, q_2).
\]
The unpruned distinguishability claim gives
\begin{equation}
  \label{eq:p3-prune-witness}
  m \ell + \max(q_1, q_2) \leq n \quad \text{and} \quad x[m \ell + q_1] \neq
  x[m \ell + q_2].
\end{equation}
We have already established routing and data consistency in the induced
subconfiguration, and every surviving scale has an entry path. Pinning
therefore identifies the partial evaluations along every surviving continuation
at row $\ell$ with the two data values in
\equationref{eq:p3-prune-witness}, which are different. By
\equationref{eq:p3-div-path},
\[
  d_m^{\mathrm{post}}(q_1, q_2) \leq \ell,
\]
and hence
\[
  m d_m^{\mathrm{post}}(q_1, q_2) + \max(q_1, q_2) \leq m \ell + \max(q_1, q_2)
  \leq n.
\]
The post-pruning divergence value is finite, so
\corollaryref{cor:p3-div-witness}, applied in the induced subconfiguration,
gives
\[
  x[m d_m^{\mathrm{post}}(q_1, q_2) + q_1] \neq x[m d_m^{\mathrm{post}}(q_1,
  q_2) + q_2].
\]
This is the distinguishability claim and completes the proof.
\end{proof}

We apply the lemma to Post-process. Each pass through line 3 of
\algorithmref{alg:p3-post} deletes at least one scale and adds nothing, so the
loop terminates. It leaves an induced subconfiguration of the post-cascade one.
If no capacity violation is found in line 1, no deletion occurs: before
post-processing, every non-root scale still has an incoming edge, because
preprocessing creates scales only together with such an edge and the cascade
deletes no edge, while every sink is beyond $n$ by
\lemmaref{lem:p3-pre-pruning}. Otherwise, when the loop stops with a nonempty
scale set, the emptiness of $\mathcal{D}$ says exactly that every surviving
non-root scale is the target of a surviving edge and that every surviving sink
is greater than $n$. Thus \lemmaref{lem:p3-pruning-stability} applies in either
case.

It remains to treat the case in which the loop deletes every scale. Let $N$ be
the old epoch's just-observed index --- the value of the counter upon entry to
Post-process. Line 5 of \algorithmref{alg:p3-post} doubles $c_{\max}$ and calls
Initialize$(a)$. Let $y$ be the suffix of the old epoch's input from that
position --- the input of the new epoch --- so $y[j] \coloneqq x[N+j]$ and
$y[0] = a$. Relative to this new epoch, Initialize produces
\[
  S = \{1\}, \quad \mathcal{E} = \emptyset, \quad Q_1 = \{0\}, \quad \tau_1(0)
  = y[0].
\]
The unique run of Eval$(0)$ returns $y[0]$; routing and distinguishability are
vacuous; and the only scale is a sink with $1 > 0$, whose state set is the
first-occurrence set $\mathcal{F}_1 = \{0\}$. Hence all four claims of
\lemmaref{lem:p3-inv} hold for the one-symbol prefix $y[:1]$. The increment in
line 19 of \algorithmref{alg:p3-update} merely sets the counter for the next
round to $1$, exactly as after an ordinary initialization. The doubled capacity
does not enter any of the four claims.

\begin{proofof}{\lemmaref{lem:p3-inv}}
We induct over the rounds, restarting the induction whenever a new epoch
begins. On the first round of an epoch, Initialize creates the one-scale
configuration just checked above (for the first epoch it is invoked by line 1
of \algorithmref{alg:p3-update}; for every later epoch it is invoked by line 5
of \algorithmref{alg:p3-post}). Now suppose the claims hold after the prefix
$x[:n]$ of the current epoch, where $n \geq 1$, and process the symbol $a =
x[n]$. The preprocessing and cascade arguments from
\sectionref{app:p3-transient-inv} through
\sectionref{app:p3-cascade-complete} culminate in
\lemmaref{lem:p3-pre-pruning}, so the post-cascade configuration satisfies all
four claims for $x[:n+1]$. If Post-process leaves a nonempty fragment,
\lemmaref{lem:p3-pruning-stability} preserves those claims. If it empties the
fragment, the restart verification above establishes them as the base case for
the new epoch. Line 19 then advances only the next-round counter. The induction
therefore covers every round of every epoch and proves \lemmaref{lem:p3-inv}.
\end{proofof}

\section{Mistake and Compression Bounds}
\label{app:p3-proof}

This appendix completes the proof of \theoremref{thm:p3}, taking the invariants
of \lemmaref{lem:p3-inv} (\appendixref{app:p3-correctness}) as given. We begin
with auxiliary results of three kinds: certifications of the bisection and of
the rounded weights by which the implementation of Predict computes its vote
(\sectionref{app:p3-exponent}), a lemma on plurality voting under a charging
budget (\sectionref{app:budgeted-plurality}), and the branch prediction game
with its mistake bound, recalled from \citep{Kosoy2026b}
(\sectionref{app:branch-prediction}). \sectionref{app:p3-state-bound} then
bounds the number of states at every scale, and \sectionref{app:p3-stat} and
\sectionref{app:p3-compression} prove the mistake bound
\equationref{eq:p3-mistakes} and the compression bound
\equationref{eq:p3-state}.

\subsection{Computing the Exponent}
\label{app:p3-exponent}

Recall the exponent $\eta \geq 0$ of the prediction weights
(\sectionref{sec:p3-predict}), the unique solution of $f(\eta) = 1$ for $f(t)
\coloneqq \sum_{r \in R} r^{-t}$ --- a continuous, strictly decreasing function
with $f(0) = \abs{R} \geq 1$ and limit $0$ --- and the subroutine Bisect
(\algorithmref{alg:p3-bisect}, \sectionref{app:p3-predict}), by which the
implementation of Predict approximates it (line 2 of
\algorithmref{alg:p3-predict-fp}). The following lemmas certify the subroutine,
and the calibration that \algorithmref{alg:p3-predict-fp} builds on it; as in
\sectionref{app:p3-predict}, $r_{\max} \coloneqq \max R$.

\begin{lemma}[Certified bisection]
\label{lem:p3-bisect}
For every rational $\epsilon \in (0, 1 / 2]$, every run of Bisect$(\epsilon)$
halts, within $O(\log(\abs{R} \log(r_{\max}) / \epsilon))$ iterations, and
returns a rational $\tilde{\eta} \in [0, \ceil{\log \abs{R}}]$ with
\[
  \abs{f(\tilde{\eta}) - 1} \leq \epsilon \quad \text{and} \quad
  \abs{\tilde{\eta} - \eta} \leq 2 \epsilon / \ln 2 \leq 3 \epsilon.
\]
\end{lemma}

\begin{proof}
Line 1 sets the endpoints to integers, and each iteration of line 5 replaces
one of them by their midpoint; so $l$ and $h$ stay dyadic rationals with $0
\leq l \leq h \leq \ceil{\log \abs{R}}$ --- in particular, the returned
$\tilde{\eta} = l$ is such --- and every iteration halves the width $h - l$.
From the initial width $\ceil{\log \abs{R}}$, the guard of line 2 therefore
fails within $O(\log(\abs{R} \log(r_{\max}) / \epsilon))$ iterations, and upon
exit
\[
  h - l \leq \epsilon / (2 \abs{R} \ceil{\log r_{\max}}).
\]
We use two bounds on the derivative $-f'(t) = \sum_{r \in R} r^{-t} \ln r$:
every arity lying between $2$ and $r_{\max}$, while $0 < r^{-t} \leq 1$ for $t
\geq 0$ and $\ln x \leq \log x$ for $x \geq 1$,
\[
  f(t) \ln 2 \leq -f'(t) \leq \abs{R} \ln r_{\max} \leq \abs{R} \ceil{\log
  r_{\max}}.
\]

\emph{Invariant.} Whenever line 2 is evaluated, $f(l) \geq 1 - \epsilon / 2$
and $f(h) \leq 1 + \epsilon / 2$. Initially, $f(l) = f(0) = \abs{R} \geq 1$,
and $f(h) \leq \abs{R} \, 2^{-\ceil{\log \abs{R}}} \leq 1$, since $f(t) \leq
\abs{R} 2^{-t}$, every arity being at least $2$. Line 5 preserves the property:
when it moves $l$ to $t$, line 4 returned $\tilde{f} \geq 1$, so $f(t) \geq
\tilde{f} - \epsilon / 2 \geq 1 - \epsilon / 2$; when it moves $h$, line 4
returned $\tilde{f} < 1$, so $f(t) \leq \tilde{f} + \epsilon / 2 < 1 + \epsilon
/ 2$.

\emph{Value accuracy.} Upon exit, the exit width and the upper derivative bound
give $f(l) - f(h) \leq (h - l) \cdot \abs{R} \ceil{\log r_{\max}} \leq \epsilon
/ 2$; combined with the invariant,
\[
  1 - \epsilon / 2 \leq f(l) \leq f(h) + \epsilon / 2 \leq 1 + \epsilon,
\]
so $\abs{f(\tilde{\eta}) - 1} \leq \epsilon$.

\emph{Root accuracy.} As $\epsilon \leq 1 / 2$, the monotone $f$ is at least
$\min(f(\tilde{\eta}), f(\eta)) \geq 1 - \epsilon \geq 1 / 2$ between
$\tilde{\eta}$ and $\eta$, so $-f' \geq \ln(2) / 2$ there, by the lower
derivative bound; hence $\epsilon \geq \abs{f(\tilde{\eta}) - f(\eta)} \geq
(\ln(2) / 2) \, \abs{\tilde{\eta} - \eta}$, i.e.\ $\abs{\tilde{\eta} - \eta}
\leq 2 \epsilon / \ln 2 \leq 3 \epsilon$.
\end{proof}

Lines 1--3 of \algorithmref{alg:p3-predict-fp} build on Bisect to calibrate the
fixed-point arithmetic of the prediction vote; the next lemma certifies the
parameters they produce.

\begin{lemma}[Implementation accuracy]
\label{lem:p3-fp-accuracy}
Consider a run of lines 1--3 of \algorithmref{alg:p3-predict-fp}, at a round
whose time index is $n \geq 1$. Then:
\begin{enumerate}
  \item $F = O(\log r_{\max} \cdot \log(r_{\max} n))$, and
    \begin{equation}
      \label{eq:p3-precision}
      8 \abs{R} \log(2 r_{\max} n) \cdot u \leq (r_{\max} n)^{-\log r_{\max}};
    \end{equation}
  \item $\hat{\eta}$ is rational, with $0 \leq \hat{\eta} - \eta \leq 6
    \epsilon$;
  \item for every $r \in R$: $\rho_r \in [0, r^{-\eta}]$ and $r^{-\eta} -
    \rho_r \leq u$.
\end{enumerate}
\end{lemma}

\begin{proof}
\emph{Claim 1.} Abbreviate $a \coloneqq \ceil{\log r_{\max}}$ and $b \coloneqq
\ceil{\log(2 r_{\max} n)}$, so that line 1 sets $F = (a + 1) b + 2$; since $a <
\log(r_{\max}) + 1$ and $b < \log(r_{\max} n) + 2$, the order of growth
follows. As for \equationref{eq:p3-precision}, both of its sides are positive,
$\log u = -F$, and $\log((r_{\max} n)^{-\log r_{\max}}) = -\log r_{\max} \cdot
\log(r_{\max} n)$; taking logarithms and rearranging, the claimed inequality
becomes
\[
  \log(8 \abs{R} \log(2 r_{\max} n)) + \log r_{\max} \cdot \log(r_{\max} n)
  \leq F.
\]
We bound the two summands in turn. For the first, $\log 8 = 3$; $\log \abs{R}
\leq \log r_{\max} \leq a$, since $R \subseteq \{2, \dots, r_{\max}\}$; and
$\log(2 r_{\max} n) \leq b \leq 2^{b - 1}$ --- the latter because $b$ is a
positive integer --- whence $\log \log(2 r_{\max} n) \leq b - 1$: altogether,
$\log(8 \abs{R} \log(2 r_{\max} n)) \leq 3 + a + (b - 1)$. For the second, $0
\leq \log r_{\max} \leq a$ and $0 \leq \log(r_{\max} n) = \log(2 r_{\max} n) -
1 \leq b - 1$, so $\log r_{\max} \cdot \log(r_{\max} n) \leq a (b - 1)$.
Summing, the left-hand side of the displayed inequality is at most
\[
  (3 + a + (b - 1)) + a (b - 1) = a b + b + 2 = (a + 1) b + 2 = F.
\]

\emph{Claim 2.} The accuracy $\epsilon = u / (16 \ceil{\log r_{\max}})$ set by
line 1 is rational and at most $1 / 2$, so \lemmaref{lem:p3-bisect} applies:
Bisect$(\epsilon)$ returns a rational within $3 \epsilon$ of $\eta$, and the
$\hat{\eta}$ of line 2, rational like $\epsilon$, satisfies $0 \leq \hat{\eta}
- \eta \leq 6 \epsilon$.

\emph{Claim 3.} By claim 2, $0 \leq r^{-\eta} - r^{-\hat{\eta}} \leq
(\hat{\eta} - \eta) \ln r \leq 6 \epsilon \ceil{\log r_{\max}} \leq u / 2$, and
line 3 concedes at most $2^{-F - 1} = u / 2$ more; moreover $\rho_r \leq
r^{-\hat{\eta}} \leq r^{-\eta}$ and $\rho_r \geq 0$, by the specification of
that line.
\end{proof}

\subsection{Budgeted Plurality}
\label{app:budgeted-plurality}

The following lemma abstracts the argument by which the mistakes of $P^3$ are
bounded in \sectionref{app:p3-stat}. Its setting is that of prediction with
expert advice \citep{Cesa-Bianchi2006}, with two features dictated by our
application, in which the experts are the infinite paths of
\sectionref{sec:p3-predict}: the learner predicts by a plurality vote among the
experts currently admitted to the vote, under a fixed prior and without
reweighting --- there are too many experts to carry individual weights --- and
the set of admitted experts shrinks over time by a rule that is arbitrary as
far as the lemma is concerned. What the lemma asks is that no expert be charged
with a wrong vote or an abstention more than a bounded number of times while
admitted, and that the admitted experts retain some prior mass. The prediction
is required to attain the plurality only up to a tolerance, because $P^3$
evaluates its vote in finite-precision arithmetic
(\sectionref{app:p3-predict}).

\textbf{Setting.} Let $(H, w)$ be a probability space, whose points we call
\emph{experts} and whose measure $w$ we call the \emph{prior}. A \emph{play} of
$T$ rounds consists, for every round $t = 1, \dots, T$, of a measurable set
$H_t \subseteq H$ of \emph{live} experts, a measurable \emph{vote} $v_t : H_t
\to \Sigma \cup \{\bot\}$ (the value $\bot$ representing abstention), an
\emph{outcome} $y_t \in \Sigma$, a \emph{tolerance} $\epsilon_t \geq 0$, and a
\emph{prediction} $\hat{y}_t \in \Sigma$ that maximizes the voting mass up to
the tolerance: $w\{h \in H_t \mid v_t(h) = \hat{y}_t\} \geq w\{h \in H_t \mid
v_t(h) = a\} - \epsilon_t$ for every $a \in \Sigma$. We say that the expert $h$
is \emph{charged} at round $t$ if $h \in H_t$ and $v_t(h) \neq y_t$ --- whether
it voted wrongly or abstained --- and that round $t$ is a \emph{mistake} if
$\hat{y}_t \neq y_t$.

\begin{lemma}[Budgeted plurality]
\label{lem:budgeted-plurality}
Consider a play of $T$ rounds and a real number $K \geq 1$ s.t.:
\begin{enumerate}
  \item the live sets decrease: $H_{t+1} \subseteq H_t$ for every $t < T$;
  \item the tolerances are commensurate with the live mass: $2 \epsilon_t \leq
    w(H_t)$ for every $t \leq T$;
  \item every expert is charged at most $K$ times;
  \item $w(H_T) > 0$.
\end{enumerate}
Then the number of mistakes is at most $8 K \ln(2 / w(H_T))$.
\end{lemma}

\begin{proof}
For $t = 1, \dots, T + 1$, let $\kappa_t(h)$ be the number of rounds $t' < t$
at which $h$ is charged, and consider the potential
\[
  \Phi_t \coloneqq \int_{H_t} (2K - \kappa_t(h)) \,\mathrm{d}w(h),
\]
where $\Phi_{T+1}$, the potential after the last round, is taken over the last
live set: $H_{T+1} \coloneqq H_T$. By hypothesis 3, $\kappa_t \leq K$, so the
integrand lies between $K$ and $2K$; hence $\Phi_1 = 2K w(H_1) \leq 2K$, while
$\Phi_{T+1} \geq K w(H_{T+1}) = K w(H_T)$. Fix $t \leq T$, and let $C_t$ be the
set of experts charged at round $t$, so that $\kappa_{t+1} - \kappa_t$ is the
indicator of $C_t$. Since $H_{t+1} \subseteq H_t$ (for $t = T$ by the
convention above) and the integrand is positive,
\[
  \Phi_{t+1} \leq \int_{H_t} (2K - \kappa_{t+1}(h)) \,\mathrm{d}w(h) = \Phi_t -
  w(C_t).
\]
Suppose that round $t$ is a mistake. The experts of $H_t$ left uncharged at $t$
are exactly those voting for $y_t$; the experts voting for $\hat{y}_t \neq y_t$
are all charged, and their mass falls short of the mass of the experts voting
for $y_t$ by at most $\epsilon_t$, by the tolerance property of the prediction.
Hence $w(H_t) - w(C_t) \leq w(C_t) + \epsilon_t$, so $w(C_t) \geq (w(H_t) -
\epsilon_t) / 2 \geq w(H_t) / 4 \geq \Phi_t / (8K)$, the middle inequality by
hypothesis 2 and the last because the integrand is at most $2K$. Thus every
mistake shrinks the potential by the factor $1 - 1 / (8K) \leq e^{-1 / (8K)}$,
and a play with $M$ mistakes satisfies
\[
  K w(H_T) \leq \Phi_{T+1} \leq e^{-M / (8K)} \Phi_1 \leq 2K e^{-M / (8K)},
\]
which rearranges to the claim.
\end{proof}

The lemma is a budgeted variant of the Halving algorithm
\citep{Cesa-Bianchi2006}, which is the case $K = 1$, with exact predictions
($\epsilon_t = 0$) and experts discarded at their first wrong vote: the live
set is then the version space, and an expert consistent with all outcomes keeps
$w(H_T) \geq 1 / \abs{H}$ under the uniform prior on a finite $H$, whence
$O(\log \abs{H})$ mistakes. More generally, discarding experts at their $K$-th
wrong vote yields $O(K \log \abs{H})$ mistakes whenever some expert votes
wrongly fewer than $K$ times. This is weaker than the $O(K + \log \abs{H})$
mistakes of the weighted majority algorithm \citep{Littlestone1994}, which
reweights the experts multiplicatively; but the plurality rule keeps no
per-expert state, which is what allows $P^3$ to aggregate, in polylogarithmic
space, over the exponentially many paths of its scale fragment.

\subsection{Branch Prediction}
\label{app:branch-prediction}

The last auxiliary ingredient is the branch prediction game of
\citep{Kosoy2026b}, an online prediction game on a growing rooted tree: in
\sectionref{app:p3-stat-edge}, the wrong forecasts of a single edge of the
scale fragment are bounded by coupling them to the mistakes of a fixed strategy
in one instance of this game. For the reader's convenience, we reproduce the
definitions and the mistake bound; the proof of the latter is that of
Proposition A.3 of \citep{Kosoy2026b}.

\begin{definition}[Branch prediction game]
\label{def:branch-prediction}
The \emph{branch prediction game} is a sequential game between a learner and
nature. Its state consists of a finite set $E$ of \emph{elements}, a rooted
tree, and an assignment of each element to a vertex of the tree; initially, $E$
is empty and the tree consists of a single root vertex. On each round, nature
makes one of the following moves.
\begin{itemize}
  \item \textbf{Introduction.} Nature extends the tree by a --- possibly empty
    --- chain of new vertices hanging below an existing vertex, adjoins a new
    element to $E$, and places it at the end of the chain (at the existing
    vertex itself, when the chain is empty).
  \item \textbf{Activation, of an element located at a leaf.} Nature attaches a
    new child to the leaf and moves the element to it; no prediction is made.
  \item \textbf{Activation, of an element located at an internal vertex $v$.}
    The learner predicts a child of $v$. Nature then moves the element to a
    child of $v$ --- an existing one, or a new child attached now --- and a
    \emph{mistake} occurs when the child taken differs from the child
    predicted.
\end{itemize}
\end{definition}

An activation moves an element to a child of its current vertex, so an element
that has entered the subtree of a vertex never leaves it; in particular, the
number $N(v)$ of elements located at a vertex $v$ or below it is nondecreasing
over the play.

\begin{definition}[Plurality Branch]
\label{def:pb}
The \emph{Plurality Branch} (PB) strategy predicts, at every activation at an
internal vertex $v$, a child $c$ of $v$ maximizing $N(c)$, breaking ties
arbitrarily.
\end{definition}

Note that the element being activated is located at $v$ itself, below no child
of $v$, so it contributes to no $N(c)$. The following mistake bound is
Proposition A.3 of \citep{Kosoy2026b}.

\begin{proposition}[Branch prediction bound]
\label{prop:branch-prediction}
In every play of the branch prediction game, PB makes at most $3 \abs{E}
\floor{\log \abs{E}}$ mistakes, where $\abs{E}$ is the final number of
elements.
\end{proposition}

\subsection{State Bound}
\label{app:p3-state-bound}

Fix a finite nonempty $R \subseteq \{2, 3, 4, \dots\}$, put $r_{\max} \coloneqq
\max R$ and $\sigma \coloneqq \abs{\Sigma}$, and let $x \in \Sigma^N$, $N \geq
2$, with $s \coloneqq \LZLC_R(x)$. We prove that every retained scale has $O(s
+ \sigma)$ states. After every complete update, Post-process ensures $\abs{Q_m}
\leq c_{\max}$ at each retained scale, so it remains to bound the capacity.
During an epoch, let $y$ be the suffix processed in that epoch, let $G$ be a
small exact layered ZLP with $y \sqsubseteq \val(G)$, and put $p_j \coloneqq
m_G(j)$. For $j < d(G)$, when the local time reaches $p_j$, Preprocess creates
the edge $(p_j, \ar_G(j))$ to $p_{j+1}$; we show that the resulting initial
segment of this scale path survives every subsequent update. At each $p_j$ with
$j < d(G)$, distinguishability injects $Q_{p_j}$ into layer $j$ of $G$. At
$p_{d(G)}$, it makes $q \mapsto y[q]$ injective on the stored states. These
injections bound the protected state sets by $\max(w(G), \sigma)$. Once
$c_{\max}$ is at least this bound, the capacity sweep cannot delete a protected
scale. The protected incoming and outgoing edges prevent the structural sweep
from deleting an interior scale, while the current endpoint is greater than the
local time and therefore cannot be deleted as a sink. In particular, the root
survives post-processing, so the epoch does not restart and the capacity is not
doubled again.

We first construct a uniformly small program for every possible epoch suffix.
Two details require care. A width-$s$ program for $x$ need not be exact;
\lemmaref{lem:reg} provides exactness at the cost of $\sigma$ extra width.
Moreover, the original program represents position $0$ of $x$ at its root,
whereas the program for the suffix $x[b:N]$ must represent position $b$ there.
The following lemma implements this shift by one mixed-radix carry bit.

\begin{lemma}
\label{lem:p3-factor-program}
For every $b < N$, there is an exact $R$-adic layered ZLP $G_b$ such that, for
$y \coloneqq x[b:N]$,
\[
  y \sqsubseteq \val(G_b) \quad \text{and} \quad w(G_b) \leq 2(s + \sigma).
\]
\end{lemma}

\begin{proof}
Let $H$ be a width-$s$ $R$-adic layered ZLP with $x \sqsubseteq \val(H)$.
Applying \lemmaref{lem:reg} with $n \coloneqq N$ yields an exact $R$-adic
layered ZLP $H'$ with $x \sqsubseteq \val(H')$ and $w(H') \leq s + \sigma$.
Write $Q'_j$ for its layers, $q'_0$ for its root, $\delta'$ for its transition
function, $d^* \coloneqq d(H')$, $r_j \coloneqq \ar_{H'}(j)$ for $j < d^*$, and
$M \coloneqq \abs{\val(H')} = m_{H'}(d^*)$; then $M \geq N$, since $x
\sqsubseteq \val(H')$.

We now shift the origin from $0$ to $b$. For $j < d^*$, let
\[
  b_j \coloneqq \floor{b / (m_{H'}(j))} \bmod r_j
\]
be the mixed-radix digits of $b < M$ in the numeration system of $H'$. Choose
$a_0 \coloneqq y[0]$. Construct $G_b$ by pairing the vertices of $H'$ with a
carry bit and retaining only pairs reachable from the root $(q'_0, 0)$. Thus
the layer of $G_b$ at depth $j$ is a subset of $Q'_j \times \{0, 1\}$. Write
$\delta_b$ for the transition function of $G_b$. For every $j < d^*$, reachable
pair $(q, c)$ at depth $j$, and digit $i < r_j$, put
\[
  u \coloneqq i + b_j + c, \quad e \coloneqq u \bmod r_j, \quad c' \coloneqq
  \floor{u / r_j},
\]
and define
\[
  \delta_b((q, c))[i] \coloneqq \begin{cases}
    (\delta'(q)[e], c') & \text{if } j < d^* - 1, \\
    \delta'(q)[e] & \text{if } j = d^* - 1 \text{ and } c' = 0, \\
    a_0 & \text{if } j = d^* - 1 \text{ and } c' = 1.
  \end{cases}
\]
Exactness of $H'$ ensures $\delta'(q)[e] \in Q'_{j+1}$ when $j < d^* - 1$ and
$\delta'(q)[e] \in \Sigma$ when $j = d^* - 1$. Also, $u \leq 2 r_j - 1$ gives
$c' \in \{0, 1\}$. Hence every production before the last layer has only vertex
children, and every production at the last layer has only terminal children.
Thus $G_b$ is a layered ZLP with the arity profile of $H'$, hence $R$-adic, and
it is exact by \lemmaref{lem:exact}; moreover,
\[
  w(G_b) \leq 2 w(H') \leq 2(s + \sigma).
\]
For $t < \abs{y}$, we have $b + t < N \leq M$. The recurrence defining $e$ and
$c'$ therefore computes the mixed-radix digits of $b + t$ without an overflow
beyond the last digit. Consequently, the descent of $t$ in $G_b$ follows the
descent of $b + t$ in $H'$, and
\[
  \val(G_b)[t] = \val(H')[b + t] = x[b + t] = y[t].
\]
Thus $y \sqsubseteq \val(G_b)$, as required.
\end{proof}

The following lemma supplies the state-set bound required by the capacity
argument. For $j < d(G)$, descent maps each stored residue at scale $m_G(j)$ to
a vertex in layer $j$ of the exact program, and distinguishability makes this
map injective. At the terminal scale $m_G(d(G))$, distinguishability instead
makes the map $q \mapsto y[q]$ injective on the stored states. We include this
terminal scale because it is the beyond-horizon endpoint that keeps the
protected path from being pruned as a sink at or below the local time.

\begin{lemma}
\label{lem:p3-path-state-bound}
Let $y \in \Sigma^L$, let $G$ be an exact layered ZLP with $y \sqsubseteq
\val(G)$, and put $p_j \coloneqq m_G(j)$ for $0 \leq j \leq d(G)$. Consider a
post-cascade configuration of $P^3$ after observing $y[:n+1]$, where $n < L$,
that satisfies the distinguishability claim of
\lemmaref{lem:p3-pre-pruning}. Then, whenever $p_j \in S$,
\[
  \abs{Q_{p_j}} \leq \begin{cases}
    w(G) & \text{if } j < d(G), \\
    \sigma & \text{if } j = d(G).
  \end{cases}
\]
\end{lemma}

\begin{proof}
Fix $j < d(G)$. Exactness implies that the descent of every $q \in [p_j]$
reaches a well-defined vertex $\psi_j(q)$ of layer $j$. We show that $\psi_j$
is injective on $Q_{p_j}$. Suppose otherwise that distinct $q_1, q_2 \in
Q_{p_j}$ reach the same vertex. By distinguishability, there is $\ell \in
\mathbb{N}$ such that
\begin{equation}
  \label{eq:p3-path-separation}
  p_j \ell + \max(q_1, q_2) \leq n \quad \text{and} \quad y[p_j \ell + q_1]
  \neq y[p_j \ell + q_2]
\end{equation}
On the other hand, applying \equationref{eq:descent} at layer $j$ to both
positions in \equationref{eq:p3-path-separation} gives
\[
  y[p_j \ell + q_1] = \val_G(\psi_j(q_1))[\ell] = \val_G(\psi_j(q_2))[\ell] =
  y[p_j \ell + q_2],
\]
a contradiction. Hence $\psi_j$ injects $Q_{p_j}$ into layer $j$ of $G$, so
$\abs{Q_{p_j}} \leq w(G)$.

It remains to consider $j = d(G)$. Since $p_j = \abs{\val(G)} \geq L > n$, the
inequality in the distinguishability claim rules out every row $\ell \geq 1$.
Therefore, for any distinct $q_1, q_2 \in Q_{p_j}$, distinguishability at row
$0$ gives $q_1, q_2 \leq n$ and $y[q_1] \neq y[q_2]$. Thus $q \mapsto y[q]$ is
injective on $Q_{p_j}$, so $\abs{Q_{p_j}} \leq \sigma$.
\end{proof}

Set
\[
  c^* \coloneqq 2(s + \sigma).
\]

\begin{lemma}
\label{lem:p3-no-large-restart}
If an epoch begins with $c_{\max} \geq c^*$, it processes the entire remaining
suffix of $x$ without restarting.
\end{lemma}

\begin{proof}
Let $b$ be the global index of the epoch's first symbol, put $y \coloneqq
x[b:N]$, and apply \lemmaref{lem:p3-factor-program} to obtain $G_b$. Write $d
\coloneqq d(G_b)$ and $p_j \coloneqq m_{G_b}(j)$ for $0 \leq j \leq d$.

Consider the round with local index $n < \abs{y}$, in which the algorithm
observes $y[n]$. We maintain a protected initial segment
\[
  p_0 \to p_1 \to \dots \to p_k
\]
of the scale path of $G_b$ such that its endpoint satisfies $p_k > n$
immediately after Preprocess. In the first round, this segment consists only of
$p_0 = 1 > 0$. Inductively, suppose that the segment survived the preceding
round. Its endpoint then satisfies $p_k \geq n$ in the current round. If $n <
p_k$, Preprocess preserves the segment because it only adds scales and edges.
If $n = p_k$, then $k < d$, because $n < \abs{y} \leq p_d$; Preprocess creates
the edge $(p_k, \ar_{G_b}(k))$ and its target $p_{k+1} > n$. Thus the claimed
segment exists after Preprocess in either case. The update cascade changes no
scales or edges, so the segment is still present when post-processing begins.

After the cascade, the configuration satisfies the distinguishability claim of
\lemmaref{lem:p3-pre-pruning}, so \lemmaref{lem:p3-path-state-bound}, with the
width bound of \lemmaref{lem:p3-factor-program}, gives $\abs{Q_{p_j}} \leq
\max(w(G_b), \sigma) \leq c^* \leq c_{\max}$ at every scale of the protected
segment. Hence the capacity sweep in line 1 of \algorithmref{alg:p3-post}
selects none of them. The structural sweep cannot select one either: every
$p_j$ with $1 \leq j \leq k$ has the protected incoming edge from $p_{j-1}$;
every $p_j$ with $j < k$ has the protected outgoing edge to $p_{j+1}$; and the
endpoint $p_k$ is greater than the local index $n$, so it is not a sink at or
below the local time. Therefore the entire protected segment survives
post-processing. This proves the induction, and in particular the root scale
$p_0$ survives every round of the epoch. Consequently line 5 of
\algorithmref{alg:p3-post} never fires during the epoch.
\end{proof}

\begin{lemma}
\label{lem:p3-state-bound}
Throughout the processing of $x$, after every complete update,
\[
  c_{\max} < 4(s + \sigma) \quad \text{and} \quad \abs{Q_m} < 4(s + \sigma)\
  \text{for every } m \in S.
\]
\end{lemma}

\begin{proof}
The capacity begins at $1$ and changes only when a restart doubles it. By
\lemmaref{lem:p3-no-large-restart}, a restart cannot occur in an epoch whose
initial capacity is at least $c^*$. If the capacity ever first reaches or
exceeds $c^*$, its value immediately before that doubling was smaller than
$c^*$; its new value is therefore smaller than $2 c^* = 4(s + \sigma)$, and it
remains fixed thereafter. If the capacity never reaches $c^*$, it is smaller
than $c^* < 4(s + \sigma)$ throughout. Finally, after every complete update,
Post-process has removed every scale whose state set exceeded the current
capacity. Hence $\abs{Q_m} \leq c_{\max} < 4(s + \sigma)$ for every retained
scale $m$.
\end{proof}

In particular, the number of states at every retained scale is $O(s + \sigma)$,
with an absolute implied constant --- that is, $O(s)$ for a fixed alphabet.
This is the form of the state bound used in \sectionref{app:p3-stat} and
\sectionref{app:p3-compression}.

\subsection{Statistical Efficiency}
\label{app:p3-stat}

In this section we prove the mistake bound \equationref{eq:p3-mistakes} by
applying \lemmaref{lem:budgeted-plurality} within each epoch, with the infinite
paths of \sectionref{sec:p3-predict} as the experts and the mass defined there
as the prior. Three things are to be verified. The vote of
\algorithmref{alg:p3-predict} is the plurality vote of the unblocked paths ---
attained, by the finite-precision implementation, up to a tolerance the lemma
accommodates --- whose set only shrinks within an epoch while its mass stays
bounded below (\sectionref{app:p3-stat-vote}). Every vote is the forecast of a
leaf edge, and the forecasts of a single edge couple, for the duration of its
lifetime, to one instance of the branch prediction game
(\sectionref{app:branch-prediction}), which bounds the number of wrong
forecasts of the edge in terms of the capacity
(\sectionref{app:p3-stat-edge}). Since a path votes only through the edges of
its own prefix, every path is charged a bounded number of times, and the lemma
bounds the mistakes of an epoch; summing over the epochs, whose capacities form
a geometric progression bounded via \lemmaref{lem:p3-state-bound}, yields the
claim (\sectionref{app:p3-stat-mistakes}).

\textbf{Setting and conventions.} We retain the notation of
\sectionref{app:p3-state-bound} --- $R$, $r_{\max} \coloneqq \max R$, $\sigma$,
$x \in \Sigma^N$ and $s \coloneqq \LZLC_R(x)$ --- and $\eta \geq 0$ is the
exponent of \sectionref{sec:p3-predict}, the unique solution of $\sum_{r \in R}
r^{-\eta} = 1$; note that $\eta \leq \log \abs{R}$, since $\abs{R} 2^{-\eta}
\geq \sum_{r \in R} r^{-\eta} = 1$. Fix an epoch, beginning at the global index
$b$; as in \appendixref{app:p3-correctness}, we re-index its suffix $y
\coloneqq x[b:N]$ from $0$. The capacity $c \coloneqq c_{\max}$ is constant
during the epoch (\sectionref{app:p3-post}). Let $T$ be the local index of the
last round of the epoch --- the round whose post-processing restarts the
predictor, or the last round of $x$ --- so that $T \leq \abs{y} - 1$; we assume
$T \geq 1$, an epoch with $T = 0$ making no predictions. The \emph{prediction
rounds} of the epoch are $n = 1, \dots, T$: the prediction at its round $0$ is
made by the previous epoch or, for the first epoch, by the empty state
(\sectionref{app:p3-init}), at the cost of at most one mistake. At the
prediction round $n$, Predict runs on the state obtained after updating on
$y[:n]$; when needed, we index the algorithmic quantities by the round, writing
$S_n$, $\mathcal{E}_n$, $Q_m^{(n)}$, $\mathcal{A}_n$, $\mathcal{L}_n$,
$\phi_n$, $\theta_n$ and $\pi_n$, the last two being the ideal quantities that
lines 10--13 of \algorithmref{alg:p3-predict} define
(\sectionref{app:p3-predict}). The last observed index is then $n - 1$, so
\lemmaref{lem:p3-inv} applies with its $n$ replaced by $n - 1$. In particular,
every state of every scale is a position below $n$: at a scale $m \leq n - 1$
because states are residues modulo $m$, and at a scale $m \geq n$ by the
horizon claim. We write $\hat{y}_n$ for the prediction at round $n$, as in
\sectionref{sec:protocol}.

\subsubsection{Paths, Mass and the Vote}
\label{app:p3-stat-vote}

An \emph{infinite path} is a sequence $\xi = (\xi_0, \xi_1, \dots) \in
R^\omega$ of arities; its \emph{scales} are $m_j(\xi) \coloneqq \prod_{i < j}
\xi_i$ for $j \in \mathbb{N}$, so that $m_0(\xi) = 1$, and its \emph{edges} are
the pairs $(m_j(\xi), \xi_j)$. Its \emph{walk} from a scale $m$ is the path $m
\to m m_1(\xi) \to m m_2(\xi) \to \dots$ through the lattice $R^!$. The
\emph{mass} $\mu$ is the probability measure on $R^\omega$ under which the
coordinates are independent with $\mu(\xi_j = r) = r^{-\eta}$: the law of the
random walk of \sectionref{sec:p3-predict}. For a finite sequence $\alpha \in
R^J$, the \emph{cylinder set} $Z(\alpha) \coloneqq \{\xi \in R^\omega \mid
\xi_j = \alpha_j \text{ for all } j < J\}$ of the paths beginning with $\alpha$
has mass
\begin{equation}
  \label{eq:p3-cylinder-mass}
  \mu(Z(\alpha)) = \prod_{j < J} \alpha_j^{-\eta} = (\alpha_0 \alpha_1 \cdots
  \alpha_{J-1})^{-\eta}.
\end{equation}
Thus the paths whose walk from $m$ begins with a given walk from $m$ to $m'$
have mass $(m' / m)^{-\eta}$, whichever factorization of $m' / m$ that walk
traverses. We use the continuations of
\sectionref{app:p3-path-semantics} relative to the fragment of a given round:
for a prediction round $n$ and $m \in S_n$, $\mathcal{C}_m^{(n)}$ denotes the
set of continuations from $m$ in $(S_n, \mathcal{E}_n)$ --- the directed paths
of the fragment from $m$ to a sink --- and, for $\gamma = (m_0, r_0, \dots,
r_{L-1}, m_L) \in \mathcal{C}_m^{(n)}$, we put $Z(\gamma) \coloneqq Z((r_0,
\dots, r_{L-1}))$: the set of paths whose walk from $m$ follows $\gamma$, of
mass $(m_L / m)^{-\eta}$ by \equationref{eq:p3-cylinder-mass}.

The following facts about the evolution of the scale fragment within an epoch
are read off \algorithmref{alg:p3-init}, \algorithmref{alg:p3-preprocess},
\algorithmref{alg:p3-update} and \algorithmref{alg:p3-post}.

\begin{lemma}[Fragment dynamics]
\label{lem:p3-fragment-dynamics}
Let $1 \leq n \leq T$ be a prediction round of the epoch.
\begin{enumerate}
  \item Every edge $(m, r) \in \mathcal{E}_n$ satisfies $m \leq n - 1$;
    consequently, a scale $m \in S_n$ is a sink of the fragment if and only if
    $m \geq n$. Moreover, every scale of $S_n$ is at most $r_{\max} n$.
  \item $1 \in S_n$ and every other scale of $S_n$ is the target of an edge of
    $\mathcal{E}_n$; consequently, every scale of $S_n$ has an entry path.
  \item \textbf{(Single lifetime.)} For every $m \in R^!$ and $r \in R$, the
    set of prediction rounds $n'$ with $(m, r) \in \mathcal{E}_{n'}$ is an
    interval $\{n' \mid m + 1 \leq n' \leq n_e\}$ for some $n_e \leq T$,
    possibly empty.
  \item \textbf{(Monotone growth.)} If $m \in S_{n'}$ for every prediction
    round $n'$ with $n \leq n' \leq n''$, then $Q_m^{(n)} \subseteq
    Q_m^{(n'')}$. Moreover, $\abs{Q_m^{(n)}} \leq c$ for every $m \in S_n$.
\end{enumerate}
\end{lemma}

\begin{proof}
The configuration at the prediction round $1$ is the one produced by Initialize
--- $S_1 = \{1\}$, $\mathcal{E}_1 = \emptyset$ and $Q_1^{(1)} = \{0\}$ --- and
the configuration at the prediction round $n + 1 \leq T$ is obtained from the
one at $n$ by the round $n$ of the epoch, whose time index is $n$ and which
does not restart, the restart that ends the epoch, if any, belonging to round
$T$. Within a round, the scale fragment is modified only by Preprocess, which
adjoins scales (line 3 of \algorithmref{alg:p3-preprocess}) and creates edges
(line 4), and by Post-process, which deletes scales together with their
incident edges (line 3 of \algorithmref{alg:p3-post}); the cascade of
\algorithmref{alg:p3-update} alters neither $S$ nor $\mathcal{E}$.

\emph{Claim 1.} Line 4 of \algorithmref{alg:p3-preprocess} creates the edge
$(m, r)$ only in the round whose time index is $m$ (line 1), so every edge of
$\mathcal{E}_n$ has source $m \leq n - 1$. Consequently no scale $m \geq n$ has
an outgoing edge, while every sink is greater than $n - 1$ by the horizon claim
of \lemmaref{lem:p3-inv}. Finally, a non-root scale is adjoined by line 3 of
\algorithmref{alg:p3-preprocess} as a multiple $n' r'$ of the time index $n'
\leq n - 1$ of its round, with $r' \leq r_{\max}$; so every scale of $S_n$ is
at most $\max(1, r_{\max}(n - 1)) \leq r_{\max} n$.

\emph{Claim 2.} We show by induction on $n$ that every non-root scale of $S_n$
is the target of an edge of $\mathcal{E}_n$. This holds for $n = 1$. In the
round $n < T$, Preprocess, which deletes nothing, adjoins a scale only together
with an edge into it (lines 3--4 of \algorithmref{alg:p3-preprocess}); the
cascade changes nothing; and Post-process either deletes nothing, or leaves its
loop when the set $\mathcal{D}$ of line 4 of \algorithmref{alg:p3-post} is
empty, whose first constituent is empty exactly when every surviving non-root
scale is the target of a surviving edge; the property thus passes from $S_n$ to
$S_{n+1}$. Moreover $S_n \neq \emptyset$, since line 5 of
\algorithmref{alg:p3-post} fires at no round before $T$. Starting from any
scale of $S_n$ and following incoming edges, which strictly decrease the scale,
one reaches a scale of $S_n$ with no incoming edge, which must be the root:
thus $1 \in S_n$, and the edges traversed form an entry path.

\emph{Claim 3.} The edge $(m, r)$ is created only in the round $m$, by line 4
of \algorithmref{alg:p3-preprocess}, and deleted only by line 3 of
\algorithmref{alg:p3-post}; since the round $m$ does not recur within the
epoch, a deleted edge is never re-created. Thus, if $(m, r) \in
\mathcal{E}_{n'}$, the edge was created in the round $m \leq n' - 1$ and
survived the post-processing of the rounds $m, \dots, n' - 1$, so that $(m, r)
\in \mathcal{E}_{n''}$ for every prediction round $n''$ with $m + 1 \leq n''
\leq n'$. The set of prediction rounds at which $(m, r)$ is present is
therefore empty or of the form $\{m + 1, \dots, n_e\}$, with $n_e$ its maximum.

\emph{Claim 4.} For a scale $m \in S$, the state set $Q_m$ is modified within a
round only by line 12 of \algorithmref{alg:p3-update}, which adds a state, and
by the deletion of $m$; line 3 of \algorithmref{alg:p3-preprocess} assigns
$Q_{n' r'}$ only for a scale $n' r' \notin S$. If $m \in S_{n'} \cap
S_{n'+1}$, then $m$ survives the post-processing of round $n'$, the only step
of that round at which a scale is deleted; hence $Q_m^{(n')} \subseteq
Q_m^{(n'+1)}$, and the first assertion follows by induction on $n'' - n$. For
the second, $Q_1^{(1)} = \{0\}$, while for $n \geq 2$ line 1 of
\algorithmref{alg:p3-post} at round $n - 1$ collects every scale with
$\abs{Q_m} > c_{\max} = c$, the loop deletes them all, and no state is added
afterwards.
\end{proof}

Fix a prediction round $n$. The walk of $\xi \in R^\omega$ from a scale $m \in
S_n$ is \emph{unblocked} (at round $n$) if $\xi \in Z(\gamma)$ for some $\gamma
\in \mathcal{C}_m^{(n)}$: the walk follows edges of the fragment until it
reaches a sink. Such a $\gamma$ is unique, since no continuation is a proper
prefix of another --- the sink at which a continuation ends has no outgoing
edge --- and we denote it by $\gamma_n(m, \xi)$. The set $\mathcal{U}_n(m)$ of
the paths whose walk from $m$ is unblocked is thus the union of the cylinder
sets $Z(\gamma)$, $\gamma \in \mathcal{C}_m^{(n)}$, which are finite in number
and pairwise disjoint; the same holds of every set of paths considered below,
so the masses involved are finite sums, and no measure theory is involved. The
unblocked paths of \sectionref{sec:p3-predict} are the paths whose walk from
the root is unblocked, and we write $\mathcal{U}_n \coloneqq
\mathcal{U}_n(1)$; for $\xi \in \mathcal{U}_n$, $J_n(\xi)$ denotes the length
of $\gamma_n(1, \xi)$, and $M_n(\xi) \coloneqq m_{J_n(\xi)}(\xi)$ is the sink
at which the walk of $\xi$ from the root ends.

Let $(m, q) \in \mathcal{A}_n$ be an active pair, so that $q = n \bmod m$, let
$\xi \in \mathcal{U}_n(m)$, and write $\gamma_n(m, \xi) = (m_0, r_0, \dots,
r_{L-1}, m_L)$. The \emph{descent} of $n$ along $\gamma_n(m, \xi)$ visits the
pairs $(m_j, n \bmod m_j)$, $j \leq L$. The pair at $j = 0$ is $(m, q)$, which
is active; and if the pair at $j < L$ is active and its edge is not a leaf edge
--- $(m_j, n \bmod m_j, r_j) \notin \mathcal{L}_n$ --- then, by the definition
of $\mathcal{L}_n$, the transition $\delta_{m_j}^{r_j}(n \bmod m_j, \floor{n /
m_j} \bmod r_j)$, a state of $Q_{m_{j+1}}^{(n)}$, is congruent to $n$ modulo
$m_{j+1}$, hence equals $n \bmod m_{j+1}$: the pair at $j + 1$ is active as
well. Let $I_n(m, \xi)$ be the least $j < L$ with $(m_j, n \bmod m_j, r_j) \in
\mathcal{L}_n$, if there is one; by induction, the pairs at the indices up to
$I_n(m, \xi)$ are active --- all the pairs of the descent, if there is no such
$j$. The \emph{vote} of $\xi$ from $(m, q)$ is
\[
  \vote_n(m, \xi) \coloneqq \begin{cases}
    \phi_n(m_j, n \bmod m_j, r_j) & \text{if } j = I_n(m, \xi) \text{ exists},
      \\
    \bot & \text{otherwise,}
  \end{cases}
\]
and in the second case we say that $\xi$ \emph{abstains} from $(m, q)$. Being
determined by $\gamma_n(m, \xi)$, the vote is constant on each cylinder set
$Z(\gamma)$, $\gamma \in \mathcal{C}_m^{(n)}$. At the root --- the pair $(1,
0)$ is active at every round, $Q_1 = \{0\}$ being installed by
\algorithmref{alg:p3-init} and never modified --- we write $I_n(\xi) \coloneqq
I_n(1, \xi)$ and $\vote_n(\xi) \coloneqq \vote_n(1, \xi)$, and simply say that
$\xi$ abstains. Thus $\vote_n(\xi) = \phi_n(m_I(\xi), n \bmod m_I(\xi), \xi_I)$
with $I \coloneqq I_n(\xi)$: the forecast of the first leaf edge met by the
descent of $n$ along $\xi$, which is an edge of $\xi$. The unwinding of
\algorithmref{alg:p3-predict} described in \sectionref{app:p3-predict} is the
following identity.

\begin{lemma}[Vote decomposition]
\label{lem:p3-vote-decomposition}
At every prediction round $n$ of the epoch, the ideal quantities defined by
lines 10--13 of \algorithmref{alg:p3-predict} satisfy $\theta_n(m) =
\mu(\mathcal{U}_n(m))$ for every $m \in S_n$, and
\[
  \pi_n(m, q)[a] = \mu\{\xi \in \mathcal{U}_n(m) \mid \vote_n(m, \xi) = a\}
\]
for every $(m, q) \in \mathcal{A}_n$ and $a \in \Sigma$. In particular,
$\theta_n(1) = \mu(\mathcal{U}_n)$ and $\pi_n(1, 0)[a] = \mu\{\xi \in
\mathcal{U}_n \mid \vote_n(\xi) = a\}$.
\end{lemma}

\begin{proof}
We argue by induction over $S_n$ in decreasing order, which is the order in
which lines 10--13 of \algorithmref{alg:p3-predict} compute $\theta_n$ and
$\pi_n$; at the scale $m$, both identities are thus available at every scale of
$S_n$ larger than $m$. If $m$ is a sink, $\mathcal{C}_m^{(n)}$ consists of the
length-zero continuation $(m)$ alone, with $Z((m)) = R^\omega$; so
$\mu(\mathcal{U}_n(m)) = 1 = \theta_n(m)$ by line 11, and every path abstains
from the active pair at $m$, if there is one, its descent meeting no edge ---
in agreement with line 13, both of whose sums are empty.

Let $m$ be a non-sink. A continuation from $m$ consists of an edge $(m, r) \in
\mathcal{E}_n$ followed by a continuation $\gamma' \in \mathcal{C}_{m
r}^{(n)}$; writing $(m, r) \gamma'$ for their concatenation, $Z((m, r)
\gamma')$ is the set of paths $\xi$ with $\xi_0 = r$ and $\xi' \coloneqq
(\xi_1, \xi_2, \dots) \in Z(\gamma')$, and $\mu(Z((m, r) \gamma')) = r^{-\eta}
\mu(Z(\gamma'))$ by \equationref{eq:p3-cylinder-mass}. Summing over the
continuations from $m$, grouped by their first edge,
\[
  \mu(\mathcal{U}_n(m)) = \sum_{r \in R : (m, r) \in \mathcal{E}_n} r^{-\eta}
  \sum_{\gamma' \in \mathcal{C}_{m r}^{(n)}} \mu(Z(\gamma')) = \sum_{r \in R :
  (m, r) \in \mathcal{E}_n} r^{-\eta} \mu(\mathcal{U}_n(m r)) = \theta_n(m),
\]
by the induction hypothesis at the scales $m r$ and line 11. Now let $(m, q)
\in \mathcal{A}_n$, and let $\xi \in Z((m, r) \gamma')$ as above. If $(m, q, r)
\in \mathcal{L}_n$, then $I_n(m, \xi) = 0$ and $\vote_n(m, \xi) = \phi_n(m, q,
r)$. Otherwise the pair $(m r, n \bmod m r)$ is active, the descent of $n$
along $(m, r) \gamma'$ continues from this pair with the descent along $\gamma'
= \gamma_n(m r, \xi')$, and $\vote_n(m, \xi) = \vote_n(m r, \xi')$. Writing
$V(m')[a] \coloneqq \mu\{\xi \in \mathcal{U}_n(m') \mid \vote_n(m', \xi) = a\}$
for the mass voting $a$ from the active pair of a scale $m'$, and summing the
masses $\mu(Z((m, r) \gamma')) = r^{-\eta} \mu(Z(\gamma'))$ over the
continuations on whose cylinder sets the vote is $a$, grouped by their first
edge, we obtain for every $a \in \Sigma$
\[
  V(m)[a] = \sum_{r : (m, q, r) \in \mathcal{L}_n,\ \phi_n(m, q, r) = a}
  r^{-\eta} \mu(\mathcal{U}_n(m r)) + \sum_{r : (m, r) \in \mathcal{E}_n,\ (m,
  q, r) \notin \mathcal{L}_n} r^{-\eta} V(m r)[a],
\]
which the induction hypothesis at the scales $m r$ --- $\mu(\mathcal{U}_n(m r))
= \theta_n(m r)$ and $V(m r)[a] = \pi_n(m r, n \bmod m r)[a]$ --- identifies
with the right-hand side of line 13, the argument $\delta_m^r(q, \floor{n / m}
\bmod r)$ of $\pi_n$ there being $n \bmod m r$ when $(m, q, r) \notin
\mathcal{L}_n$.
\end{proof}

The implementation of \algorithmref{alg:p3-predict} replaces its lines 10--14
by \algorithmref{alg:p3-predict-fp} (\sectionref{app:p3-predict}), so the
prediction $\hat{y}_n$ maximizes the computed vector $\hat{\pi}_n(1, 0)$, and
the ideal voting masses $\pi_n(1, 0)$ only up to a tolerance, which the
following lemma quantifies.

\begin{lemma}[Precision]
\label{lem:p3-precision}
At every prediction round $n$ of the epoch, the prediction $\hat{y}_n$
satisfies
\[
  \pi_n(1, 0)[\hat{y}_n] \geq \pi_n(1, 0)[a] - (r_{\max} n)^{-\eta} / 2 \quad
  \text{for every } a \in \Sigma.
\]
\end{lemma}

\begin{proof}
Let $u = 2^{-F}$, $\rho_r$, $\hat{\theta}_n$ and $\hat{\pi}_n$ be the unit, the
rounded weights and the approximations computed by
\algorithmref{alg:p3-predict-fp} at round $n$, so that $8 \abs{R} \log(2
r_{\max} n) u \leq (r_{\max} n)^{-\log r_{\max}}$ by
\equationref{eq:p3-precision}, while $\rho_r \in [0, r^{-\eta}]$ with
$r^{-\eta} - \rho_r \leq u$ for every $r \in R$ (claim 3 of
\lemmaref{lem:p3-fp-accuracy}); each product of lines 5 and 7 of
\algorithmref{alg:p3-predict-fp} is rounded down with an error smaller than
$u$, and the sums are exact. For $m \in S_n$, let $L(m)$ be the largest length
of a continuation in $\mathcal{C}_m^{(n)}$; the scales at least double along
the edges of a continuation and stay at most $r_{\max} n$ (claim 1 of
\lemmaref{lem:p3-fragment-dynamics}), so $L(m) \leq \log(r_{\max} n)$, and
$L(m) \geq L(m r) + 1$ for every edge $(m, r) \in \mathcal{E}_n$. We show, by
induction over $S_n$ in decreasing order as in
\lemmaref{lem:p3-vote-decomposition}, that
\[
  \abs{\hat{\theta}_n(m) - \theta_n(m)} \leq 2 \abs{R} L(m) u \quad \text{for
  every } m \in S_n,
\]
and that $\hat{\pi}_n(m, q)[a]$ deviates from $\pi_n(m, q)[a]$ by at most the
same bound, for the active pair $(m, q)$ at $m$, if any, and every $a \in
\Sigma$.

At a sink, line 5 of \algorithmref{alg:p3-predict-fp} and line 11 of
\algorithmref{alg:p3-predict} both assign $1$, and the sums of lines 7 and 13
are empty: all the deviations are zero. Let $m$ be a non-sink. The sums of
lines 5 and 7 of \algorithmref{alg:p3-predict-fp} run over the same index sets
as those of lines 11 and 13 of \algorithmref{alg:p3-predict}, matching each
term $\floor{\rho_r \cdot \hat{z}}_u$ of the former with the term $r^{-\eta} z$
of the latter at the same edge $(m, r) \in \mathcal{E}_n$, where the computed
value $\hat{z}$ approximates the ideal value $z \in [0, 1]$: for line 5, and at
the leaf edges of line 7, $\hat{z} = \hat{\theta}_n(m r)$ and $z = \theta_n(m
r) = \mu(\mathcal{U}_n(m r))$; at the remaining edges of line 7, $\hat{z} =
\hat{\pi}_n(m r, n \bmod m r)[a]$ and $z = \pi_n(m r, n \bmod m r)[a]$, a mass
as well, by \lemmaref{lem:p3-vote-decomposition}. By the induction hypothesis
at $m r$, $\abs{\hat{z} - z} \leq 2 \abs{R} L(m r) u$ in every case, so
\[
  \abs{\floor{\rho_r \cdot \hat{z}}_u - r^{-\eta} z} \leq u + \rho_r
  \abs{\hat{z} - z} + (r^{-\eta} - \rho_r) z \leq 2 u + r^{-\eta} \cdot 2
  \abs{R} L(m r) u,
\]
the three summands of the middle expression accounting for the rounding of the
product, the deviation of $\hat{z}$ and the deviation of the weight. The second
inequality bounds the latter two summands separately: $\rho_r \leq r^{-\eta}$
and the induction hypothesis give $\rho_r \abs{\hat{z} - z} \leq r^{-\eta}
\cdot 2 \abs{R} L(m r) u$, while $r^{-\eta} - \rho_r \leq u$ and $z \leq 1$
give $(r^{-\eta} - \rho_r) z \leq u$, which joins the rounding cost to make the
$2 u$. Each of $\hat{\theta}_n(m)$ and $\hat{\pi}_n(m, q)[a]$ accumulates at
most $\abs{R}$ such terms --- the two sums of line 7 range over disjoint sets
of edges --- so, summing the bound over them, with $\sum_{r \in R} r^{-\eta} =
1$ and $L(m r) \leq L(m) - 1$, the total deviation is at most $2 \abs{R} u + 2
\abs{R} (L(m) - 1) u = 2 \abs{R} L(m) u$, completing the induction.

At the root, whose active pair is $(1, 0)$, the deviation bound and $L(1) \leq
\log(r_{\max} n)$ give $\abs{\hat{\pi}_n(1, 0)[a] - \pi_n(1, 0)[a]} \leq 2
\abs{R} \log(r_{\max} n) u$ for every $a \in \Sigma$. Hence, for every $a \in
\Sigma$,
\[
  \begin{aligned}
    \pi_n(1, 0)[\hat{y}_n] &\geq \hat{\pi}_n(1, 0)[\hat{y}_n] - 2 \abs{R}
      \log(r_{\max} n) u \\
    &\geq \hat{\pi}_n(1, 0)[a] - 2 \abs{R} \log(r_{\max} n) u \\
    &\geq \pi_n(1, 0)[a] - 4 \abs{R} \log(r_{\max} n) u,
  \end{aligned}
\]
the first and third inequalities applying the root deviation bound at
$\hat{y}_n$ and at $a$ respectively, and the second holding because $\hat{y}_n$
maximizes $\hat{\pi}_n(1, 0)$ (line 8 of
\algorithmref{alg:p3-predict-fp}). Finally, $4 \abs{R} \log(r_{\max} n) u \leq
4 \abs{R} \log(2 r_{\max} n) u \leq (r_{\max} n)^{-\log r_{\max}} / 2 \leq
(r_{\max} n)^{-\eta} / 2$: the middle inequality is
\equationref{eq:p3-precision} halved, and the last holds since $\eta \leq \log
\abs{R} \leq \log r_{\max}$.
\end{proof}

\begin{lemma}[Abstention]
\label{lem:p3-abstention}
Let $\xi \in \mathcal{U}_n$. Then $\vote_n(\xi) = \bot$ if and only if
$M_n(\xi) = n$. In particular, $\xi$ abstains only when $n$ is a scale of
$\xi$.
\end{lemma}

\begin{proof}
Write $\gamma_n(1, \xi) = (m_0, r_0, \dots, r_{J-1}, m_J)$, so that $m_j =
m_j(\xi)$, $r_j = \xi_j$, and $m_J = M_n(\xi)$ is a sink, whence $M_n(\xi) \geq
n$ by claim 1 of \lemmaref{lem:p3-fragment-dynamics}. If $\xi$ abstains, every
pair of the descent of $n$ along $\gamma_n(1, \xi)$ is active; in particular $n
\bmod M_n(\xi) \in Q_{M_n(\xi)}^{(n)}$, and, every state being a position below
$n$, $n \bmod M_n(\xi) < n$, which forces $M_n(\xi) = n$. Conversely, suppose
that $M_n(\xi) = n$. For $j < J$, the scale $m_{j+1} = m_j r_j$ divides $n$, so
$n \bmod m_j = 0$ and the digit $\floor{n / m_j} \bmod r_j$ is $0$. We show by
induction on $j \leq J$ that the pair $(m_j, 0)$ is active and, if $j < J$,
that its edge $(m_j, 0, r_j)$ is not a leaf edge. The pair $(1, 0)$ is active.
If $(m_j, 0)$ is active and $j < J$, routing (claim 2 of
\lemmaref{lem:p3-inv}) gives $\delta_{m_j}^{r_j}(0, 0) \leq 0$, so the
transition is the state $0 = n \bmod m_{j+1}$ of $Q_{m_{j+1}}^{(n)}$: the edge
is not a leaf edge, and the pair $(m_{j+1}, 0)$ is active. Hence $I_n(\xi)$
does not exist, and $\xi$ abstains.
\end{proof}

\begin{lemma}[Monotonicity and mass of the unblocked set]
\label{lem:p3-unblocked-monotone}
For every prediction round $n$ of the epoch:
\begin{enumerate}
  \item if $n < T$ then $\mathcal{U}_{n+1} \subseteq \mathcal{U}_n$;
  \item $\mu(\mathcal{U}_n) \geq (r_{\max} n)^{-\eta}$.
\end{enumerate}
\end{lemma}

\begin{proof}
For claim 1, let $\xi \notin \mathcal{U}_n$; we show that $\xi \notin
\mathcal{U}_{n+1}$. Let $j$ be the largest index s.t.\ the edges $(m_i(\xi),
\xi_i)$, $i < j$, all belong to $\mathcal{E}_n$; it exists, as $S_n$ is finite
while the scales of $\xi$ are unbounded. Then $m_j(\xi) \in S_n$ --- it is the
root (claim 2 of \lemmaref{lem:p3-fragment-dynamics}) or the target of the edge
at $j - 1$ --- and it is not a sink, since otherwise the walk of $\xi$ from the
root would follow a continuation of $(S_n, \mathcal{E}_n)$ up to it. By claim 1
of \lemmaref{lem:p3-fragment-dynamics}, $m_j(\xi) \leq n - 1$, and by the
maximality of $j$, $(m_j(\xi), \xi_j) \notin \mathcal{E}_n$. By claim 3 of the
same lemma, the prediction rounds at which this edge is present form an
interval beginning at $m_j(\xi) + 1 \leq n$; since the edge is absent at round
$n$, it is absent at round $n + 1$. Now, if $\xi \in \mathcal{U}_{n+1}$, its
walk from the root reaches the sink $M_{n+1}(\xi) \geq n + 1$ of $(S_{n+1},
\mathcal{E}_{n+1})$ (claim 1 at round $n + 1$) through edges of
$\mathcal{E}_{n+1}$; as $m_j(\xi) \leq n - 1 < M_{n+1}(\xi)$, the edge
$(m_j(\xi), \xi_j)$ is one of them --- a contradiction.

For claim 2, the largest scale $M$ of the finite nonempty set $S_n$ is a sink
(an edge out of it would have a larger target in $S_n$), with $M \leq r_{\max}
n$ by claim 1 of \lemmaref{lem:p3-fragment-dynamics}, and it has an entry path
$\alpha$ by claim 2; then $\alpha \in \mathcal{C}_1^{(n)}$, so $Z(\alpha)
\subseteq \mathcal{U}_n$ and $\mu(\mathcal{U}_n) \geq \mu(Z(\alpha)) =
M^{-\eta} \geq (r_{\max} n)^{-\eta}$ by
\equationref{eq:p3-cylinder-mass}.
\end{proof}

\subsubsection{Edge Budgets via Branch Prediction}
\label{app:p3-stat-edge}

In this subsubsection, the forecasts of a single edge are coupled to an
instance of the branch prediction game of
\sectionref{app:branch-prediction}. The coupling goes through the refinement
hierarchy of the subsequences at the target scale of the edge, which is
determined by the data alone.

\begin{definition}
\label{def:p3-refinement-tree}
Let $M \geq 1$ and $\ell \in \mathbb{N}$. The \emph{live residues} at row $\ell$
are $\Lambda^M_\ell \coloneqq [M]$ if $\ell = 0$ and $\Lambda^M_\ell \coloneqq
\{u \in [M] \mid M(\ell - 1) + u < \abs{y}\}$ if $\ell \geq 1$: the residues
whose subsequences modulo $M$ have at least $\ell$ entries within $y$. For $u,
u' \in \Lambda^M_\ell$, write $u \eqv^M_\ell u'$ when $y[M j + u] = y[M j +
u']$ for every $j < \ell$; this is an equivalence relation on
$\Lambda^M_\ell$, and $[u]^M_\ell$ denotes the equivalence class of $u$. The
\emph{refinement tree} $\mathcal{T}_M$ has as vertices the pairs $(C, \ell)$
with $C$ an $\eqv^M_\ell$-equivalence class, as root $([M], 0)$, and an edge
$(C, \ell) \to (C', \ell + 1)$ whenever $C' \subseteq C$.
\end{definition}

Since $\Lambda^M_{\ell+1} \subseteq \Lambda^M_\ell$ and $\eqv^M_{\ell+1}$
refines $\eqv^M_\ell$ there, the children of $(C, \ell)$ partition $C \cap
\Lambda^M_{\ell+1}$, and every residue $u \in [M]$ determines a path
$([u]^M_0, 0), ([u]^M_1, 1), \dots$ in $\mathcal{T}_M$, of length equal to the
number of entries of its subsequence within $y$. For $M = k^{m+1}$,
$\mathcal{T}_M$ is the level-$m$ tree of \citep{Kosoy2026b} (Definition A.14
there). The following lemma is the counterpart of Lemma A.13 of
\citep{Kosoy2026b}.

\begin{lemma}[Leaf-edge semantics]
\label{lem:p3-leaf-semantics}
Let $n$ be a prediction round and $(m, q, r) \in \mathcal{L}_n$ a leaf edge,
and put $M \coloneqq m r$, $\ell \coloneqq \floor{n / M}$, $u \coloneqq n \bmod
M$ and $i \coloneqq \floor{n / m} \bmod r$, so that $u = q + i m$. Let $q'
\coloneqq \delta_m^r(q, i)$, and let $D$, $\nu$, $q^*$ and $\phi_n(m, q, r)$ be
the quantities computed for this leaf edge by lines 5--9 of
\algorithmref{alg:p3-predict}. Then:
\begin{enumerate}
  \item $q' < u$ and $u \eqv^M_\ell q'$.
  \item Every $\tilde{q} \in D$ satisfies $\tilde{q} < u$.
  \item $D = \{\tilde{q} \in Q_M^{(n)} \mid \tilde{q} \eqv^M_\ell u\}$.
  \item Distinct $\tilde{q}_1, \tilde{q}_2 \in D$ satisfy $\tilde{q}_1
    \neqv^M_{\ell+1} \tilde{q}_2$.
  \item For every $(\hat{q}, \hat{i}) \in Q_m^{(n)} \times [r]$ with $\hat{u}
    \coloneqq \hat{i} m + \hat{q} < u$: $\delta_m^r(\hat{q}, \hat{i}) \leq
    \hat{u}$ and $\hat{u} \eqv^M_{\ell+1} \delta_m^r(\hat{q}, \hat{i})$.
  \item $\nu(q') \geq 1$; consequently, $\nu(q^*) \geq 1$.
  \item $\phi_n(m, q, r) = y[M \ell + q^*]$; consequently, the forecast is
    wrong, $\phi_n(m, q, r) \neq y[n]$, if and only if $q^* \neqv^M_{\ell+1}
    u$.
\end{enumerate}
\end{lemma}

\begin{proof}
Throughout, \lemmaref{lem:p3-inv} is applied with its $n$ replaced by $n - 1$,
per the conventions of \sectionref{app:p3-stat}: on the configuration on which
\algorithmref{alg:p3-predict} runs, every run of Eval at a position $t \leq n -
1$ returns $y[t]$, both routing clauses hold, distinguishability (claim 3)
governs the divergence tables computed by line 1 of
\algorithmref{alg:p3-predict}, and the output maps are injective at the sinks
--- so \lemmaref{lem:p3-div-semantics} applies, and pinning
(\corollaryref{cor:p3-pinning}) is available with $N = n - 1$. Since $M \ell
\leq n$, every residue of $[M]$ belongs to $\Lambda^M_\ell$, and every residue
$w \leq u$ belongs to $\Lambda^M_{\ell+1}$, its row $\ell$ sitting at the
position $M \ell + w \leq n < \abs{y}$; this covers every membership in a live
residue set required below. Finally, the pair $(m, q)$ being active, $q = n
\bmod m \in Q_m^{(n)}$, and the identity $u = q + i m$ is
\equationref{eq:p3-entry-residue} at the edge $(m, r)$.

\emph{A transfer identity.} Let $(\hat{q}, \hat{i}) \in Q_m^{(n)} \times [r]$,
and put $\hat{u} \coloneqq \hat{i} m + \hat{q}$ and $\tilde{q} \coloneqq
\delta_m^r(\hat{q}, \hat{i})$. Then $\tilde{q} \leq \hat{u}$ --- the upper
bound of routing --- and
\[
  y[M j + \hat{u}] = y[M j + \tilde{q}] \quad \text{for every } j \in \mathbb{N}
  \text{ with } M j + \hat{u} \leq n - 1.
\]
Indeed, fix an entry path $\alpha$ to $m$ --- one exists by claim 2 of
\lemmaref{lem:p3-fragment-dynamics} --- and let $\alpha'$ be the entry path to
$M$ obtained by appending the edge $(m, r)$. At an edge $(a, s)$ of $\alpha$,
the scale $a s$ divides $m$, so the digits of $\hat{u} = \hat{q} + \hat{i} m$
and of $\hat{q}$ agree there: $\floor{\hat{u} / a} \equiv \floor{\hat{q} / a}
\pmod{s}$. The route recursion \equationref{eq:p3-route} along $\alpha'$
therefore follows the route of $\hat{q}$ along $\alpha$ --- whose endpoint is
$\rho_\alpha(\hat{q}) = \hat{q}$, by claim 2 of \lemmaref{lem:p3-route-map} ---
and finishes by consuming the digit $\floor{\hat{u} / m} = \hat{i}$ at the
appended edge: $\rho_{\alpha'}(\hat{u}) = \delta_m^r(\hat{q}, \hat{i}) =
\tilde{q}$. Now fix $j$ with $t \coloneqq M j + \hat{u} \leq n - 1$ and any
$\gamma \in \mathcal{C}_M^{(n)}$. The route factorization
\equationref{eq:p3-route-factor} evaluates the run of Eval$(t)$ along $\alpha'
\gamma$ as $\Eval_{\alpha' \gamma}(t) = \PEval_\gamma(\tilde{q}, j)$, while
data consistency at the observed position $t$ evaluates it as $y[t]$; and
pinning at the position $M j + \tilde{q} \leq t$ gives $\PEval_\gamma(\tilde{q},
j) = y[M j + \tilde{q}]$. Combining the three identities proves the claim.

\emph{Claim 1.} Apply the transfer identity to the pair $(q, i)$, for which
$\hat{u} = u$ and $\tilde{q} = q'$. It gives $q' \leq u$, and equality is
impossible: the leaf-edge condition of line 3 of
\algorithmref{alg:p3-predict} reads $q' \not\equiv n \pmod{M}$, while $u = n
\bmod M$. Hence $q' < u$. For $j < \ell$, the position $M j + u \leq M (\ell -
1) + u = n - M \leq n - 1$, so the identity gives $y[M j + u] = y[M j + q']$:
that is, $u \eqv^M_\ell q'$.

\emph{Claim 2.} Let $\tilde{q} \in D$; for $\tilde{q} = q'$ the claim is part
of claim 1, so suppose $\tilde{q} \neq q'$. Distinguishability at the distinct
states $q', \tilde{q} \in Q_M^{(n)}$ gives $M d_M(q', \tilde{q}) + \tilde{q}
\leq n - 1$, and $d_M(q', \tilde{q}) \geq \ell$ by the divergence condition of
line 6 of \algorithmref{alg:p3-predict}; hence
\[
  M \ell + \tilde{q} \leq M d_M(q', \tilde{q}) + \tilde{q} \leq n - 1 < n = M
  \ell + u,
\]
whence $\tilde{q} < u$.

\emph{Claim 3.} Let $\tilde{q} \in D$. By claim 1 and transitivity, the
required $\tilde{q} \eqv^M_\ell u$ reduces to $\tilde{q} \eqv^M_\ell q'$, which
is trivial for $\tilde{q} = q'$. Otherwise, $d_M(q', \tilde{q}) \geq \ell$ and
claim 1 of \lemmaref{lem:p3-div-semantics} give $\PEval_\gamma(q', j) =
\PEval_\gamma(\tilde{q}, j)$ for every $\gamma \in \mathcal{C}_M^{(n)}$ and $j
< \ell$; both states being smaller than $u$ (claims 1 and 2), the positions $M
j + q'$ and $M j + \tilde{q}$ are observed, and pinning turns the agreement
into $y[M j + q'] = y[M j + \tilde{q}]$ for every $j < \ell$: $\tilde{q}
\eqv^M_\ell q'$. Conversely, let $\tilde{q} \in Q_M^{(n)}$ satisfy $\tilde{q}
\eqv^M_\ell u$; we may assume $\tilde{q} \neq q'$, the case $\tilde{q} = q'$
being trivial. By claim 1 and transitivity, $\tilde{q} \eqv^M_\ell q'$.
Distinguishability provides the value $\ell^* \coloneqq d_M(q', \tilde{q})$
together with $M \ell^* + \max(q', \tilde{q}) \leq n - 1$ and $y[M \ell^* + q']
\neq y[M \ell^* + \tilde{q}]$; were $\ell^* < \ell$, this disagreement would
contradict $\tilde{q} \eqv^M_\ell q'$ at the row $\ell^*$. Hence $d_M(q',
\tilde{q}) \geq \ell$, i.e.\ $\tilde{q} \in D$.

\emph{Claim 4.} Let $\tilde{q}_1 \neq \tilde{q}_2 \in D$. Distinguishability
provides $\ell^* \coloneqq d_M(\tilde{q}_1, \tilde{q}_2)$ with $M \ell^* \leq n
- 1 < M (\ell + 1)$ --- so $\ell^* \leq \ell$ --- and $y[M \ell^* +
\tilde{q}_1] \neq y[M \ell^* + \tilde{q}_2]$: a disagreement at a row below
$\ell + 1$ between residues of $\Lambda^M_{\ell+1}$ (claim 2). Hence
$\tilde{q}_1 \neqv^M_{\ell+1} \tilde{q}_2$.

\emph{Claim 5.} The inequality $\delta_m^r(\hat{q}, \hat{i}) \leq \hat{u}$ is
the first part of the transfer identity. For $j \leq \ell$, the position $M j +
\hat{u} \leq M \ell + \hat{u} < M \ell + u = n$, so the identity applies at
every row $j \leq \ell$: $\hat{u} \eqv^M_{\ell+1} \delta_m^r(\hat{q},
\hat{i})$.

\emph{Claim 6.} By heredity (claim 2 of \lemmaref{lem:p3-inv}) applied to the
state $q' \in Q_M^{(n)}$, the residue $q' \bmod m$ belongs to $Q_m^{(n)}$ and
$\delta_m^r(q' \bmod m, \floor{q' / m}) = q'$. Since $q' < u = n \bmod M$
(claim 1), the pair $(q' \bmod m, \floor{q' / m})$ is counted by line 7 of
\algorithmref{alg:p3-predict}: $\nu(q') \geq 1$. And $d_M(q', q') = \infty$
puts $q' \in D$, so the maximization of line 8 gives $\nu(q^*) \geq \nu(q')
\geq 1$.

\emph{Claim 7.} The winner $q^*$ of line 8 belongs to $D$, so $q^* < u$ by
claim 2: the position $M \ell + q^* < M \ell + u = n$ is observed, and data
consistency evaluates the run of Eval performed by line 9 at $\phi_n(m, q, r) =
y[M \ell + q^*]$. The forecast is therefore wrong exactly when $y[M \ell + q^*]
\neq y[n] = y[M \ell + u]$. By claim 3, $q^* \eqv^M_\ell u$, and both residues
lie in $\Lambda^M_{\ell+1}$; the relation $q^* \eqv^M_{\ell+1} u$ is thus
equivalent to the single further equality $y[M \ell + q^*] = y[M \ell + u]$,
whose negation is the wrong forecast.
\end{proof}

Fix now an edge $e = (m, r)$ with a nonempty lifetime $\{m + 1, \dots, n_e\}$
(claim 3 of \lemmaref{lem:p3-fragment-dynamics}), put $M \coloneqq m r$, and
let $\hat{Q}_m \coloneqq Q_m^{(n_e)}$, which by claim 4 of
\lemmaref{lem:p3-fragment-dynamics} contains $Q_m^{(n)}$ for every round $n$ of
the lifetime. We call $n$ a \emph{forecasting round} of $e$ if $(m, n \bmod m,
r) \in \mathcal{L}_n$ --- such rounds lie in the lifetime --- and we say that
$e$ \emph{forecasts wrongly} at such a round when $\phi_n(m, n \bmod m, r) \neq
y[n]$.

\begin{definition}[Edge game]
\label{def:p3-edge-game}
The \emph{edge game} $G_e$ is the following instance of the branch prediction
game (\definitionref{def:branch-prediction}), whose tree is at all times a
subtree of $\mathcal{T}_M$ containing its root.
\begin{itemize}
  \item The elements are the residues $E_e \coloneqq \{\hat{i} m + \hat{q} \mid
    \hat{q} \in \hat{Q}_m, \hat{i} < r\} \subseteq [M]$.
  \item The element $\hat{u} \in E_e$ is introduced just before the round
    $t(\hat{u}) \coloneqq \min\{n \mid m + 1 \leq n \leq n_e, \hat{u} \bmod m
    \in Q_m^{(n)}\}$, at the vertex $([\hat{u}]^M_{\ell_0}, \ell_0)$ with
    $\ell_0 \coloneqq \abs{\{j \in \mathbb{N} \mid M j + \hat{u} <
    t(\hat{u})\}}$, nature attaching the vertices of the path from the root to
    it that are still missing.
  \item At every round $n$ of the lifetime with $\hat{u} \coloneqq n \bmod M
    \in E_e$ and $t(\hat{u}) \leq n$, the element $\hat{u}$ is activated, after
    the introductions of the round; nature moves it to $([\hat{u}]^M_{\ell +
    1}, \ell + 1)$, where $\ell \coloneqq \floor{n / M}$.
  \item The learner plays PB (\definitionref{def:pb}), breaking ties so as to
    agree with line 8 of \algorithmref{alg:p3-predict} under the correspondence
    of \lemmaref{lem:p3-edge-coupling}.
\end{itemize}
\end{definition}

For a vertex $v$ of $\mathcal{T}_M$ and a round $n$ of the lifetime, let $N_n(v)$
denote the number of elements $\hat{u}$ with $t(\hat{u}) \leq n$ located, just
before the activation of round $n$, at $v$ or below it in $\mathcal{T}_M$.

\begin{lemma}[Coupling]
\label{lem:p3-edge-coupling}
Let $n$ be a forecasting round of $e$, and let $\ell$, $u$, $D$, $\nu$ and
$q^*$ be as in \lemmaref{lem:p3-leaf-semantics} for the leaf edge $(m, n \bmod
m, r)$. Put $v_n \coloneqq ([u]^M_\ell, \ell)$ and $c^{\tilde{q}} \coloneqq
([\tilde{q}]^M_{\ell+1}, \ell + 1)$ for $\tilde{q} \in D \cup \{u\}$. Just
before the activation of round $n$ in $G_e$:
\begin{enumerate}
  \item \textbf{(Element vertex.)} Every element $\hat{u} \in E_e$ with
    $t(\hat{u}) \leq n$ is located at $([\hat{u}]^M_{D_n(\hat{u})},
    D_n(\hat{u}))$, where $D_n(\hat{u}) \coloneqq \abs{\{j \in \mathbb{N} \mid
    M j + \hat{u} < n\}}$ equals $\ell + 1$ if $\hat{u} < u$ and $\ell$
    otherwise. In particular, $u$ is an element with $t(u) \leq n$, located at
    $v_n$.
  \item \textbf{(Population.)} $\nu(\tilde{q}) = N_n(c^{\tilde{q}})$ for every
    $\tilde{q} \in D$.
  \item \textbf{(Children.)} The map $\tilde{q} \mapsto c^{\tilde{q}}$ is a
    bijection from $\{\tilde{q} \in D \mid \nu(\tilde{q}) \geq 1\}$ onto the
    set of children of $v_n$ in the game tree. In particular, $v_n$ is an
    internal vertex.
  \item \textbf{(Mistakes.)} The activation of round $n$ is an activation of
    $u$ at $v_n$, at which PB predicts $c^{q^*}$ and nature moves $u$ to $c^u$;
    PB makes a mistake at this activation if and only if $e$ forecasts wrongly
    at round $n$.
\end{enumerate}
\end{lemma}

\begin{proof}
The moments of the play of $G_e$ are ordered round by round, the introductions
of a round preceding its activation. We use the following consequence of
\lemmaref{lem:p3-fragment-dynamics}: the prediction rounds of the lifetime are
consecutive (claim 3), and $m$ belongs to the fragment at each of them, so the
state sets $Q_m^{(n')}$ grow monotonically along the lifetime (claim 4); hence,
for a round $n'$ of the lifetime and $\hat{u} \in E_e$,
\[
  t(\hat{u}) \leq n' \quad \text{if and only if} \quad \hat{u} \bmod m \in
  Q_m^{(n')}.
\]

\emph{Claim 1.} The case form of the depth is immediate: $M j + \hat{u} < n = M
\ell + u$ holds for $j \leq \ell$ when $\hat{u} < u$ and for $j < \ell$ when
$\hat{u} \geq u$. Extending the notation to every round $n'$ of the lifetime,
$D_{n'}(\hat{u}) \coloneqq \abs{\{j \in \mathbb{N} \mid M j + \hat{u} < n'\}}$,
we prove by induction along the play: at every moment, the game tree is a
subtree of $\mathcal{T}_M$ containing, with each of its vertices, the entire
path of $\mathcal{T}_M$ from the root to it --- so the parent of a game vertex
is its parent in $\mathcal{T}_M$, and the elements below a game vertex in
$\mathcal{T}_M$ are exactly the elements below it in the game tree --- and,
just before the activation of every round $n'$, each element $\hat{u}$ with
$t(\hat{u}) \leq n'$ is located at $([\hat{u}]^M_{D_{n'}(\hat{u})},
D_{n'}(\hat{u}))$. This vertex exists: the rows of $\hat{u}$ counted by
$D_{n'}(\hat{u})$ sit at positions smaller than $n' \leq T < \abs{y}$, so
$\hat{u} \in \Lambda^M_{D_{n'}(\hat{u})}$.

Initially the game tree is the root of $\mathcal{T}_M$, and there are no
elements. Consider an introduction, of an element $\hat{u}$ just before its
round $t(\hat{u})$. By the containment hypothesis, the vertices of the path of
$\mathcal{T}_M$ from the root to $([\hat{u}]^M_{\ell_0}, \ell_0)$ that are
already present form an initial segment of that path, so the missing ones form
a chain hanging below the deepest present vertex --- exactly what the
introduction rule of \definitionref{def:branch-prediction} attaches --- and the
element is placed at its end, at the depth $\ell_0 = D_{t(\hat{u})}(\hat{u})$,
as the location assertion prescribes. Consider next an activation, of the
element $\hat{u} = n' \bmod M$ at a round $n'$ with $t(\hat{u}) \leq n'$. The
rows of $\hat{u}$ before $n'$ are $0, \dots, \floor{n' / M} - 1$, so by the
location hypothesis the element sits at $([\hat{u}]^M_{\ell'}, \ell')$ with
$\ell' \coloneqq \floor{n' / M}$, and the vertex $([\hat{u}]^M_{\ell' + 1},
\ell' + 1)$ to which nature moves it --- it exists, its defining row sitting at
the position $M \ell' + \hat{u} = n' < \abs{y}$ --- is its child in
$\mathcal{T}_M$: whether that child is present or fresh, the move is one the
game permits, and the containment persists. Finally, the location assertion
propagates across the play: an element $\hat{u}$ moves only at its own
activations, which occur exactly at the rounds $n'$ with $n' \bmod M = \hat{u}$
and $t(\hat{u}) \leq n'$ and increase its depth by one, while $D_{n' +
1}(\hat{u}) - D_{n'}(\hat{u}) = \abs{\{j \in \mathbb{N} \mid M j + \hat{u} =
n'\}}$ equals $1$ at exactly these rounds and $0$ otherwise.

In particular, at the forecasting round $n$ the pair $(m, n \bmod m)$ is
active, so $u \bmod m = n \bmod m \in Q_m^{(n)}$: hence $u \in E_e$ with $t(u)
\leq n$, located at $([u]^M_{D_n(u)}, D_n(u)) = ([u]^M_\ell, \ell) = v_n$.

\emph{Claim 2.} Fix $\tilde{q} \in D$. The vertex $c^{\tilde{q}}$ exists:
$\tilde{q} < u$, by claim 2 of \lemmaref{lem:p3-leaf-semantics}, places its row
$\ell$ at the position $M \ell + \tilde{q} < n < \abs{y}$, so $\tilde{q} \in
\Lambda^M_{\ell+1}$. Consider
\[
  \mathcal{S} \coloneqq \{\hat{u} \in E_e \mid \hat{u} \bmod m \in Q_m^{(n)},\
  \hat{u} < u,\ \hat{u} \eqv^M_{\ell+1} \tilde{q}\};
\]
we show that both sides equal $\abs{\mathcal{S}}$.

For $\nu(\tilde{q})$: the map $(\hat{q}, \hat{i}) \mapsto \hat{i} m + \hat{q}$
is a bijection from $Q_m^{(n)} \times [r]$ onto $\{\hat{u} \in E_e \mid \hat{u}
\bmod m \in Q_m^{(n)}\}$, under which the condition $\hat{i} m + \hat{q} < n
\bmod M$ of line 7 of \algorithmref{alg:p3-predict} is $\hat{u} < u$; it
therefore suffices to show, for a pair $(\hat{q}, \hat{i})$ with $\hat{u}
\coloneqq \hat{i} m + \hat{q} < u$, that $\delta_m^r(\hat{q}, \hat{i}) =
\tilde{q}$ if and only if $\hat{u} \eqv^M_{\ell+1} \tilde{q}$. The forward
direction is claim 5 of \lemmaref{lem:p3-leaf-semantics}. Conversely, suppose
$\hat{u} \eqv^M_{\ell+1} \tilde{q}$, and put $q_0 \coloneqq \delta_m^r(\hat{q},
\hat{i}) \in Q_M^{(n)}$. Claim 5 of \lemmaref{lem:p3-leaf-semantics} gives
$\hat{u} \eqv^M_{\ell+1} q_0$, whence $q_0 \eqv^M_{\ell+1} \tilde{q}$; in
particular $q_0 \eqv^M_\ell \tilde{q} \eqv^M_\ell u$ --- the latter by claim 3
of \lemmaref{lem:p3-leaf-semantics} at $\tilde{q} \in D$ --- and the same claim
puts $q_0 \in D$. Distinct members of $D$ are inequivalent at the level $\ell +
1$ (claim 4 of \lemmaref{lem:p3-leaf-semantics}), so $q_0 = \tilde{q}$. Hence
$\nu(\tilde{q}) = \abs{\mathcal{S}}$.

For $N_n(c^{\tilde{q}})$: an element $\hat{u} \in E_e$ enters this count
exactly when $t(\hat{u}) \leq n$ --- equivalently, $\hat{u} \bmod m \in
Q_m^{(n)}$ --- and its vertex lies at $c^{\tilde{q}}$ or below it in
$\mathcal{T}_M$. By claim 1, the vertex is $([\hat{u}]^M_{D_n(\hat{u})},
D_n(\hat{u}))$ with $D_n(\hat{u}) \in \{\ell, \ell + 1\}$; since
$c^{\tilde{q}}$ has depth $\ell + 1$, lying at or below it means being it:
$D_n(\hat{u}) = \ell + 1$ --- that is, $\hat{u} < u$ --- and $\hat{u}
\eqv^M_{\ell+1} \tilde{q}$. These are the conditions defining $\mathcal{S}$, so
$N_n(c^{\tilde{q}}) = \abs{\mathcal{S}} = \nu(\tilde{q})$.

\emph{Claim 3.} For $\tilde{q} \in D$ with $\nu(\tilde{q}) \geq 1$, claim 2
places an element at $c^{\tilde{q}}$ just before the activation of round $n$,
so $c^{\tilde{q}}$ belongs to the game tree; its parent in $\mathcal{T}_M$ is
$([\tilde{q}]^M_\ell, \ell) = v_n$ --- claim 3 of
\lemmaref{lem:p3-leaf-semantics} gives $\tilde{q} \eqv^M_\ell u$ --- and the
parent of a game vertex is its parent in $\mathcal{T}_M$ (claim 1), so
$c^{\tilde{q}}$ is a child of $v_n$ in the game tree. The map is injective,
distinct members of $D$ being inequivalent at the level $\ell + 1$ (claim 4 of
\lemmaref{lem:p3-leaf-semantics}).

For surjectivity, let $c$ be a child of $v_n$ in the game tree. The play
attached $c$ at some earlier moment --- the game tree began as the bare root
--- and every attachment deposits an element at the attached vertex or below
it: an introduction places the introduced element at the end of the attached
chain, and an activation moves the activated element to the attached child.
Since an element only ever moves to a child of its current vertex, it never
leaves the subtree of a vertex; hence some element $\hat{u}$ --- introduced at
a round $t(\hat{u}) \leq n$ --- lies at $c$ or below it just before the
activation of round $n$. By claim 1, its vertex has depth $D_n(\hat{u}) \in
\{\ell, \ell + 1\}$, and lying at or below the depth-$(\ell + 1)$ vertex $c$
forces $D_n(\hat{u}) = \ell + 1$ --- that is, $\hat{u} < u$ --- with $c =
([\hat{u}]^M_{\ell+1}, \ell + 1)$; and $t(\hat{u}) \leq n$ gives $\hat{q}
\coloneqq \hat{u} \bmod m \in Q_m^{(n)}$. Put $\hat{i} \coloneqq \floor{\hat{u}
/ m}$ and $\tilde{q} \coloneqq \delta_m^r(\hat{q}, \hat{i})$. By claim 5 of
\lemmaref{lem:p3-leaf-semantics}, $\hat{u} \eqv^M_{\ell+1} \tilde{q}$; moreover
$\tilde{q} \in D$, by claim 3 of \lemmaref{lem:p3-leaf-semantics}: $\tilde{q}
\eqv^M_\ell \hat{u} \eqv^M_\ell u$, the latter because $c$ is a child of $v_n$
in $\mathcal{T}_M$, i.e.\ $[\hat{u}]^M_{\ell+1} \subseteq [u]^M_\ell$. Hence
$c^{\tilde{q}} = ([\hat{u}]^M_{\ell+1}, \ell + 1) = c$, and the pair $(\hat{q},
\hat{i})$ is counted by line 7 of \algorithmref{alg:p3-predict}:
$\nu(\tilde{q}) \geq 1$. Thus $c$ has a preimage.

In particular, the domain of the bijection is nonempty --- $q' \in D$ has
$\nu(q') \geq 1$ by claim 6 of \lemmaref{lem:p3-leaf-semantics} --- so $v_n$
has a child in the game tree: it is internal.

\emph{Claim 4.} By claim 1, $u$ is an element with $t(u) \leq n$, so round $n$
activates it, at its vertex $v_n$; and $v_n$ is internal (claim 3), so this is
an activation at which the learner predicts, and nature moves $u$ to
$([u]^M_{\ell+1}, \ell + 1) = c^u$. PB predicts the child of $v_n$ whose
subtree of the game tree holds the most elements --- by claim 1, the child $c$
maximizing $N_n(c)$. By claims 2 and 3, the populations of the children are
precisely the values $\nu(\tilde{q})$ over $\tilde{q} \in D$ with
$\nu(\tilde{q}) \geq 1$; the maximum $\nu(q^*)$ of $\nu$ over all of $D$ is at
least $1$ (claim 6 of \lemmaref{lem:p3-leaf-semantics}), so it is not affected
by discarding the candidates with $\nu = 0$: the maximal child population is
$N_n(c^{q^*}) = \nu(q^*)$, and the tie-breaking stipulated in
\definitionref{def:p3-edge-game} makes PB predict $c^{q^*}$. Finally, PB errs
at this activation exactly when $c^{q^*} \neq c^u$, i.e.\ --- both $q^*$ and
$u$ lying in $\Lambda^M_{\ell+1}$ --- when $q^* \neqv^M_{\ell+1} u$, which by
claim 7 of \lemmaref{lem:p3-leaf-semantics} is exactly a wrong forecast of $e$
at round $n$.
\end{proof}

Claim 4 matches every wrong forecast of $e$ with a mistake of PB in $G_e$. Note
that, unlike in the mistake accounting of TP (Lemma A.20 of
\citep{Kosoy2026b}), no activation at a leaf of the game tree needs separate
treatment: at a forecasting round, the child $c^{q'}$ of the routed target is
populated (claims 2 and 3, with claim 6 of
\lemmaref{lem:p3-leaf-semantics}), so the activated vertex is internal.

For $c \geq 1$, put
\begin{equation}
  \label{eq:p3-edge-budget}
  B(c) \coloneqq 3 r_{\max} c \log(r_{\max} c).
\end{equation}

\begin{lemma}[Edge budget]
\label{lem:p3-edge-budget}
For every edge $e = (m, r)$, the number $W_e$ of prediction rounds of the epoch
at which $e$ forecasts wrongly satisfies $W_e \leq B(c)$; more precisely, if
the lifetime of $e$ is nonempty then
\[
  W_e \leq 3 \abs{E_e} \floor{\log \abs{E_e}} \leq B(c).
\]
\end{lemma}

\begin{proof}
An edge whose lifetime is empty has no forecasting round, so $W_e = 0 \leq
B(c)$. Suppose then that the lifetime is nonempty, and consider the play of
$G_e$. Every $\hat{u} \in E_e$ is introduced during it --- the minimum defining
$t(\hat{u})$ is over a set containing $n_e$, since $\hat{u} \bmod m \in
\hat{Q}_m = Q_m^{(n_e)}$ --- so the final number of elements is $\abs{E_e}$. By
claim 4 of \lemmaref{lem:p3-edge-coupling}, every prediction round at which $e$
forecasts wrongly contributes a mistake of PB in $G_e$, distinct rounds
contributing distinct activations; hence $W_e \leq 3 \abs{E_e} \floor{\log
\abs{E_e}}$, by \propositionref{prop:branch-prediction}. Finally, the map
$(\hat{q}, \hat{i}) \mapsto \hat{i} m + \hat{q}$ puts $\hat{Q}_m \times [r]$ in
bijection with $E_e$, so $\abs{E_e} = r \abs{\hat{Q}_m} \leq r c \leq r_{\max}
c$, by the capacity bound of claim 4 of
\lemmaref{lem:p3-fragment-dynamics} at the round $n_e$; since also $\floor{\log
\abs{E_e}} \leq \log(r_{\max} c)$, we conclude $3 \abs{E_e} \floor{\log
\abs{E_e}} \leq 3 r_{\max} c \log(r_{\max} c) = B(c)$.
\end{proof}

\subsubsection{Path Budgets and the Mistake Bound}
\label{app:p3-stat-mistakes}

We say that $\xi \in R^\omega$ is \emph{charged} at the prediction round $n$ if
$\xi \in \mathcal{U}_n$ and $\vote_n(\xi) \neq y[n]$ --- that is, if $\xi$ is
unblocked and either votes wrongly or abstains. This is the notion of
\sectionref{app:budgeted-plurality} for the play with the experts $R^\omega$
under the prior $\mu$, the live sets $\mathcal{U}_n$, the votes $\vote_n$, the
outcomes $y[n]$, the predictions $\hat{y}_n$ and the tolerances $\epsilon_n
\coloneqq (r_{\max} n)^{-\eta} / 2$, $n = 1, \dots, T$. Put
\begin{equation}
  \label{eq:p3-path-budget}
  K \coloneqq K(c, T) \coloneqq (\floor{\log T} + 1)(B(c) + 1).
\end{equation}

\begin{lemma}[Path budget]
\label{lem:p3-path-budget}
Every $\xi \in R^\omega$ is charged at most $K$ times during the epoch.
\end{lemma}

\begin{proof}
Let $n$ be a prediction round at which $\xi$ is charged: $\xi \in
\mathcal{U}_n$ and $\vote_n(\xi) \neq y[n]$. We assign to $n$ an index $j \in
\mathbb{N}$ with $m_j(\xi) \leq T$. If $\xi$ abstains at $n$, we take $j
\coloneqq J_n(\xi)$, so that $m_j(\xi) = M_n(\xi) = n \leq T$ by
\lemmaref{lem:p3-abstention}. Otherwise $j \coloneqq I_n(\xi)$ exists: the
triple $(m_j(\xi), n \bmod m_j(\xi), \xi_j)$ belongs to $\mathcal{L}_n$, so $n$
is a forecasting round of the edge $(m_j(\xi), \xi_j)$, and $\vote_n(\xi) =
\phi_n(m_j(\xi), n \bmod m_j(\xi), \xi_j) \neq y[n]$: the edge forecasts
wrongly at $n$. A leaf edge lies over an edge of the fragment (line 3 of
\algorithmref{alg:p3-predict}), so $(m_j(\xi), \xi_j) \in \mathcal{E}_n$, and
$m_j(\xi) \leq n - 1 < T$ by claim 1 of
\lemmaref{lem:p3-fragment-dynamics}.

The arities being at least $2$, $m_j(\xi) \geq 2^j$, so every assigned index
satisfies $2^j \leq T$: the assigned indices lie among $0, \dots, \floor{\log
T}$, at most $\floor{\log T} + 1$ values. Fix an index $j$. At most one
prediction round equals $m_j(\xi)$, so at most one abstention is assigned to
$j$; and the wrong votes assigned to $j$ occur at distinct prediction rounds at
which the single edge $(m_j(\xi), \xi_j)$ forecasts wrongly, of which there are
at most $B(c)$ by \lemmaref{lem:p3-edge-budget}. Every index thus accounts for
at most $B(c) + 1$ charges, and $\xi$ is charged at most $(\floor{\log T} +
1)(B(c) + 1) = K$ times.
\end{proof}

\begin{lemma}[Epoch mistake bound]
\label{lem:p3-epoch-mistakes}
The number of mistakes at the prediction rounds $1, \dots, T$ of the epoch is
at most
\[
  8 K(c, T)(\eta \ln(r_{\max} T) + \ln 2).
\]
\end{lemma}

\begin{proof}
We first check that the data above indeed constitute a play in the sense of
\sectionref{app:budgeted-plurality}. The set $\mathcal{U}_n$ is the union of
the finitely many pairwise disjoint cylinder sets $Z(\gamma)$, $\gamma \in
\mathcal{C}_1^{(n)}$, on each of which $\vote_n$ is constant, so the live sets
and the sets $\{\xi \in \mathcal{U}_n \mid \vote_n(\xi) = a\}$, $a \in \Sigma$,
are measurable; and the prediction $\hat{y}_n$ maximizes the voting mass up to
the tolerance $\epsilon_n$: the voting mass of $a$ is $\pi_n(1, 0)[a]$, by
\lemmaref{lem:p3-vote-decomposition}, and $\hat{y}_n$ maximizes $\pi_n(1, 0)$
up to $(r_{\max} n)^{-\eta} / 2 = \epsilon_n$, by
\lemmaref{lem:p3-precision}. The hypotheses of
\lemmaref{lem:budgeted-plurality} hold with the budget $K = K(c, T) \geq 1$ of
\equationref{eq:p3-path-budget}: the live sets decrease and
$\mu(\mathcal{U}_T) \geq (r_{\max} T)^{-\eta} > 0$, by claims 1 and 2 of
\lemmaref{lem:p3-unblocked-monotone}; the tolerances satisfy $2 \epsilon_n =
(r_{\max} n)^{-\eta} \leq \mu(\mathcal{U}_n)$, by claim 2 of the same lemma;
and every expert is charged at most $K$ times, by
\lemmaref{lem:p3-path-budget}. The lemma therefore bounds the number of
mistakes at the rounds $1, \dots, T$ by
\[
  8 K \ln(2 / \mu(\mathcal{U}_T)) \leq 8 K (\eta \ln(r_{\max} T) + \ln 2).
\]
\end{proof}

It remains to sum over the epochs.

\begin{lemma}
\label{lem:p3-mistake-bound}
\[
  M_{P^3}(x) = O(r_{\max}(s + \sigma) \log(r_{\max}(s + \sigma)) \cdot \log N
  \cdot (1 + \eta \log(r_{\max} N))).
\]
\end{lemma}

\begin{proof}
By \sectionref{app:p3-post}, the restarts divide the processing of $x$ into
epochs; number them $i = 0, \dots, E$ in order, let $b_i$ be the global index
at which the $i$-th epoch begins, and let $T_i$ be the local index of its last
round, so that $b_0 = 0$, $b_{i+1} = b_i + T_i$ for $i < E$, and $b_E + T_E = N
- 1$; the $i$-th epoch runs with the capacity $2^i$. The prediction rounds of
the epochs --- the global rounds $b_i + 1, \dots, b_i + T_i$ --- together with
the round $0$, thus partition $[N]$. At round $0$, the empty state predicts, at
the cost of at most one mistake (\sectionref{app:p3-init}). If $T_i \geq 1$,
\lemmaref{lem:p3-epoch-mistakes} bounds the mistakes at the prediction rounds
of the $i$-th epoch by $8 K(2^i, T_i)(\eta \ln(r_{\max} T_i) + \ln 2)$; if $T_i
= 0$, the epoch makes no predictions. Since $T_i < N$, while $\eta \geq 0$ and
$K(c, T)$ is nondecreasing in $T$, each bound is at most $8 K(2^i, N)(\eta
\ln(r_{\max} N) + \ln 2)$, whence
\[
  M_{P^3}(x) \leq 1 + 8 (\eta \ln(r_{\max} N) + \ln 2)(\floor{\log N} + 1)
  \sum_{i = 0}^{E} (B(2^i) + 1).
\]
The capacity $2^E$ is in force after the complete update of the round at which
the last epoch begins, so $2^E < 4(s + \sigma)$ by
\lemmaref{lem:p3-state-bound}. Since $B$ is nondecreasing,
\[
  \sum_{i = 0}^{E} B(2^i) = 3 r_{\max} \sum_{i = 0}^{E} 2^i \log(r_{\max} 2^i)
  \leq 3 r_{\max} \log(r_{\max} 2^E) \cdot 2^{E + 1} = 2 B(2^E) \leq 2 B(4(s +
  \sigma)),
\]
where $B(4(s + \sigma)) = 12 r_{\max}(s + \sigma) \log(4 r_{\max}(s +
\sigma))$; as the number of epochs is $E + 1 \leq \log(8(s + \sigma))$, it
follows that $\sum_{i = 0}^{E} (B(2^i) + 1) = O(r_{\max}(s + \sigma)
\log(r_{\max}(s + \sigma)))$. Combined with $\floor{\log N} + 1 \leq 2 \log N$
and $\eta \ln(r_{\max} N) + \ln 2 = \ln 2 \cdot (\eta \log(r_{\max} N) + 1)$,
this yields the claim.
\end{proof}

Since $\eta \leq \log \abs{R} \leq \log r_{\max}$,
\lemmaref{lem:p3-mistake-bound} gives in particular
\[
  M_{P^3}(x) = O(r_{\max} \log r_{\max} \cdot (s + \sigma) \log(r_{\max}(s +
  \sigma)) \cdot \log N \cdot \log(r_{\max} N)),
\]
with an absolute implied constant: this is the mistake bound
\equationref{eq:p3-mistakes} of \theoremref{thm:p3}, in whose statement the
length $N$ is denoted $n$ and $r_{\max}$ is denoted $r$. For $\abs{R} = 1$,
where $\eta = 0$, the last factor of \lemmaref{lem:p3-mistake-bound} is $1$,
and the bound there is $O(r_{\max}(s + \sigma) \log(r_{\max}(s + \sigma)) \cdot
\log N)$.

\subsection{Compression Bound}
\label{app:p3-compression}

In this subsection we prove the compression bound
\equationref{eq:p3-state}. The notation of \sectionref{app:p3-state-bound}
remains in force --- $R$, $r_{\max} \coloneqq \max R$, $\sigma$, $x \in
\Sigma^N$ with $N \geq 2$, and $s \coloneqq \LZLC_R(x)$ --- and so do the
conventions of \sectionref{app:p3-stat} concerning epochs and their prediction
rounds. Recall from \sectionref{sec:p3-model} what the persistent state of
$P^3$ comprises: the time index and the capacity $c_{\max}$; the scales $m \in
S$, each carrying its state set $Q_m$ and its partial output function $\tau_m$,
defined on a subset of $Q_m$; and the edges $(m, r) \in \mathcal{E}$, each
carrying its transition table $\delta_m^r$. Nothing else persists: the
divergence tables and the quantities of \algorithmref{alg:p3-predict} are
recomputed at every call (\sectionref{sec:p3-model},
\sectionref{app:p3-predict}). Two of the counts needed to bound the size of all
this are already available: after every complete update, every retained scale
has fewer than $4(s + \sigma)$ states (\lemmaref{lem:p3-state-bound}), and
every scale is at most $r_{\max}$ times the time index (claim 1 of
\lemmaref{lem:p3-fragment-dynamics}). The count still missing is that of the
scales themselves.

\begin{lemma}[Fragment count]
\label{lem:p3-fragment-count}
At every prediction round $n$ of every epoch, $\abs{S_n} \leq (\log(2 r_{\max}
n))^{\abs{R}}$.
\end{lemma}

\begin{proof}
Every scale of $S_n$ belongs to $R^!$ (\sectionref{sec:p3-model}) and is at
most $r_{\max} n$ (claim 1 of \lemmaref{lem:p3-fragment-dynamics}); it is
therefore the product of a multiset of elements of $R$ in which every
multiplicity is at most $\log(r_{\max} n)$ --- each element of $R$ is at least
$2$, so its power dividing a product bounded by $r_{\max} n$ has exponent at
most $\log(r_{\max} n)$. Assigning to every scale one such multiset is
injective, the product of the multiset recovering the scale; and the multisets
in which each of the $\abs{R}$ elements of $R$ has multiplicity at most
$\floor{\log(r_{\max} n)}$ number $(1 + \floor{\log(r_{\max}
n)})^{\abs{R}} \leq (\log(2 r_{\max} n))^{\abs{R}}$.
\end{proof}

We next fix the binary encoding of the state, up to which
\sectionref{app:p3-update} identified Update and Predict with the protocol's
functions $\mathcal{U}$ and $\mathcal{P}$. Consider the configuration at a
prediction round $n$ of an epoch, and put $W \coloneqq \ceil{\log(r_{\max} n +
1)}$: every scale is at most $r_{\max} n$ (claim 1 of
\lemmaref{lem:p3-fragment-dynamics}) and every state is a residue smaller than
its scale, so $W$ bits accommodate any of them. Integers designated below as
\emph{self-delimited} are written in a prefix-free code of $O(\log(v + 2))$
bits for the value $v$ (say, Elias gamma applied to $v + 1$); every other field
has a fixed width, computable from $R$, $\Sigma$ and the fields already read,
so the state is recovered unambiguously from the concatenation of the following
blocks.
\begin{enumerate}
  \item The time index $n$ and the capacity $c_{\max}$, self-delimited.
  \item The number of scales, self-delimited, followed by the scales in
    increasing order, $W$ bits each.
  \item For every $m \in S$ in increasing order: $\abs{Q_m}$, self-delimited,
    followed by the residues of $Q_m$ in increasing order, $W$ bits each.
  \item For every $m \in S$ in increasing order: the domain of $\tau_m$, as one
    bit per state of $Q_m$ in increasing order, followed by the values of
    $\tau_m$ on it, in the same order, each as an index into a fixed
    enumeration of $\Sigma$, $\ceil{\log \max(\sigma, 2)}$ bits wide.
  \item For every $m \in S$ in increasing order: one bit per $r \in R$ in
    increasing order, flagging whether $(m, r) \in \mathcal{E}$; then, for
    every flagged $r$, the table $\delta_m^r$ in row-major order --- the rows
    indexed by $Q_m$ and the columns by $[r]$, both in increasing order ---
    each entry being the rank of the target state in the increasing order of
    $Q_{m r}$, $\ceil{\log \max(\abs{Q_{m r}}, 2)}$ bits wide. This width is
    indeed available to the decoder: $(m, r) \in \mathcal{E}$ requires $m r \in
    S$ (\sectionref{sec:p3-model}), so $\abs{Q_{m r}}$ was read in block 3.
\end{enumerate}
The empty state is encoded in the same format, with a zero scale count and, the
time index being immaterial there (\sectionref{app:p3-init}), $n = 0$: a string
of $O(1)$ bits.

\begin{lemma}
\label{lem:p3-compression}
\[
  S_{P^3}(x) = O((s + \sigma) \cdot (\abs{R} r_{\max} \log(s + \sigma) +
  \log(r_{\max} N)) \cdot (\log(2 r_{\max} N))^{\abs{R}}).
\]
\end{lemma}

\begin{proof}
The states assumed by the protocol on $x$ are $s_t$, $0 \leq t < N$, where $s_0
= s_{\mathrm{init}}$ is the empty state, of $O(1)$ bits, and $s_t$ for $t \geq
1$ is the persistent state after updating on $x[:t]$. By
\sectionref{app:p3-post}, the rounds at which Initialize runs divide the
processing of $x$ into epochs, the epoch beginning at such a round $b$
processing the suffix $x[b:N]$ as a fresh sequence; hence, for $t \geq 1$, the
state $s_t$ is the configuration, at its prediction round $n \coloneqq t - b$,
of the epoch beginning at the last round $b < t$ at which Initialize ran, and
$1 \leq n \leq t < N$.

Fix $t \geq 1$ and consider the blocks of $s_t$. By
\lemmaref{lem:p3-state-bound}, $c_{\max} < 4(s + \sigma)$, and $\abs{Q_m} <
4(s + \sigma)$ at every scale; hence every symbol field of block 4 and every
rank field of block 5 is $O(\log(s + \sigma))$ bits wide --- for the former,
$\ceil{\log \max(\sigma, 2)} \leq \ceil{\log(s + \sigma)}$, since $\sigma \leq
s + \sigma$ and $s + \sigma \geq 2$. Moreover, $\abs{S} \leq (\log(2 r_{\max}
n))^{\abs{R}} \leq (\log(2 r_{\max} N))^{\abs{R}}$ by
\lemmaref{lem:p3-fragment-count}, and $W = \ceil{\log(r_{\max} n + 1)} =
O(\log(r_{\max} N))$. Block 1 therefore occupies $O(\log N + \log(s + \sigma))$
bits; block 2, $O(\abs{S} W)$ bits; block 3, at most $\abs{S} \cdot (O(\log(s +
\sigma)) + 4(s + \sigma) W)$ bits; block 4, at most $\abs{S} \cdot 4(s +
\sigma) \cdot O(\log(s + \sigma))$ bits; and block 5, at most $\abs{S} \cdot
\abs{R} \cdot (1 + 4(s + \sigma) r_{\max} \cdot O(\log(s + \sigma)))$ bits, a
scale having at most $\abs{R}$ tables of at most $4(s + \sigma) r_{\max}$
entries each. Summing,
\[
  \begin{aligned}
    \abs{s_t} &= O(\log N) + \abs{S} \cdot O((s + \sigma) \cdot (\abs{R}
      r_{\max} \log(s + \sigma) + \log(r_{\max} N))) \\
    &= O((s + \sigma) \cdot (\abs{R} r_{\max} \log(s + \sigma) + \log(r_{\max}
      N)) \cdot (\log(2 r_{\max} N))^{\abs{R}}),
  \end{aligned}
\]
and maximizing over $t$ proves the claim.
\end{proof}

The implied constant is absolute: this is the compression bound
\equationref{eq:p3-state} of \theoremref{thm:p3}, in whose statement the length
$N$ is denoted $n$ and $r_{\max}$ is denoted $r$. In particular, since $s \leq
\max(N - 1, 1)$ (\propositionref{prop:lzlc-wcm}) while $\sigma$, $r_{\max}$ and
$\abs{R}$ do not depend on $x$, the size of the state is polynomial in the
length of the observed word, as computational feasibility
(\definitionref{def_comp_eff}) requires of it.

\end{document}